\documentclass[a4paper]{article}
\usepackage[left=2.5cm,top=2cm, right=2.5cm]{geometry} 
\usepackage[dvipsnames]{xcolor}
\definecolor{darkblue}{rgb}{0.0, 0.0, 0.55}

\def\mb{{\mathcal B}}
\def\mc{{\mathcal C}}

\def\e{\mathbb{E}}

\def\im{\imath}

\def\sm{\mathsf{M}}

\def\sj{\mathsf{J}}
\def\enc{\mathsf{Enc}}
\def\dec{\mathsf{Dec}}

\usepackage{amsmath, amsthm, amssymb,subcaption}
\usepackage{caption}

\usepackage{graphicx}
\usepackage{dsfont}
\usepackage{tikz}
\usetikzlibrary{decorations.markings}
\usepackage[colorlinks,citecolor=blue]{hyperref}

\newtheorem{thm}{Theorem}
\newtheorem{lem}{Lemma}
\newtheorem{prop}{Proposition}
\newtheorem{cor}{Corollary}
\newtheorem{definition}{Definition}
\theoremstyle{remark}	\newtheorem{remark}{Remark}

\newcommand\blfootnote[1]{%
  \begingroup
  \renewcommand\thefootnote{}\footnote{#1}%
  \addtocounter{footnote}{-1}%
  \endgroup
}

\title{\huge Information Theoretic Cutting of a Cake\blfootnote{This paper has been presented in part at Information Theory Workshop (ITW), 2012.}}

\newcommand{\gainwrt}[4]{\mathcal{G}_{#1}\left(#2, #3 \| #4 \right)}

\newcommand{\dset}{\mathcal{D}}
\newcommand{\vset}{\mathcal{V}}

\newcommand{\gsel}{G^\text{sel}}
\newcommand{\dsel}{\Delta^\text{sel}}

\newcommand{\va}{v_A}

\newcommand{\vb}{v_B}

\newcommand{\vaset}{\mathcal{V}_A}
\newcommand{\vbset}{\mathcal{V}_B}

\newcommand{\Sa}{S_A}
\newcommand{\Sb}{S_B}
\newcommand{\Sab}{S}
\newcommand{\Sae}{S_A^\epsilon}
\newcommand{\Sbe}{S_B^\epsilon}
\newcommand{\Se}{S^\epsilon}

\newcommand{\risky}{\mathsf{R}}
\newcommand{\nonrisky}{\mathsf{NR}}
\newcommand{\Bleft}{\mathsf{L}}
\newcommand{\Bright}{\mathsf{R}}
\newcommand{\selfish}{\mathsf{S}}
\newcommand{\nonselfish}{\mathsf{NS}}

\newcommand{\ga}[1]{\mathcal{G}_A\left( #1 \right )}
\newcommand{\gawrt}[2]{\mathcal{G}_A\left( #1 \Vert #2 \right )}
\newcommand{\gb}[1]{\mathcal{G}_B\left( #1 \right )}
\newcommand{\gbwrt}[2]{\mathcal{G}_B\left( #1 \Vert #2 \right )}
\newcommand{\gx}[1]{\mathcal{G}_X\left( #1 \right )}
\newcommand{\gxwrt}[2]{\mathcal{G}_X\left( #1 \right \Vert #2)}

\newcommand{\Gamax}{G_A^\text{max}}
\newcommand{\Gamin}{G_A^\text{min}}
\newcommand{\Gbmax}{G_B^\text{max}}
\newcommand{\Gbmin}{G_B^\text{min}}

\newcommand{\gamax}{g_A^\text{max}}
\newcommand{\gamin}{g_A^\text{min}}
\newcommand{\gbmax}{g_B^\text{max}}
\newcommand{\gbmin}{g_B^\text{min}}

\newcommand{\is}{\theta}

\newcommand{\ta}{\tilde{a}}

\newcommand{\one}[1]{\mathds{1}_{[#1]}}
\newcommand{\pp}[1]{\left ( #1 \right )^+}
\newcommand{\np}[1]{\left ( #1 \right )^-}
\newcommand{\ev}[1]{\mathbb{E} \left [ #1 \right ] }
\newcommand{\evwrt}[2]{\mathbb{E}_{#1} \left [ #2 \right ] }

\newcommand{\cream}{%
\begin{tikzpicture}[x=8pt,y=8pt]
 \draw (0,0) circle (0.5);
\end{tikzpicture}
}

\newcommand{\chocolate}{%
\begin{tikzpicture}[x=8pt,y=8pt]
 \fill (0,0) circle (0.5);
\end{tikzpicture}
}

\newcommand{\creamy}{
\begin{tikzpicture}[x=8pt,y=8pt]
 \draw (0,0) circle (0.5);
 \fill (0.1710,-0.4698) arc (-70:70:0.5);
\end{tikzpicture}
}

\newcommand{\chocolaty}{
\begin{tikzpicture}[x=8pt,y=8pt]
 \draw (0,0) circle (0.5);
 \fill (0.1710,0.4698) arc (70:290:0.5);
\end{tikzpicture}
}

\newcommand{\half}{%
\begin{tikzpicture}[x=8pt,y=8pt]
 \draw (0,0) circle (0.5);
 \fill (0,0.5) arc (90:270:0.5);
 \draw (-0.7,0) -- (0.7,0);
\end{tikzpicture}
}

\newcommand{\full}{%
\begin{tikzpicture}[x=8pt,y=8pt]
 \draw (0,0) circle (0.5);
 \fill (-0.5,0) arc (180:360:0.5);
 \draw (-0.7,0) -- (0.7,0);
\end{tikzpicture}
}

\newcommand{\divcream}{
\begin{tikzpicture}[x=8pt,y=8pt]
 \draw (0,0) circle (0.5);
 \fill (-0.5,0) arc (180:360:0.5);
 \draw[color=blue] (-0.9,0) -- (0,0) -- (0.7,0.4677);
\end{tikzpicture}
}

\newcommand{\divchoc}{
\begin{tikzpicture}[x=8pt,y=8pt]
 \draw (0,0) circle (0.5);
 \fill (-0.5,0) arc (180:360:0.5);
 \draw[color=blue] (-0.9,0) -- (0,0) -- (0.7,-0.4677);
\end{tikzpicture}
}

\newcommand{\aw}[1]{\text{AW}\left (#1 \right)}
\newcommand{\awwrt}[2]{\text{AW}_{#1}\left (#2 \right)}
\newcommand{\daw}[3]{\Psi\left(#1, #2 \Vert #3 \right )}
\newcommand{\dawhead}{\Psi}
\newcommand{\dawmax}[2]{\Psi^\ast\left(#1\Vert #2\right)}
\newcommand{\dawmaxhead}{\Psi^\ast}

\newcommand{\dga}[2]{\Delta_{#1}\left ( #2 \right )}
\newcommand{\dgamax}[2]{\Delta_{#1}^\ast \left ( #2 \right )}
\newcommand{\dgamaxhead}{\Delta^\ast}
\newcommand{\dgahead}{\Delta}

\newcommand{\bmin}{b_\text{min}}
\newcommand{\bmax}{b_\text{max}}
\newcommand{\tD}{\tilde{\Delta}}
\newcommand{\hb}{\hat{b}}
\newcommand{\hs}{\hat{s}}
\newcommand{\hr}{\hat{\rho}}

\newcommand{\av}{\textbf{a}}
\newcommand{\bv}{\textbf{b}}

\newcommand{\vv}{\textbf{v}}
\newcommand{\Vv}{\textbf{V}}
\newcommand{\tav}{\tilde{\textbf{a}}}

\newcommand{\reals}{\mathbb{R}}
\DeclareMathOperator*{\supremum}{sup}

\date{}
\author{Payam~Delgosha$^*$~and~Amin~Gohari$^\dagger$
\\\small$*$ Department of Electrical Engineering and Computer Sciences, University of California, Berkeley
\\\small pdelgosha@eecs.berkeley.edu
\\\small$\dagger$ Department of Electrical Engineering, Sharif University of Technology
\\\small aminzadeh@sharif.edu
}
\begin{document}
 \maketitle
 \allowdisplaybreaks

\begin{abstract}
Cutting a cake is a metaphor for the problem of dividing a resource (cake) among
several agents. The problem becomes non-trivial when the agents have different
valuations for different parts of the cake (i.e. one agent may like chocolate
while the other may like cream). A fair division of the cake is one that takes
into account the individual valuations of agents and partitions the cake based
on some fairness criterion. Fair division may be accomplished in a distributed
or centralized way. Due to its natural and practical appeal, it has been a
subject of study in economics. To best of
our knowledge the role of \emph{partial} information in fair division has not been
studied so far from an information theoretic perspective. 
Given the diversity of problems in fair division, we consider certain specific (yet important) problems that capture different aspects of information exchange in a fair division setting. From the class of distributed algorithms, we consider the classical Divide and Choose (DC) problem between two parties. Here, we study  the effect of partial spying and voluntarily sharing of information in both one-shot and asymptotic scenarios. Furthermore, we consider  implicit information transmission through actions for the repeated version of the problem. While identifying  subgame perfect Nash equilibrium in repeated games with incomplete information on both sides is very difficult in general, for the special case of division of two items, we find a more stringent trembling hand perfect 
equilibrium.
Next, from the class of centralized algorithms, we consider the Adjusted Winner (AW) algorithm between two players Alice and Bob.  Brams and Taylor showed that if Alice can fully spy on Bob, she can trick the algorithm. We consider the same setup when partial spying is allowed, and study the growth rate of Alice's utility per spying bit. Via a transformation from AW to DC, it is shown that the problem reduces to the one studied earlier for DC. However, if Alice is forced to only spy certain simple structured functions of Bob's valuation, an upper bound on the growth rate of utility per spying bit is derived. This bound is shown to be tight in some cases. 
We also consider a centralized algorithm for
maximizing the overall welfare of the agents under the Nash collective utility
function (CUF). This corresponds to a clustering problem.
By 
observing 
a  link between this problem and the portfolio
selection problem in stock markets, we provide an upper bound on the increase of
the Nash CUF for a clustering refinement.
\end{abstract} 
\section{Introduction}

In many applications a number of parties are interested in possessing a
limited resource, e.g. a set of goods or metaphorically a cake.\footnote{For instance, in networking and wireless communications, optimal power allocation is a challenge \emph{e.g.} see \cite{BLSD09, GM00, YGC02}.} Each of the
parties has his own valuation of different parts of the cake, and each has full,
partial or no information about the valuation of the other parties. Finding a
way to divide a cake fairly has attracted the attention of economists and
mathematicians for a long time. Although
information theory is developed for studying communication
systems \cite{Bandwagon}, it gives us tools to quantify information in other
fields (such as fair-division) where \emph{partial} information is of
relevance. For instance consider a 
division game between Alice and Bob where Bob is unwilling
to let Alice spy on his information (as that information can be advantageous for Alice, since it would reduce her uncertainty about Bob's actions). One of the results of this paper is to show that there are cases where if Bob learns that the
spying rate of Alice exceeds a certain threshold, he will become
willing to voluntarily share even more information with Alice. Thus, identifying
when this happens can be of importance to Alice and Bob
in designing the rules of the game.

\subsection{Fair Division}
In this paper we assume that the reader is familiar with network information theory but not necessarily with fair division. Before trying to find a fair division, one must define
the term ``fairness''.
Several
criteria of fairness have been introduced to judge the goodness of a division
where none of which subsumes the others \cite{brams:cake:1996}. Here we will
give a brief introduction to
four of them. Assume that $k$ denotes the number of parties.
\begin{itemize}
  \item A division is said to be
\emph{proportional} if each party receives at least $1/k$ of the entire
cake
w.r.t.\ his own valuation.
  \item A division is said to be
\emph{equitable} if the piece of the cake each party obtains w.r.t.\
his own valuation is exactly equal to what the other parties receive (w.r.t.\
their own valuation).
  \item A division is said to be \emph{envy-free} if no party believes
that, w.r.t.\ his own valuation, the piece another party has received is
more valuable than his own.
\item A division is said to be \emph{efficient} or \emph{Pareto optimal} if
it is not possible to find another division that increases the gain of every
individual.
\end{itemize}

In the literature of fair division, there are two major assumptions regarding
the set of goods to be divided: the category of \emph{divisible} goods where
each good or item could be divided among parties, and the category of
\emph{indivisible} goods where
each item should
wholly be given to one party (e.g. a car or a laptop) \cite{brams2000win}. 
Analyzing division of divisible goods is generally easier than that of 
indivisible goods.
In the most generic scenario some of the items may be divisible, some
indivisible and some partially divisible. We take care of this generic scenario
by considering a set $\mathcal{D}$ of ``admissible" divisions of the resource.
Theoretically the set $\mathcal{D}$ is of size infinity if we have a divisible
item in the resource (since we can cut that item in any proportion). Practically
speaking, even divisible items can be cut up to a certain precision. Therefore
for simplicity we assume that the set $\mathcal{D}$ is finite (unless stated
otherwise).
Lastly, the preferences or
valuations of parties could be \emph{ordinal} or \emph{cardinal}. Here we assume
that valuations are cardinal,
i.e.\ could be modeled by non-negative real numbers.



Any algorithm providing a fair division may satisfy one or some of the fairness
conditions introduced above (see for instance \cite{Brams:Fishburn:2000,  Edelman:Fishburn:2001} for conflicts  in  fairness  criteria  and  tradeoffs).
 From another point of view, fair division may be
accomplished in a distributed or centralized way. In a
distributed algorithm the individuals should divide the cake amongst themselves,
while in a centralized one, an external referee divides the cake for them. In 
order to
address these two categories, we have chosen two prominent algorithms from the
field, \emph{Divide-and-Choose} (DC) from the category of distributed algorithms 
and
\emph{Adjusted Winner} (AW) from the category of centralized algorithms. In our discussion of centralized algorithms, we also consider 
the problem of optimizing \emph{social welfare},  another topic in fair division.


%
%
The ``\emph{I cut, you choose}" or \emph{divide-and-choose} (DC) procedure is a
well-known and
ancient
algorithm for dividing a resource among two parties \cite{brams:cake:1996}.
The story of dividing a land between Abram and Lot in the Hebrew Bible refers
to this method.
In this procedure, the first party (Alice) cuts the cake into two parts and the
second party (Bob) chooses one of the pieces, leaving the other piece for the 
first
party.
%
Note that Bob has an advantage over Alice for he can choose the best piece and
can possibly get even more than half of the total value he assigns to the cake.
In other words, when Alice does not know anything about Bob's valuation, she
should divide the cake into two parts which are equal with respect to her
valuation, so that despite Bob's choice, she gains at least half of the cake.
However Bob achieves more than half of the cake since he is free to choose. 
Since
 each party can obtain at least half the cake, this method is
proportional but not equitable \cite{jones2002equitable}.

The ``\emph{Adjusted Winner}" (AW) algorithm was originally proposed by Brams 
and
Taylor \cite{brams:cake:1996}. Since then, it has been applied to disputes  ranging  from  interpersonal  to  international  \cite{brams:cake:1996,  brams:cake:1999b}. Assume that two parties, say Alice and Bob, want
to divide a set of $m$ divisible goods. Alice's valuation vector is denoted by a
vector $\textbf{a}=(a_1,\dots,a_m)$ of $m$ non-negative real numbers that add up 
to one.
Similarly, Bob's valuation vector is denoted by $\textbf{b}=(b_1,\dots,b_m)$. We 
assume
that the value of a piece of cake for each player is the sum of the portion of
each item present in that piece times the value that player assigns to that
item. In the Adjusted Winner algorithm Alice and Bob announce their valuations
vectors to an external referee. The referee solves a set of equations to come up
with a division of the items which is proportional, equitable, envy-free and
efficient. For extensions of AW to three players or more, see \cite{Young:1994, brams:cake:1996, Moulin:2003}

Related to centralized algorithms in fair division is the problem of optimizing the \emph{social welfare} by proper division of resources across a society. 
In the literature of economics, 
a social welfare is a function
that collects the utilities or gains of each individual in the society and
returns a real value which reflects the overall welfare in the society.
Philosophical utilitarianism suggests a division strategy that maximizes the
overall happiness (or sum of the gains of the individuals).
 Thus, the rules of
 division here are not decided by selfish players but by an external judge (or
 by
 players who follow Rawls's veil of ignorance \cite{rawls1999theory}).
Another measure for social
welfare that cares not only about the overall happiness but also about its
uniform distribution over the individuals (an egalitarian philosophy) is the 
Nash
collective utility function (CUF). Nash CUF is defined to be product of the
gains of the individuals \cite{moulin2004fair}.

We
refer the reader to \cite{brams:cake:1996,brams2000win,robertson1998cake} for further reading on fair
division.

\subsection{Motivation: utility per information bit}

To the best of our knowledge, the problem of fair division has only been
analyzed 
when individuals do not know the valuation of others, or when they have
complete information about the valuations; it is not analyzed
in the case of \emph{partial information}. To motivate this study, let us begin with
the DC algorithm.
As we saw previously, the second party, Bob, has advantage in choosing the piece
he likes more.
 One way to make the algorithm more
fair is to provide Alice with partial information about Bob's valuation. For
instance if there is an item that Alice likes a lot but Bob is indifferent to it
- and Alice knows this - she can put all of it in the piece that she predicts
Bob will not choose. To quantize the
role of information in such scenarios, we need to find the gain of individuals
as a function of the rate of communication between them. This leads to
characterizing an achievable rate-gain region. The tradeoff between the
disadvantage of being the cutter and the advantage of having information is most
notably present in a seller-consumer scenario. A seller offers a good for a
price, and the consumer can choose to buy the item or keep his money. This
problem resembles the DC algorithm and our formulation (defined
later) is general enough to cover it. Setting a price by the seller resembles
cutting a cake, and the consumer's choice of buying the item is like picking one
of the two pieces ``item" or ``his money". As discussed above this transaction
scheme is naturally biased towards the chooser, i.e. the consumer. But the
seller has generally more information about the consumer's needs than the
consumer has about the true price of the item. The role of information in the
bargaining dynamic is also colorful: the consumer hides how much he really needs
the item while the seller hides how much the item is really worth.

\subsection{Our Contributions}
Given the diversity of problems in fair division, we consider certain specific (yet important) problems that capture different aspects of information exchange in a fair division setting. From the class of distributed algorithms, we consider the divide and choose problem between two parties. From the class of centralized algorithms, we consider the adjusted winner and the problem of optimizing the social welfare. 

\textbf{Divide and choose:}
In the literature of fair division, only one-shot of the divide and choose 
problem is considered. However we consider both the one-shot and repeated 
divide and choose problems. In the repeated version, the divide and choose procedure is repeated $n$ times on $n$ 
identical cakes. Then, the average gain of each party is considered during 
$n$ games. As was discussed before, parties are unaware of each other's 
valuations. However, they can achieve some amount of side information about 
the valuation of the other party which can help them achieve a better result. 
This flow of information and its impact on the result of the game is our main 
interest.

We take three rather different models for this problem in the subsequent 
sections. In our first model, we 
assume that Alice can {spy} on Bob prior to the division procedure. We call this the \emph{spying model}.  We consider both the one-shot and asymptotic cases. In the asymptotic case, we assume that parties' valuations are generated 
i.i.d.\ from a given probability distribution. The second model
is a game and information theoretic one, where Bob  chooses
 to communicate only when he expects the information to increase his gain in the divide and choose procedure. Thus, Bob can \emph{share} information of his choice (at limited rate) to Alice prior to the division procedure. We call this the  \emph{sharing model}. In the third model, which is more game theoretic, we assume that instead 
of sharing or spying prior to the divide 
and choose procedure, information is flowed implicitly during the game. More 
precisely, rather than gaining explicit information before cutting the cake, 
parties receive information about each other's valuations through their
actions. The main difference between this model and the previous ones is that 
the valuations are randomly generated from a given probability distribution 
once and for all, and are fixed during the $n$ stage game. 
Therefore, unlike the first model where different stages of the game were 
independent, parties can gain information by looking at the history of the 
game.
As a result, for instance Bob 
might have the tendency to hide his valuation by choosing the less valuable 
piece at one stage in order to deceive Alice and gain more during the next 
stages. We will explicitly identify a Trembling hand Nash equilibrium in the resulting repeated game.
Trembling hand Nash equilibrium is one of the most strict forms of equilibrium; for example, it implies the subgame perfect Nash equilibrium. 
While repeated games are widely studied in 
game theory (e.g. see \cite{aumann}), explicit identification of a subgame perfect Nash equilibrium in repeated games with incomplete information on both sides is known to be very difficult in general. We show that Bob  playing selfishly is an equilibrium and there is no incentive for him not to use his information initially, so that Alice does not learn about his valuation. Implicit communication through actions has also been studied in the information theory and control literature (e.g. see \cite{Cuff, Pulkit}). But these works do not consider game equilibriums.

\textbf{Adjusted Winner:} Brams and Taylor showed that in the case of having two goods,
i.e.\ $m=2$, when one of the parties, say Alice,
knows Bob's valuation while Bob is unaware
of this, Alice can announce an untrue valuation in order to trick the procedure
and gain more than what she otherwise would. We consider the same setup for the general $m$ goods, but with the further refinement of assuming that Alice is only partially spying on Bob. It is shown that this problem reduces to the corresponding spying problem in the divide and choose problem. Next, if Alice is only allowed to spy binary questions of the form ``Is Bob's valuation of a certain good less than a threshold?," we prove an upper bound on the growth rate of Alice's utility per the number of bits she has spied for the case of $m=2$. This bound is shown to be tight for a range of parameters.

\textbf{Maximizing social welfare:}
The last part of this paper provides a connection information theory and optimizing the
social welfare
under the Nash
CUF in fair division in large societies, which is a  clustering problem.
 This link provides an upper bound on the increase of
the Nash CUF for a clustering refinement.

\subsection{Notation and organization of the paper}
All the logarithms are in base two throughout this paper. Also $[a:b]$ for natural numbers $a$ and $b$ denotes the set $\{a,a+1,a+2,\cdots, b\}$. We will also need the following definition:
\begin{definition}\label{defip}
For a pmf $p_{X,Y}$, the \emph{information density} $\im_p(x;y)$ is defined by
\[ \im_p(x;y):=\log \frac{p(x,y)}{p(x)p(y)}.\]
\end{definition}
Markov chains are denoted by $X\rightarrow Y\rightarrow Z$, meaning that $p(x,y,z)=p(x)p(y|x)p(z|y)$.

This paper is organized as follows: in Section \ref{dc=section}, we study the divide and choose problem. In Section \ref{sec=AW}, we consider the adjusted winner algorithm and finally in Section \ref{sec=MNC}, we consider the Nash Collective Utility social welfare function. Proofs are given in the subsequent sections.

 

%
\section{Divide and Choose}\label{dc=section}
We assume that the value each player gives to different pieces 
of the cake is a
random variable on the set of possible values $\vset$ which is assumed to be
finite.  We have no specific assumption over $\vset$, but for having an
intuition, one can consider the following special case. Imagine the cake has
$m$ items: chocolate, cream, cherry,$\cdots$. In this particular example, a
valuation vector $\mathbf{v}$ is a
vector of size $m$, $(v_1, \dots, v_m)$, whose indices are nonnegative real
numbers adding up to one. The indices indicate
 interest in individual items. Thus if a certain piece of the
cake has portion $\alpha_i$ of item $i$, the value associated to this piece
w.r.t. $\mathbf{v}$ is $\sum_{i=1}^m \alpha_i v_i$. However it should be noted
that in the general case, we do not assume that valuations are vectors.
We also assume that $\dset$, the set of 
admissible divisions or admissible cuts, is finite. The
gain of each player is a deterministic function of the valuations and the
particular division $d\in\dset$. This is formalized in the following definition:
\begin{definition}
\label{def:gain}
Assume that $v_A$ and $v_B$ are the valuations of Alice and Bob
respectively and
Alice has divided the cake by $d\in\dset$. Then $\gainwrt{A}{d}{v_A}{v_B}$ and
$\gainwrt{B}{d}{v_A}{v_B}$ denote the gain of Alice and Bob respectively in
one game. We assume that $V_A$ and $V_B$ are generated from the joint
distribution $q(v_A,v_B)$, which is revealed to both Alice and Bob. The alphabet sets for these random variables are $\mathcal{V}_A=\mathcal{V}_B=\mathcal{V}$.
\end{definition}

\subsection{Spying: one-shot}
Assume that we are playing one instance of the divide and choose problem. Alice spies on Bob via a (possibly stochastic) spying function $\varphi:\mathcal{V}_B\mapsto \mathcal{M}$, and  a (possibly stochastic) cutting function $\psi:\mathcal{M}\times \mathcal{V}_B\mapsto\mathcal{D}$. 
We assume that $\mathcal{M}=[1:\sm]$ is the alphabet set of spying information $M$ and $\log\sm$ is the number of spying bits.

 \begin{thm}[One shot]\label{thm:one-shot}
Given any $q(u,d|v_A,v_B)=q(u|v_B)q(d|u,v_A)$, there is a spying strategy of a message in $[1:\sm]$  in which the  gain of Alice after spying is bounded from below by
\begin{equation}\label{eq:GP}
\e_{UDV_AV_B}\big[
\dfrac{1}{1-\sj^{-1}+\sj^{-1} 2^{\im_q(V_B;U)}}\dfrac{1}{1+ (\sj-1)\sm^{-1}2^{-\im_q(V_A;U)}
}\gainwrt{A}{D}{V_A}{V_B}\big]
\end{equation}
where $\sj$ is any natural number, expectation is with respect to $q(u|v_B)q(d|u,v_A)q(v_A, v_B)$, and the information densities are defined as in Definition \ref{defip}. The corresponding gain of Bob would be greater than or equal to
\begin{equation}\label{eq:GP}
\e_{UDV_AV_B}\big[
\dfrac{1}{1+\sj^{-1} 2^{\im_q(V_B;U)}}\dfrac{1}{1+ \sj\sm^{-1}2^{-\im_q(V_A;U)}
}\gainwrt{B}{D}{V_A}{V_B}\big],
\end{equation}
Moreover, loosening the bounds given in equation \eqref{eq:GP} gives the following lower bound on Alice's gain
\begin{align*} (1-3\times 2^{-\gamma})\e\big[
\gainwrt{A}{D}{V_A}{V_B}\big]-\bar g(1-3\times 2^{-\gamma})\mathsf{Prob}[\left\{\log\sj-\im(V_B;U)\ge\gamma,\ \mathsf{or}\  \im(V_A;U)-\log(\sj\sm^{-1})\ge\gamma\right\}],\end{align*}
where $\gamma$ is any positive number and $\bar g=\max_{d, v_A, v_B}\gainwrt{A}{d}{v_A}{v_B}$. A similar statement holds for Bob's gain.
\end{thm} 
The proof can be found in Section  \ref{sec:dc-achievability-one-shot}.
\begin{remark} The form of loosened bound makes it amenable to finite blocklength by setting $\gamma$ of order $\log(n)$. Let us bound the equation  \eqref{eq:GP} from below as follows:
\begin{equation}\label{eq:GP23}
\e_{UDV_AV_B}\big[
\dfrac{1}{1+\sj^{-1} 2^{\im_q(V_B;U)}}\dfrac{1}{1+ \sj\sm^{-1}2^{-\im_q(V_A;U)}
}\gainwrt{A}{D}{V_A}{V_B}\big]
\end{equation}
The first term in the denominator, $1+\sj^{-1} 2^{\im_q(V_B;U)}$ 
corresponds to a covering lemma in the asymptotic case, while the second term $1+ \sj\sm^{-1}2^{-\im_q(V_A;U)}$
corresponds to a packing lemma. Next, consider the special case of $J=1$ in the above bound. Then, the bound becomes
\begin{align*}
\e_{UDV_AV_B}\big[
\dfrac{1}{2^{\im_q(V_B;U)}}\gainwrt{A}{D}{V_A}{V_B}\big]&=\e_{UDV_AV_B}\big[
\dfrac{q(V_B)q(U)}{q(V_B;U)}\gainwrt{A}{D}{V_A}{V_B}\big]\\&=\e_{q(u)q(v_Av_B)q(d|uv_Av_B)}\big[
\gainwrt{A}{D}{V_A}{V_B}\big]
\end{align*}
is the payoff that could be achieved without any communication and $U$ being Alice's private randomness.
\end{remark}

\subsection{Spying: asymptotics}
Consider $n$
i.i.d.\ repetitions of the game and consider the average gain over these games.
Valuations of Alice and Bob over the $n$ games are denoted by two
sequences
of length $n$, 
$V_A^n$
for Alice and
$V_B^n$  
for Bob. These two
sequences are independently and identically generated from the joint
distribution $q(v_A,v_B)$. Let $R$ denotes the spying rate per game from Bob to Alice, i.e. it is equal to 
the total number of bits spied from Bob divided by $n$. The formal definition of an $n$-game code is in
order.

\begin{definition}
 \label{def:n-game-protocol}
An $n$-game $(n,R)$ code consists of
communication variables $C$ with encoder
$p(c|v_B^n)$ as well as a division strategy
$p(d^n|v_A^n, c)$ where
$
\frac{1}{n} H(C) \leq R.
$
The gains associated with this code are random variables
\begin{equation}
\begin{split}
\tilde{G}_A &= \frac{1}{n} \sum_{i=1}^n \gainwrt{A}{D_i}{V_{A,i}}{V_{B,i}}, \\
\tilde{G}_B &= \frac{1}{n} \sum_{i=1}^n \gainwrt{B}{D_i}{V_{A,i}}{V_{B,i}}.
\end{split}
\end{equation}
The division $D^n$ over $n$ games is then performed by Alice
based on the information she has: spying information $C$ and her
own preferences $V_A^n$. 
\end{definition}

\begin{definition}
 \label{def:achievability}
A $(R,G_A,G_B)$
rate gain tuple is said to be achievable if for any $\delta>0$ and $N$, there
exists a $(n, R)$ code with $n>N$ where the associated gains
$\tilde{G}_A$ and $\tilde{G}_B$ satisfy the following inequalities with
probability at least $1-\delta$:
\begin{equation}
 |\tilde{G}_A - G_A| < \delta \qquad |\tilde{G}_B - G_B| < \delta.
\end{equation}
\end{definition}

\begin{definition}
 \label{def:region-closure}
  The spying rate gain region is the closure of all
achievable tuples \linebreak $(R, G_A, G_B)$ and is
denoted by $\mathcal R$.
\end{definition}

There are two relaxations in our formulation in this part compared to the traditional fair division setup. 
Firstly the number of games $n$ is allowed to converge to infinity (it is not a one-shot result). Secondly we are not
following the maximin rule (i.e. maximizing the minimum gain) with probability
one. Instead we are relaxing this by requiring a guarantee with probability $1-\delta$ where
$\delta$ converges to zero only after $n$ converges to infinity. 

We have
\begin{thm}
\label{thm:DC-region}
The set $\mathcal R$ is the closure of all rate gain tuples
$(R, G_A, G_B)$ such that
\begin{equation}
 \begin{gathered}
  R > I(V_B;U|V_A), \\
  \ev{\gainwrt{A}{D}{V_A}{V_B}} = G_A, \\
  \ev{\gainwrt{B}{D}{V_A}{V_B}} = G_B, \\
 \end{gathered}
\end{equation}
for some $(U,D)$ satisfying the Markov chain relations
\begin{equation}
\begin{gathered}
 V_A \rightarrow V_B \rightarrow U, \\
 V_B \rightarrow V_A,U \rightarrow D,
\end{gathered}
\end{equation}
and random variable $D$ taking values on the set of all
divisions $\dset$.
\end{thm}

This  achievability part of this theorem follows the one-shot result given in Theorem~\ref{thm:one-shot}. Alternatively, it follows 
from a result on empirical coordination. This proof is given in Section ~\ref{sec:dc-achievability}.

\subsubsection{Spying rate and equitability}
In a practical scenario it is quite reasonable to assume that Alice uses the 
information selfishly in order to maximize her gain. Therefore we can define
the selfish gain $\gsel_A(R)$ to be the maximum gain Alice can obtain
limiting the communication rate to a value $R$, i.e.\
\begin{equation*}
 \gsel_A = \max_{(R,G_A, G_B)\in\mathcal R} G_A.
\end{equation*}
Bob always chooses
the piece he likes more with no concern about Alice's gain. Let $\gsel_B(R)$
denote the gain
associated with Bob in this case. Since we want to study the equitability of
the division (a fairness criterion discussed at the beginning of the
introduction),
we define the difference between these two
gains as
\begin{equation}
 \dsel(R) = \gsel_B(R) - \gsel_A(R).
\end{equation}
A spying rate $R$ results in an equitable division if $\dsel(R)=0$.

To illustrate several
aspects of the result, we consider a few examples. Imagine the cake has only two
items, say cream and
chocolate, and the set of possible valuations is $\vset=\{\cream, \chocolate\}$
where $\cream$ denotes complete interest in cream and no interest in
chocolate, i.e.\ $\cream=(1,0)$ while $\chocolate=(0,1)$ denotes complete
interest in chocolate. Assume the cake is half cream and half chocolate and the
set of possible divisions is $\dset=\{\half,\full\}$ where $\half$ means
dividing the cake so that in each piece we have half cream and half chocolate
and $\full$ means dividing the cake so that one piece is full cream and one is
full
chocolate. Assume that the joint distribution over valuations,
$p(v_A,v_B)$ is
as $p(\cream, \cream) = p(\chocolate, \chocolate) = 2/6$ and
$p(\cream,\chocolate) = p(\chocolate, \cream) = 1/6$.

For a fixed $R$, the set of achievable gain pairs $(G_A,G_B)$ form a region in
$\reals^2$ which is illustrated in Figure~\ref{fig:DC-region-3} for
different values of $R$. Figures \ref{fig:DC-GASel-3}, \ref{fig:DC-GBSel-3} and \ref{fig:DC-DeltaSel-3}
respectively show the values of $\gsel_A$, $\gsel_B$ and $\dsel$ as functions
of $R$ for our example. As we see, Alice's spying gain always increases with
the rate; which is expected, since she can use or ignore the extra spying
information. However, the interesting observation is that Bob's gain
increases up to some value for small rates and then decreases.  This means that
up to a
point, sharing information is advantageous for both sides. The other point
is that the value of $\dsel$ is zero only when $R=0$ and $R\geq H(V_B|V_A)$,
this
suggests that the division is equitable just in case of zero information or
full information. The reason for this is that in this example, the divisions are so that Bob's
gain is always greater or equal than that of Alice. In other words for any $v_A,
v_B \in
\vset$
and $d\in \dset$, $\gainwrt{B}{d}{v_A}{v_B} \geq \gainwrt{A}{d}{v_A}{v_B}$,
therefore we always have $G_B \geq G_A$.

\begin{figure}
 \centering
 \begin{subfigure}[t]{0.4\textwidth}
  \centering
  \includegraphics[width=\textwidth]{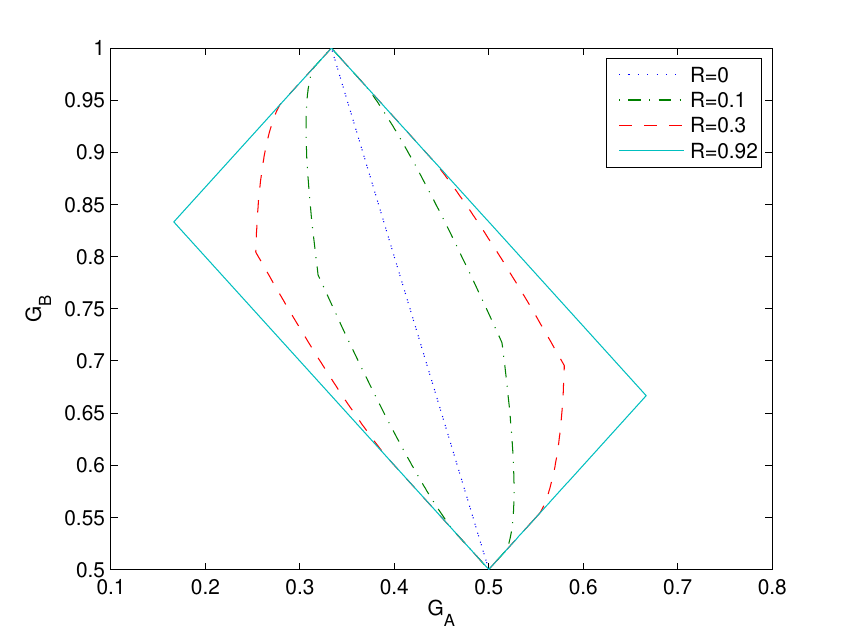}
  \caption{\label{fig:DC-region-3} The  region $(G_A, G_B)$ for a fixed $R$.}
 \end{subfigure}
 \quad
 \begin{subfigure}[t]{0.4\textwidth}
\centering
\includegraphics[width=\textwidth]{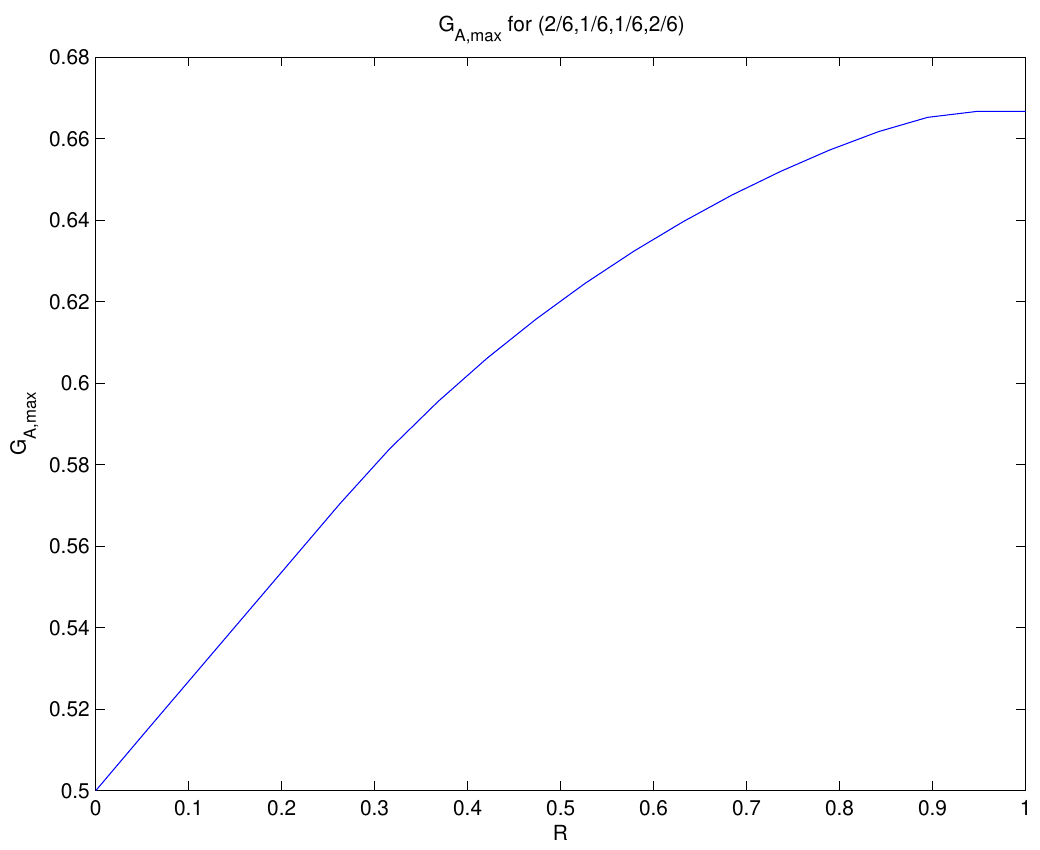}
\caption{The maximum achievable gain for Alice when she acts selfishly,
$\gsel_A$ as a function of the rate of communication $R$.
\label{fig:DC-GASel-3}}
\end{subfigure}

\begin{subfigure}[b]{0.4\textwidth}
 \centering
\includegraphics[width=\textwidth]{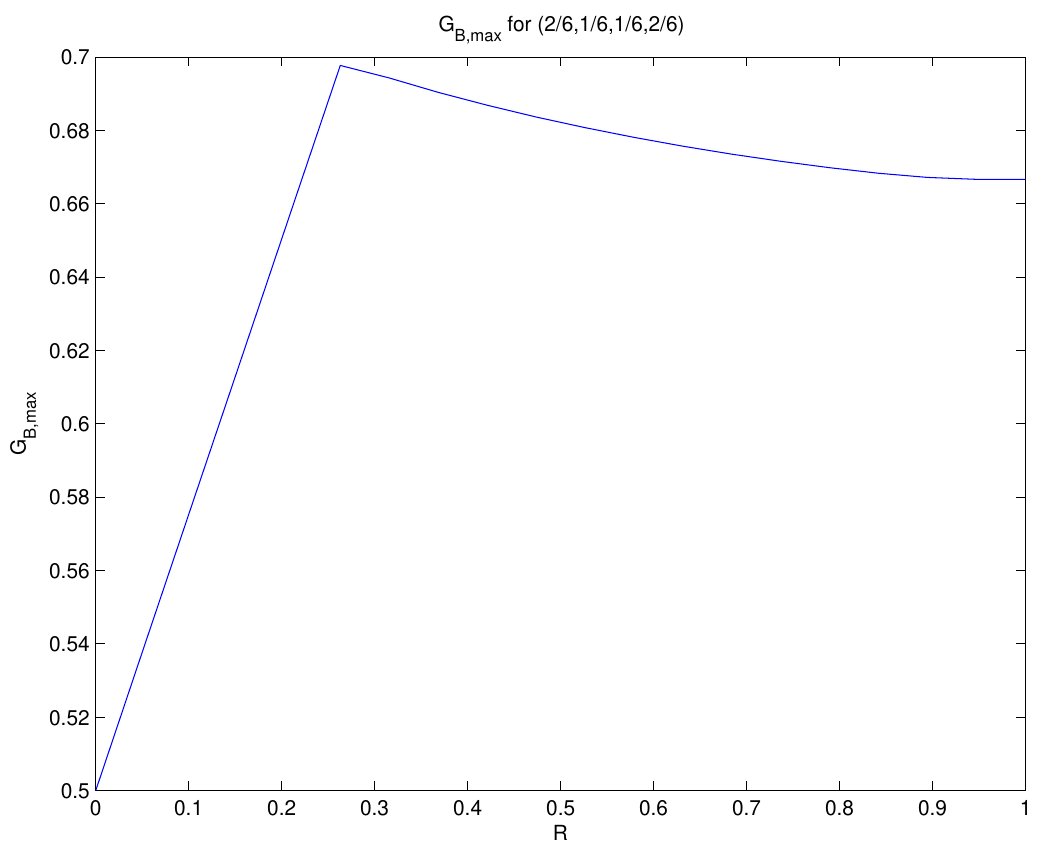}
\caption{The gain associated with Bob when Alice acts selfishly,
$\gsel_B$, as a function of the rate of communication $R$.
\label{fig:DC-GBSel-3}}
\end{subfigure}
\quad
\begin{subfigure}[b]{0.4\textwidth}
 \centering
\includegraphics[width=\textwidth]{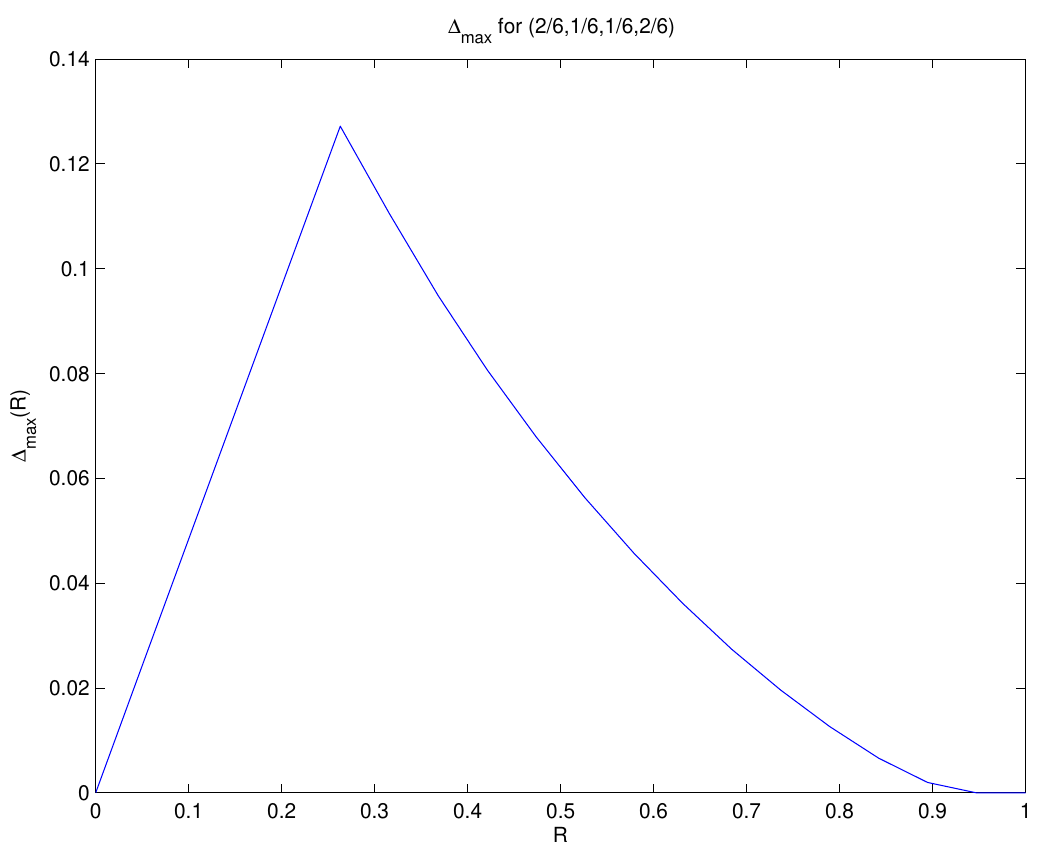}
\caption{Difference between two selfish gains, $\dsel(R) = \gsel_B(R) -
\gsel_A(R)$, as a function of the rate of communication $R$.
\label{fig:DC-DeltaSel-3}}
\end{subfigure}
\caption{Our first example of explicit information transmission in Divide and 
Choose}
\end{figure}


However, other behaviors can be observed when changing $\vset$, $\dset$
and the joint probability. For instance, by keeping $\vset$ and $\dset$
unchanged, but changing the joint probability distribution as
\begin{equation}
\label{eq:DC-sample-1-p}
 \begin{split}
  p(\cream,\cream) &= \frac{1}{14}, \quad
  p(\cream,\chocolate) = \frac{6}{14}, \\
  p(\chocolate,\cream) &= \frac{5}{14}, \quad
  p(\chocolate,\chocolate) = \frac{2}{14}.
 \end{split}
\end{equation}
We observe that Bob's gain, $\gsel_B(R)$, initially decreases and then increases
slightly, as depicted in Figure~\ref{fig:DC-GBSel-1}. In this example, unlike
the latter one, it is more probable that the two players
have different valuations, therefore in the case of zero information, it is
more beneficiary for Alice to divide the cake by $\full$ which results in a
gain of $1$ for Bob. The rate gain region for this example is illustrated in
Figure~\ref{fig:DC-2D-1}.

\begin{figure}
\centering
\begin{subfigure}[t]{0.4\textwidth}
 \centering
  \includegraphics[width=\textwidth]{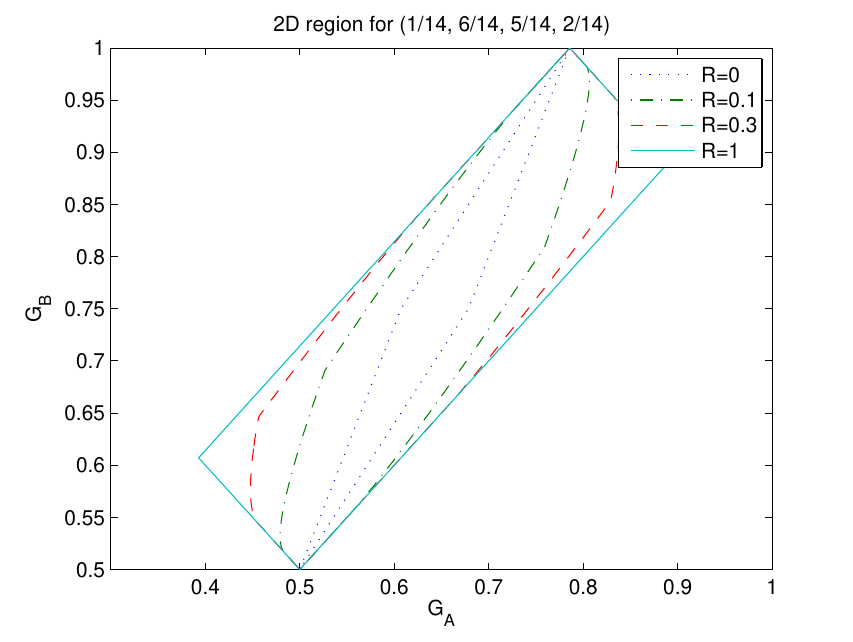}
\caption{The rate gain region \label{fig:DC-2D-1}}
\end{subfigure}\quad
\begin{subfigure}[t]{0.4\textwidth}
 \centering
  \includegraphics[width=\textwidth]{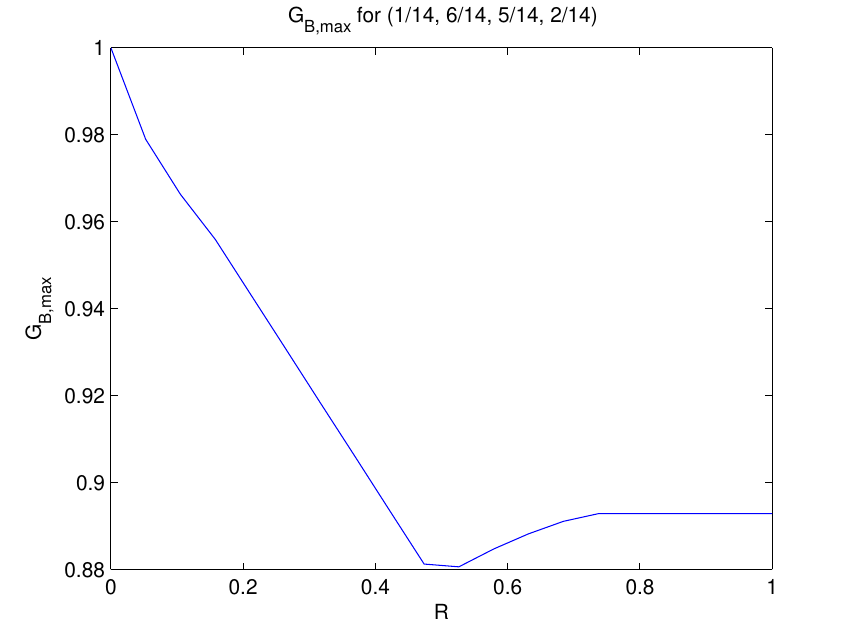}
\caption{$\gsel_B$ as a function of $R$ \label{fig:DC-GBSel-1}}
\end{subfigure}
\caption{$\gsel_B$ and the rate region for the probability distribution 
\eqref{eq:DC-sample-1-p}}
\end{figure}

As mentioned before, Bob's gain will be always greater than or equal to
Alice's for the choice of $\dset = \{ \full, \half\}$. Now, we change the
setup to,
\begin{equation}
 \label{eq:DC-setup-2}
 \begin{gathered}
  \vset = \{ \creamy, \chocolaty\}, \\
  \dset = \{ \divcream, \divchoc\}, \\
  p(\creamy, \creamy) = \frac{1}{6} \quad p(\creamy, \chocolaty) = \frac{2}{6},
\\
  p(\chocolaty, \creamy) = \frac{2}{6} \quad p(\chocolaty, \chocolaty) =
\frac{1}{6},
 \end{gathered}
\end{equation}
where $\chocolaty$ means $2/3$ interest in chocolate and $1/3$ interest in
cream, $\creamy$ means $1/3$ interest in chocolate and $2/3$ interest in cream,
$\divcream$ means dividing in a way so that in one piece we have all chocolate
and $3/16$ of the whole cream and $\divchoc$ denotes dividing in a way so that
in one piece we have all the cream and $3/16$ chocolate. In this case,
Figures~\ref{fig:DC-GASel-2}, \ref{fig:DC-GBSel-2} and \ref{fig:DC-DeltaSel-2}
show $\gsel_A$, $\gsel_B$ and $\dsel$ respectively as a functions of $R$. As we
see for a rate $R_{eq}$, $\dsel(R_{eq})=0$ which shows that with the
information rate of $R_{eq}$, the division is equitable, while for information
rate less than that amount,
Bob has advantage and with more information rate, Alice has advantage. In fact,
this
value of information makes an equilibrium between the natural advantage of Bob
over Alice and the information Alice gains about Bob's valuations.
Figure~\ref{fig:DC-2D-2} shows that the rate region for this example is a part
of a line in the plane for all values of $R$.
\begin{figure}
\centering
\begin{subfigure}[t]{0.4\textwidth}
 \centering
  \includegraphics[width=\textwidth]{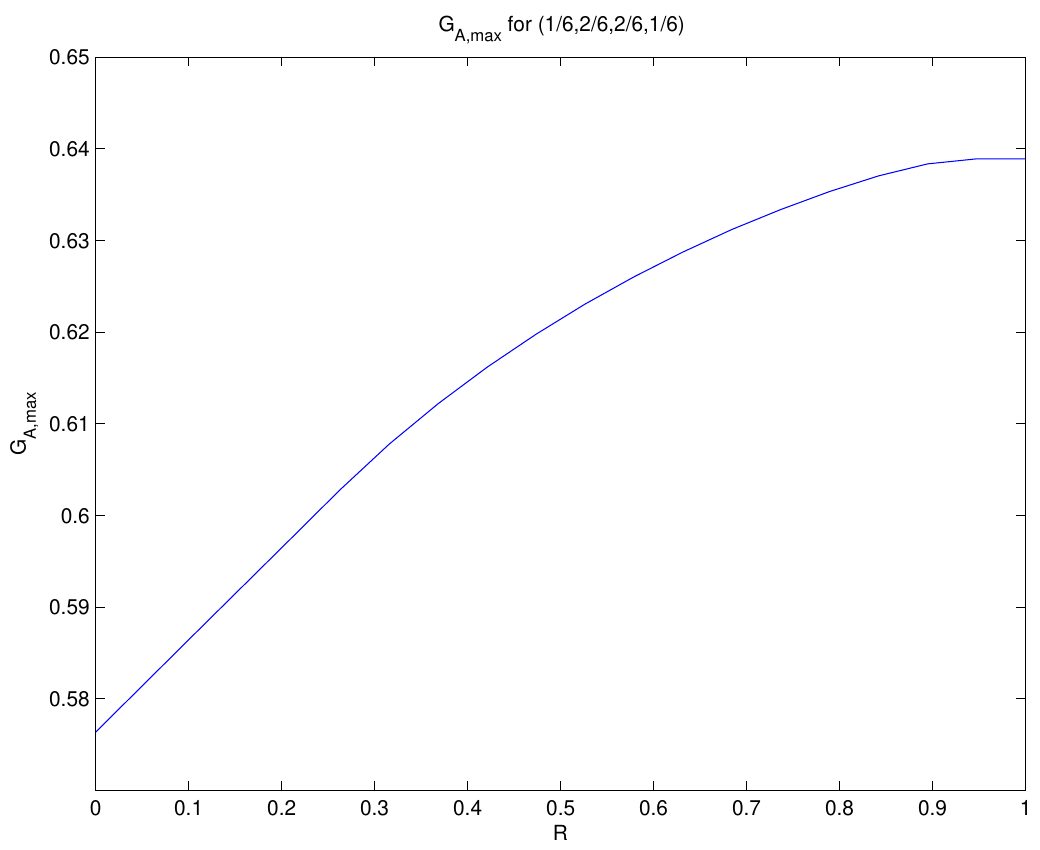}
  \caption{$\gsel_A(R)$ \label{fig:DC-GASel-2}}
\end{subfigure}
\quad
\begin{subfigure}[t]{0.4\textwidth}
 \centering
  \includegraphics[width=\textwidth]{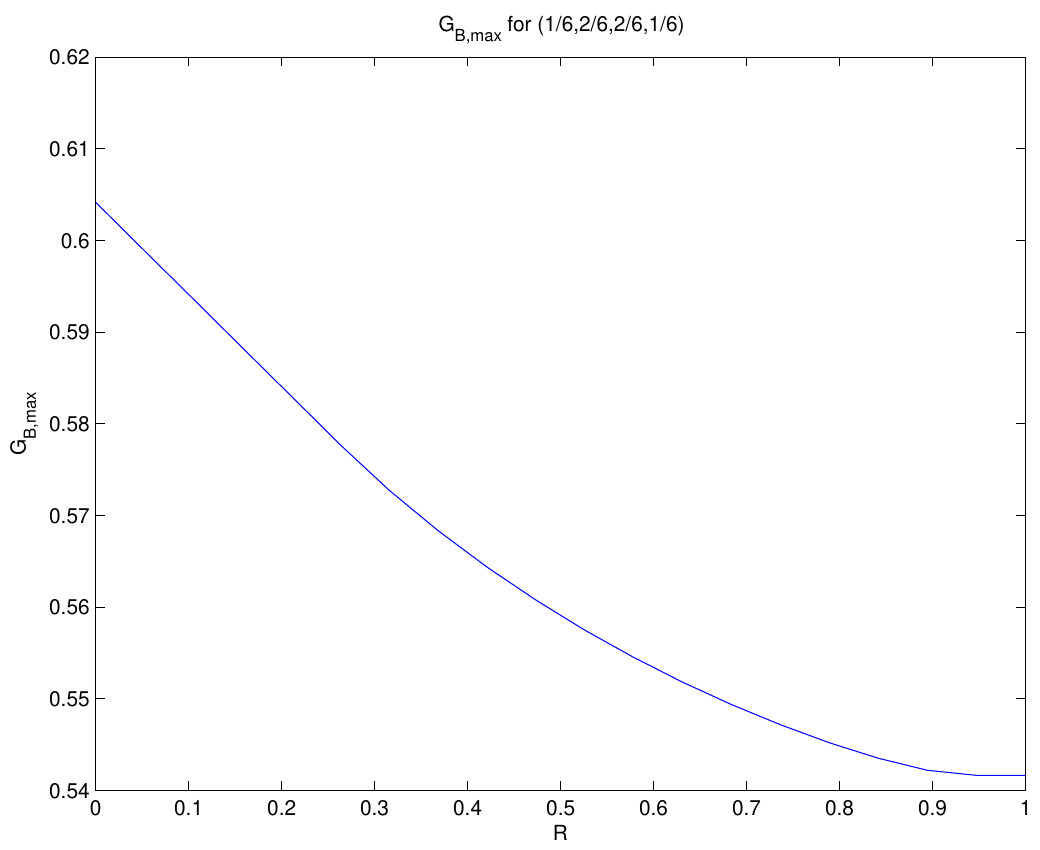}
  \caption{$\gsel_B(R)$ \label{fig:DC-GBSel-2}}
\end{subfigure}

\begin{subfigure}[t]{0.4\textwidth}
 \centering
  \includegraphics[width=\textwidth]{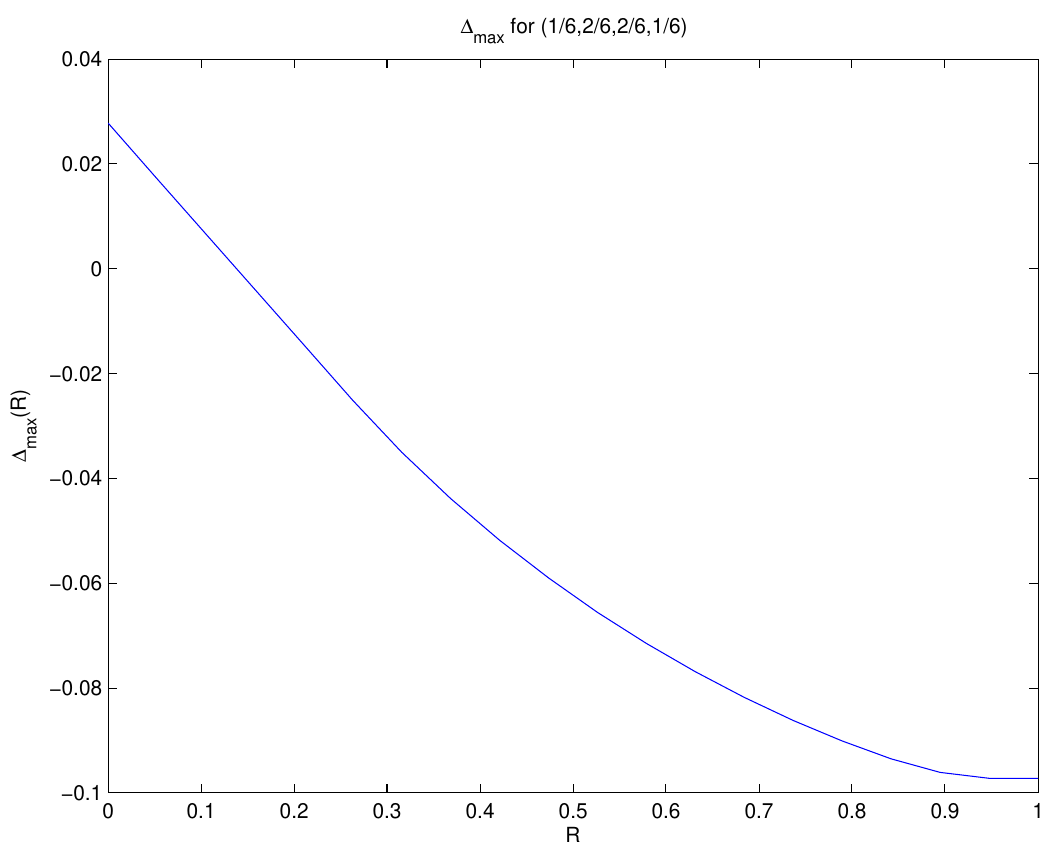}
  \caption{$\dsel(R)$ \label{fig:DC-DeltaSel-2}}
\end{subfigure}
\quad
\begin{subfigure}[t]{0.4\textwidth}
 \centering
  \includegraphics[width=\textwidth]{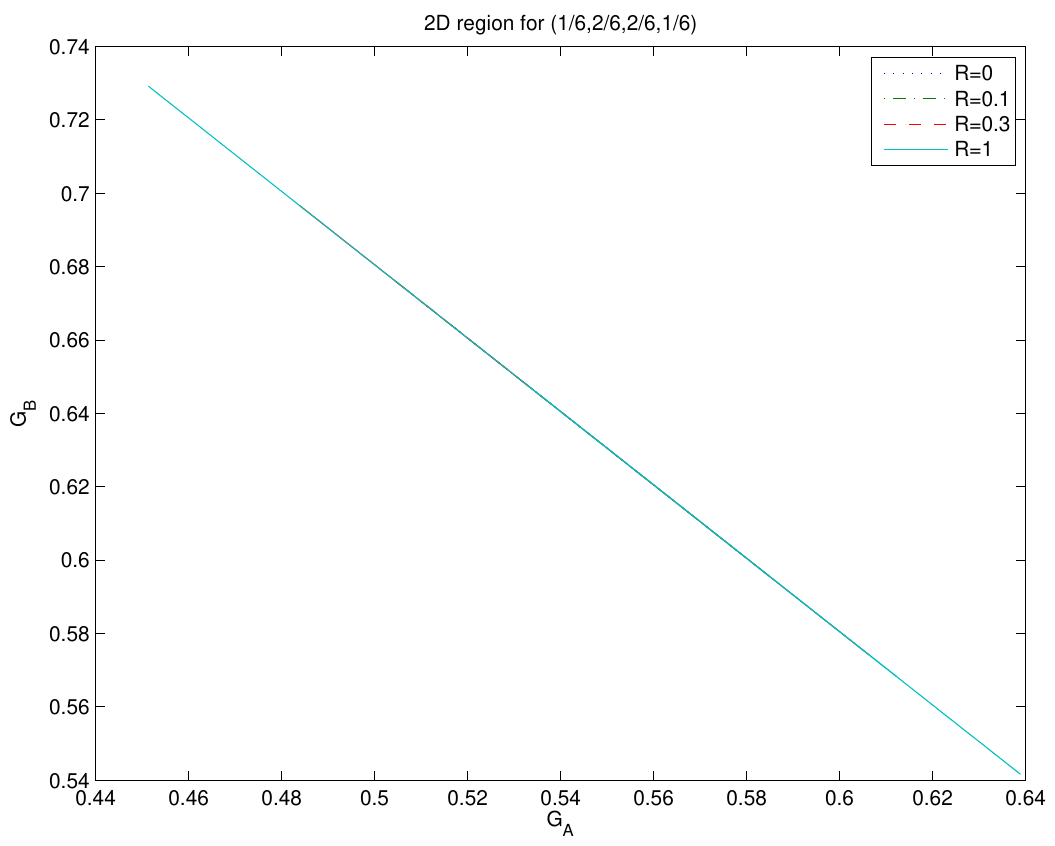}
  \caption{The rate gain region \label{fig:DC-2D-2}}
\end{subfigure}
\caption{Setup of eqn. \eqref{eq:DC-setup-2}}
\end{figure}

\begin{figure}
\centering
\begin{subfigure}[t]{0.4\textwidth}
 \centering
  \includegraphics[width=\textwidth]{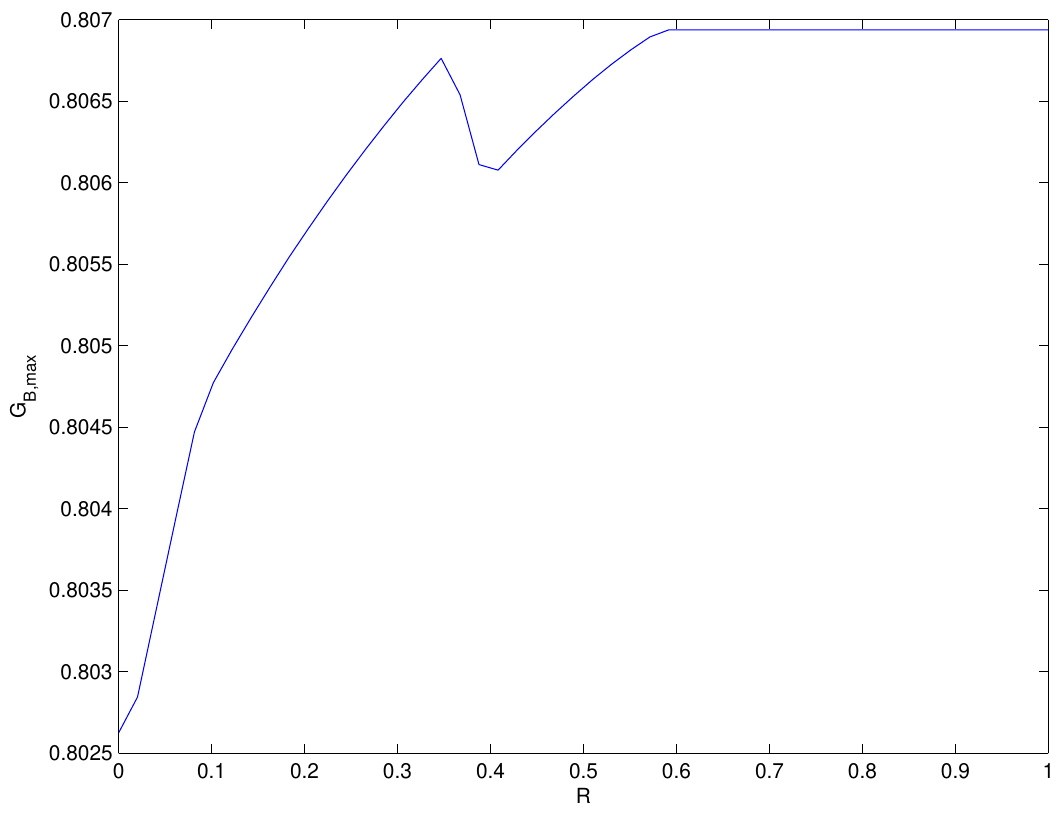}
  \caption{$\gsel_B(R)$\label{fig:DC-GBSel-5}}
\end{subfigure}
\quad
\begin{subfigure}[t]{0.4\textwidth}
 \centering
  \includegraphics[width=\textwidth]{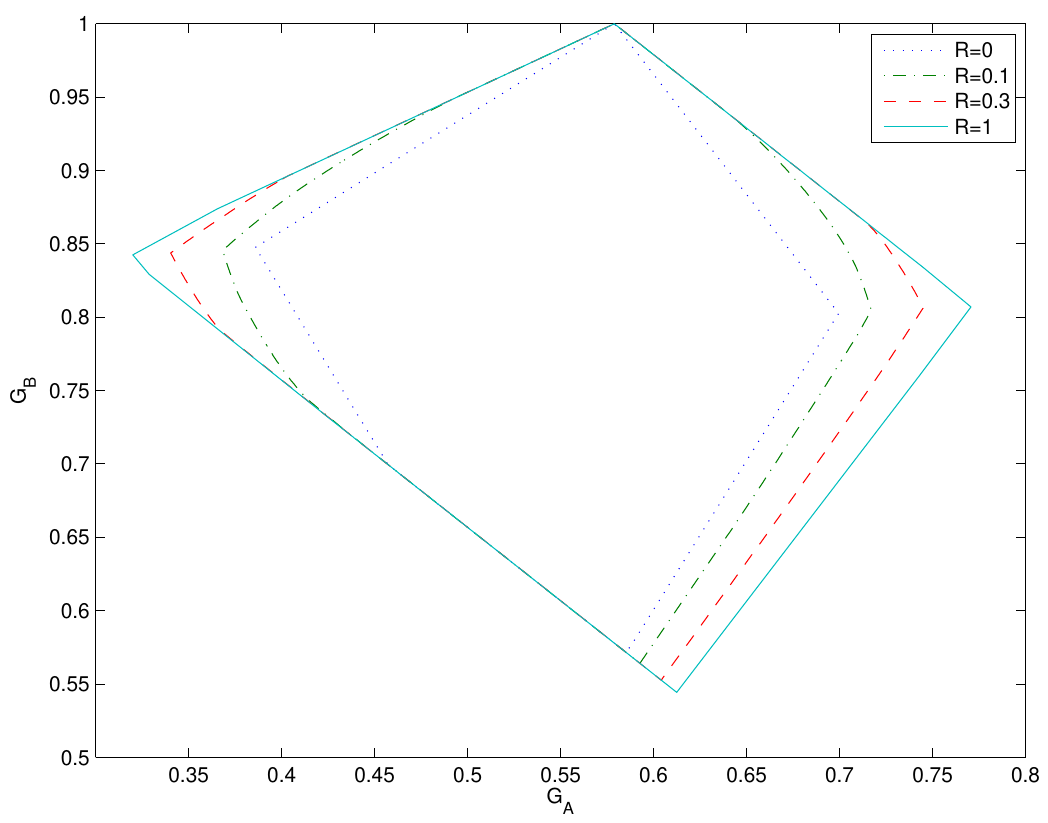}
  \caption{The rate gain region\label{fig:DC-2D-5}}
\end{subfigure} 
\caption{Setup of equation \eqref{eq:DC-setup-5}}
\end{figure}

Another interesting fact could be observed by changing the probability
distribution of \eqref{eq:DC-setup-2} into
\begin{equation}
 \label{eq:DC-setup-5}
 \begin{gathered}
  \vset = \{ \creamy, \chocolaty\}, \\
  \dset = \{ \divcream, \divchoc\}, \\
  p(\creamy, \creamy) = \frac{1}{19} \quad p(\creamy, \chocolaty) =
\frac{9}{19},
\\
  p(\chocolaty, \creamy) = \frac{2}{19} \quad p(\chocolaty, \chocolaty) =
\frac{7}{19}.
 \end{gathered}
\end{equation}
As we can see in Figure~\ref{fig:DC-GBSel-5}, Bob's gain first increases, then
decreases and then increases again. The
region of this setup is depicted in Figure~\ref{fig:DC-2D-5}. This together with
our latter observations
suggest that Bob's gain does not have an specific behavior in general.

\subsection{Information sharing \label{sec:dc-implicit22}}
Let us consider the following problem:  assume that there is a one-way communication link of limited rate $R$ from player B to player A. Unlike the previous problem, player B makes a decision as to whether communicate any information to player A (instead of
player A spying on player B by choosing the information that will be communicated by player B). Player B can potentially benefit himself by communicating cleverly (yet honestly) to player A who will make the cut. The choice of communication protocol then serves the role of the strategy of the player. This brings in a game aspect to the problem. 

Assume that Alice and Bob are playing the game on $n$ i.i.d.\ repetitions of the divide and choose game, i.e. $q(v_A^n, v_B^n)=\prod_{i=1}^n q(v_{Ai}, v_{Bi})$. The two parties have also possibly access to a shared randomness $S$ taking values in a finite set $\mathcal S$. The amount of the shared randomness can be arbitrary, but fixed before the game starts. The strategy of Bob is $q(c|v_{B}^n,s)$ where $c$ is the message on the alphabet set $[1:2^{nR}]$ and $S$ is a shared randomness between the two parties (independent of $(V_A^n, V_B^n)$). Alice's strategy is $q(d^n|v_{A}^n, c, s)$ where $d_i$ is the division by Alice in $i$-th game. Alice's payoff is the sum of the expected value of her payoff over the $n$ games (after Bob picking his most favorable part in each game).

Observe that Bob may choose not to use the shared randomness by setting  $q(c|v_{B}^n,s)= q(c|v_{B}^n)$. Also, when there is no communication between Alice and Bob, existence (or lack thereof) of shared randomness is not important in the set of gain pairs they can achieve.

 Observe that the cut and choose game is  not a zero-sum game, and hence there is no unique Nash equilibrium payoff. There may be many different optimal ways for Alice to cut the cake; different optimal ways of cutting the cake that are all the same from the perspective of Alice, but can affect Bob's average payoff. Fixing any of one these cutting strategies for Alice, she will not have any incentive for changing her strategy. Therefore, each of these strategies lead to an equilibrium. 

To state our main result, let us make the following definition, which considers the equilibriums in the one-shot instance of the problem with no communication:
\begin{definition}Given $q(v_A, v_B)$ with no communication from Bob to Alice, we define a set $\mathsf{G}(q(v_A, v_B))$ as the set of average gain pairs $(G_A, G_B)$ that could be achieved assuming that Alice is playing optimal and selfish in a single instance of the game ($n=1$). Since Alice may have several optimal ways to cut the cake, the  set $\mathsf{G}(q(v_A, v_B))$ can have more than one point. \end{definition}
Observe that if $(G_A, G_B)$ and $(G'_A, G'_B)$ are in $\mathsf{G}(q(v_A, v_B))$, we must have $G_A=G'_A$ by definition of optimality for Alice. Furthermore, $\mathsf{G}(q(v_A, v_B))$ is convex since  Alice can randomize between two optimal cutting strategies. Thus, $\mathsf{G}(q(v_A, v_B))$ can be expressed a set of pairs of the form $(G_A(q(v_A, v_B)), G_{B})$ for some $G_{B,{l}}(q(v_A, v_B))\leq G_B\leq G_{B,{u}}(q(v_A, v_B))$.

Let us now characterize a set of $\epsilon$-equilibriums, wherein any change of strategy by Alice and Bob will not increase their payoff by more than $\epsilon$.
\begin{thm}For any $p(v_A, v_B)$, fix some pair $(G_A(p(v_A, v_B)), G^*_B(p(v_A, v_B))) \in \mathsf{G}(p(v_A, v_B))$. Let 
$$G_{Bmax}(q(v_A, v_B))=\max_{q(c|v_B): I(C;V_B)\leq R}\sum_{c}q(c)G^*_B(q(v_{A}, v_{B}|c)),$$
where $q(c,v_A, v_B)=q(c|v_B)q(v_A, v_B)$. 
Let $q(c|v_B)$ be any arbitrary maximizer of $G_{Bmax}$. Then, given any $\epsilon>0$, for sufficiently large $n$, one can find 
a shared randomness assisted strategy  for players $A$ and $B$ with communication of rate $R$ such that (i) the strategies form 
an $\epsilon$-equilibrium, (ii) the corresponding payoff pair is coordinatewise within $\epsilon$ distance of \begin{align}\left(\sum_{c}q(c)G_A(q(v_{A}, v_{B}|c)), G_{Bmax}\right).\label{eqn;rate;gain;pair}\end{align}
\label{thm:sharedrandomness-DC}
\end{thm}
Proof of this theorem can be found in Section  \ref{APP:thm:sharedrandomness-DC}.

\subsection{Implicit Information Transmission \label{sec:dc-implicit}}
In this part we assume that Alice's and Bob's valuations are 
generated from a given probability distribution $q(\va, \vb)$ once and for all 
and will remain fixed during the subsequent $n$ stage game. Alice's and Bob's 
valuations are chosen from $\vaset$ and $\vbset$, respectively, which are the 
set of their permissible valuations.

As depicted in Fig. \ref{fig:one-stage-cake}, we assume that Alice has two possible actions at each steps: she can either 
choose to play risky (denoted by $\risky$) or non-risky (denoted by 
$\nonrisky$). If she plays non-risky, then independent of Bob's valuation, both 
parties receive exactly half the cake, i.e. playing non-risky is equivalent to 
cutting the cake into two pieces which worth exactly $1/2$ with respect to any 
possible valuation of Alice and Bob. Therefore the two pieces have exactly the same value for both 
players, and without loss of generality we do not consider any action for Bob. On the other hand, if Alice chooses the 
action $\risky$, 
then Bob can choose the left piece (action $\Bleft$) or the right piece (action 
$\Bright$). If Bob plays $\Bleft$, then Alice receives a gain of $G_A$ and Bob 
receives $G_B$ where $G_A$ is a function of Alice's valuation and $G_B$ is a 
function of Bob's valuation. Since the whole cake has a gain of 1 for both 
parties, if Bob chooses $\Bright$, then Alice and Bob receive $1-G_A$ and 
$1-G_B$, respectively.
We shall assume that
\[
 G_A \neq 1/2 \qquad G_B \neq 1/2
\]
If two different valuations of $v_A, v_B$  result in the same values for $G_A$ and $G_B$, then in 
the sense of strategies and equilibriums these two valuations are identical. 
Therefore in the following discussion we shall forget about valuations and instead assume that there is a 
joint distribution
\[
 p(g_A, g_B)
\]
where Alice knows $G_A$ and Bob knows $G_B$. In fact $G_A$ and $G_B$ are 
sufficient statistics. Therefore the one stage game is of the form depicted in 
Figure~\ref{fig:one-stage-cake}.

We assume that this one stage game is repeated $n$ times during which Alice and 
Bob's valuations are chosen randomly at the beginning of the game and remain 
unchanged. Also we assume that both players causally observe each other's actions and 
recall these actions as well as their own actions. Like the one stage game, 
Alice is only aware of her $G_A$ and Bob only knows $G_B$.
We denote Alice's and Bob's actions in $n$-stage game by $a_{[1:n]}$ and 
$b_{[1:n]}$ respectively. At stage $t$, Alice's and Bob's gain are denoted by 
$\ga{a_t,b_t}$ and $\gb{a_t,b_t}$, respectively. Note that these 
quantities are random variables even if $a_t$ and $b_t$ are known, since they 
depend on $G_A$ and $G_B$, which are random. In 
the same manner, 
for a sequence of steps, say $a_{[i:j]}$ and $b_{[i:j]}$, we define
\[
 \gx{a_{[i:j]},b_{[i:j]}} = \sum_{t=i}^j \gx{a_t,b_t} \qquad \qquad 
X\in\{A,B\}.
\]
which is again a random variable.
The gain of the whole game is $\ev{\gx{A^n,B^n}}$ where the expected value 
is taken over all possible $G_A,G_B$ and all possible actions resulting from 
strategies. Note that since $n$ is fixed throughout the problem, we do not need 
to normalize the gain with $n$.

When $G_A=g_A$ is known, $\ga{a_{[i:j]},b_{[i:j]}}$ will no longer be a random variable. Similarly, when $G_B=g_B$ is known, $\gb{a_{[i:j]},b_{[i:j]}}$ will no longer be a random variable. In this case, we denote these gains by
\[
\gxwrt{a_{[i:j]},b_{[i:j]}}{g_X} \qquad X \in \{A,B\}.
\]

The main result of this section is to identify an equilibrium for the repeated form of the game.

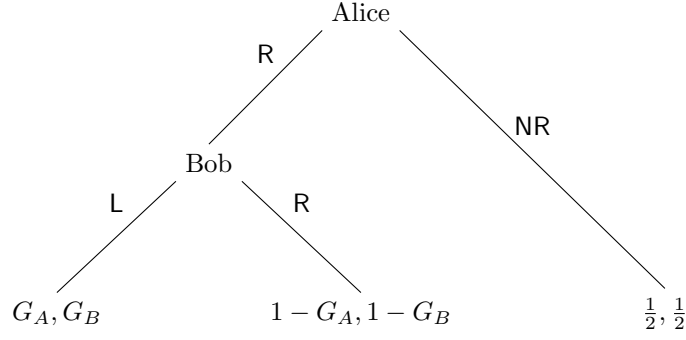
\begin{figure}
 \centering
 \begin{tikzpicture}
  \node (root) at (0,0) {Alice};
  \node (risky) at (-2,-2) {Bob};
  \node (nonrisky) at (4,-4) {$\frac{1}{2}, \frac{1}{2}$};
  \node (left) at (-4,-4) {$G_A, G_B$};
  \node (right) at (0,-4) {$1-G_A,1-G_B$};

  \draw (root.south west) --node[above=2mm] {$\risky$} (risky.north);
  \draw (risky.south west) --node[above=2mm] {$\Bleft$} (left.north);
  \draw (risky.south east) --node[above=2mm] {$\Bright$} (right.north);
  \draw (root.south east) --node[above=2mm] {$\nonrisky$} (nonrisky.north);
 \end{tikzpicture}
\caption{\label{fig:one-stage-cake} The one stage game in our cake cutting 
scenario. $G_A$ and $G_B$ are random variables chosen by the nature, Alice only 
knows $G_A$ and Bob only knows $G_B$ but both of them are aware of their joint 
probability distribution.
Note that if Alice plays $\nonrisky$, since both pieces of cake worth exactly 
$1/2$ for both parties, we can equivalently assume that Bob has no action left.}
\end{figure}

It is worthwhile to mention the differences between this setup and that of 
Section~\ref{sec:dc-implicit}:
\begin{enumerate}
 \item In the previous setup, valuations are generated independently in $n$ 
stages while in this setup they are generated once and for all.
  \item In the explicit setup, communication or spying is done prior to the divide and 
choose procedure, while in the 
implicit setup, information is transferred through actions.
\item In Section~\ref{sec:dc-implicit} we assume that Bob always chooses the 
piece which is more valuable for him, which was reasonable since stages 
were completely independent, while in the setup of this section, Bob is free to 
choose whichever part he wishes; however, as we will see, Bob is better off to 
choose the more valuable piece.
\item In this section we assume that there are only two permissible divisions 
for Alice, while in Section~\ref{sec:dc-implicit}, the set of permissible 
divisions, $\dset$ is an arbitrary finite set.
\end{enumerate}

As we will see, it is more convenient to look at Bob's action from another 
point of view. We say that Bob plays selfish, or $\selfish$, if he choses the 
piece which has more value to him, and we say he plays $\nonselfish$ is he 
chooses the cake with less value. More precisely,
\[
 \selfish = \begin{cases}
             \Bleft & G_B > 1/2 \\
             \Bright & G_B < 1/2
            \end{cases}
\]
and
\[
 \nonselfish = \begin{cases}
                \Bright & G_B > 1/2 \\
                \Bleft & G_B < 1/2
               \end{cases}
\]
Also define
\[
\Gamax = \max \{ G_A, 1-G_A\} \qquad \Gamin = \min \{ G_A, 1-G_A\},
\]
and similarly
\[
\Gbmax = \max \{ G_B, 1-G_B\} \qquad \Gbmin = \min \{ G_B, 1-G_B\}.
\]
When the value of $G_A=g_A$ is known, we use $\gamax=\max \{ g_A, 1-g_A\}$ to denote the value of $\Gamax$.
Furthermore, define the random variable $T$ as follows:
\[
 T = \begin{cases}
      0 & G_A > 1/2, G_B < 1/2 \text{ or } G_A < 1/2, G_B > 1/2 \\
      1 & \text{otherwise},
     \end{cases}
\]
Intuitively, $T=0$ means that Alice and Bob are interested in the same part of 
the cake, therefore only one of them can be happy at the same time. More 
precisely, it is easy to check that , if Alice plays $\risky$ and 
Bob plays $\selfish$, when $T=0$ Alice receive $\Gamin$ and Bob receives 
$\Gbmax$, while when $T=1$, Alice receives $\Gamax$ and Bob receives $\Gbmax$.

Before continuing the discussion, let us clarify our model with an example. For 
instance, exploiting the notation of Section~\ref{sec:dc-implicit}, assume that
$\vaset = \vbset =  \{ 
\cream, \chocolate \}$. Also define Alice's actions as $\risky = \full$ and 
$\nonrisky = \half$. It is evident that in this example, with probability 1 we 
have
\[
 \Gamax = 1 \qquad \Gamin = 0 \qquad \Gbmax = 1 \qquad \Gbmin = 
0.
\]
Also we have
\[
 T = \begin{cases}
      1 & \va = \chocolate, \vb = \cream \quad \text{or} \quad \va = \cream, 
\vb = \chocolate \\
      0 & \va = \vb = \chocolate \quad \text{or} \quad \va = \vb = \cream.
     \end{cases}
\]

We will denote Alice's and Bob's strategies by $\Sa$ and $\Sb$, 
respectively. 
Note that strategies, which are assumed to be behavioral, are nothing but 
probability distributions assigning probabilities to each action at each node, 
based on one party's observations up to that time. Therefore, we use 
probabilistic notations for strategies; for instance, $\Sa(a_3|a^2,b^2,g_A)$ 
is the probability that Alice chooses action $a_3$ at stage 3 of the game when 
$G_A = g_A$ and during stages 1 and 2, players had played 
action sequences $a^2=(a_1, a_2) ,b^2=(b_1, b_2)$.

\subsubsection{Main Result}\label{sec:THPMR}
\textbf{Identification of a strategy:}

One strategy for Bob is to always plays selfishly. We denote this 
strategy by $\Sb$.

One possible strategy for Alice in $n$-stage game, denoted by $\Sa$ is as 
follows. Intuitively speaking, Alice guesses that Bob usually plays selfishly 
in order to maximize his own gain greedily. Based 
on this assumption, she counts the 
number of times she has risked so far and calculates the number of failures (with 
gains less than half) and number of successes (gains more than half) among 
them. If the 
number of successes is more, she guesses that $T=1$ and therefore 
continues to risk. However, if she has failed most of the times, she stops 
risking and guarantees 
herself a gain of $1/2$ by playing $\nonrisky$. To be more specific, at stage $k$, Alice 
calculates
\begin{equation}
\label{eq:ng-nl}
 \begin{split}
  n_g(k) &= \left | \{ 1 \leq t \leq k-1 : \gawrt{a_t, b_t}{g_A} >
1/2 \} \right | \\
n_l(k) &= \left | \{ 1 \leq t \leq k-1 : \gawrt{a_t, b_t}{g_A} <
1/2 \} \right |,
 \end{split}
\end{equation}
then she plays risky at stage $k$ if $n_g(k) \geq n_l(k)$, and plays non-risky 
otherwise. Note that Alice always plays $\risky$ at stage 1, since at this time 
$n_l(1) = n_g(1) = 0$.

\textbf{Conditions for being an equilibrium:}

Now assume that in a one stage game, making risk is advantageous for Alice for every value of $g_A$. This means that her expected gain when playing $\risky$ is 
greater than her expected gain when playing non-risky which is $1/2$, i.e.\
\[
p(T=0|G_A = g_A) \gamin + p(T=1|G_A = g_A) \gamax > 1/2 \qquad 
\forall \, 
g_A.
\]

We want to show that the above strategy is an equilibrium in the $n$ stage 
game. In repeated games (and more generally, in extensive games) instead of 
taking Nash Equilibrium as the solution, usually Sequential Equilibrium is 
considered \cite{osborne}. We consider even a stronger equilibrium criteria 
named Trembling Hand Perfect Equilibrium (THP).
In this equilibrium, it is assumed that at each stage, the hand of each player 
might ``tremble'' and he deviates from what he is supposed to do given his 
strategy. Then, a strategic game is defined from the extensive game by 
associating a different player for every player at each stage, we call each of 
these (pseudo players) an ``agent''. Then the strategy which is going to be 
proved to be THP is fixed for one agent and the trembled strategy is considered 
for other agents. This strategy should be a best response for all agents when 
trembling probability goes to zero. For more information about THP, see 
\cite{osborne}. 

Now we are ready to state our main result
\begin{thm}
 Assume that Alice is better off playing $\risky$ at the one stage game, i.e.\
\begin{equation}
\label{eq:one-stage-risk}
p(T=0|G_A = g_A)\gamin + p(T=1|G_A = g_A) \gamax > 1/2 \qquad \forall \, 
g_A.
\end{equation}
Also we assume that $G_A \neq 1/2$ with probability 1
and $p(T=0|g_A), p(T=1|g_A) \neq 0$ for all $g_A$.
Then the strategy pair $\Sab = (\Sa,\Sb)$ described above in the $n$ stage 
game, $n \geq 1$, is 
both Nash Equilibrium and Trembling Hand Perfect Equilibrium (THP). 
Furthermore, since the game is with perfect recall, this strategy also yields a 
Sequential Equilibrium.
\end{thm}

Note that since THP is stronger than Nash Equilibrium in games with perfect 
recall, it suffices to prove 
THP. The proof of this theorem is given in Section \ref{sec:app-THP}.

%
\section{Adjusted Winner}\label{sec=AW}
Assume two parties, say Alice and Bob, are about to divide a set of $m$ goods.
Unlike the Divide and Choose method, they announce their valuations over these
goods which are nonnegative vectors of sum $1$ and size $m$,
$\av=(a_1,\dots,a_m)$
for Alice and $\bv=(b_1,\dots,b_m)$ for Bob to a third party whose duty is to
divide these
items fairly based on these announced valuations. Adjusted Winner is an
algorithm that solves a sequence of equations in order to give a
division
of the items which is equitable, envy free and efficient \cite{brams:cake:1996}.
We note that the
divide and choose method does not have these properties.

The adjusted winner algorithm divides the items as follows. Reorder the items
so that,
\begin{equation}
\frac{a_1}{b_1} \geq \frac{a_2}{b_2} \geq \dots \geq \frac{a_t}{ b_t} \geq 1 
>\frac{a_{t+1}}{b_{t+1}}
\geq \dots \geq \frac{a_m}{b_m}.
\end{equation}
Then give items $1$ through $t$ to Alice and items $t+1$ through $m$ to Bob. If
their gains at this step is equal, the job is finished. First assume Alice's
gain is more. In this case, give a portion of item $t$ to Bob so that their gain
becomes equal. If even by giving all of item $t$ this did not happen, go for
item $t-1$ and continue this procedure until the equality holds. For the
second case when Bob's gain is more, in a similar
way, give a portion of item $t+1$ to Alice to achieve equality. If this was not
sufficient, go to item $t+2$ and continue. Since eventually by giving all the
items to the party with less gain, his gain becomes more, at some point in
between their gains become equal and the procedure terminates.

\subsection{Spying in Adjusted Winner}
A motivation for studying spying in the adjusted winner game is a result by Brams and Taylor who showed
that in the case of having two items, a dishonest party who has \emph{full} information
about the other party's
valuation vector, while the other party is unaware of this and acts honestly (\emph{i.e.,} the other party is spying),
can trick the referee \cite{brams:cake:1996}. We are interested in quantifying the benefit \emph{partial} spying, by assuming that Bob announces his valuation honestly
while
Alice uses the partial
information he has gained by spying over Bob's valuation to trick the referee
and announce an untrue valuation instead of her true
valuation.

\subsubsection{Reduction to Divide and Choose}
Assume that the valuation vectors $\av=(a_1,\dots,a_m)$
and $\bv=(b_1,\dots,b_m)$  are limited to be taking values in finite sets,
and Alice can spy any arbitrary function of
Bob's valuation vector consistent with her spying rate.
The assumption that the set of valuations is finite is a practical assumption
since we can assume that the value assigned to an item by each individual is a
real number with finite precision.
Therefore the set of all valuation vectors is finite. 

We can find the trade off between the
``\emph{spying rate}" and Alice's ``\emph{spying gain}" via a simple
transformation from adjusted winner problem (AW) to the divide and choose (DC) problem as follows: 
Let Alice's
announced valuation, $\ta$, plays the role of the division $D$ in divide and
choose and the following gain functions could be defined,
\begin{equation}
 \begin{split}
  \gainwrt{A}{\tav}{\av}{\bv} &= \awwrt{A}{\tav,\bv} \cdot \av , \\
  \gainwrt{B}{\tav}{\av}{\bv} &= \awwrt{B}{\tav,\bv} \cdot \bv ,\\
 \end{split}
\end{equation}
where $\cdot$ is the inner product operation. 
Note that although the two problems have conceptual differences, by using this
transformation, we can consider this problem a special case of divide and
choose.
Also note that in this approach, the assumption of having two items, $m=2$, is not necessary.

\subsubsection{Restriction on the structure of spying \label{sec:binary-aw}}
The drawback of the above approach is that Alice is allowed to spy a complicated function of Bob's valuation vector.  Let us restrict Alice to be able to only spy a set of simpler (but more realistic)  questions  of
the form ``Is
Bob's valuation on the first item less than a particular value $\alpha$ or more
than that?''. We call these binary dividing questions. To study this problem, as in  Brams and Taylor's work, we assume that the number of
items, $m$, is equal to $2$.  In this case
the valuations are $\av=(a,1-a)$ for Alice and $\bv=(b,1-b)$ for Bob.
Therefore we can simply take $a$ and $b$ as \emph{valuation numbers} or more
simply valuations. We assume that Alice's valuation is fixed, while Bob's
valuation of the first item
is uniformly distributed in an interval $b\in[\bmin,\bmax]$. A binary dividing question divide the interval $[\bmin,\bmax]$ into two subintervals.  A sequence of $k$ dividing questions is like doing a dictionary search and can be represented in terms of cutting points $b_0< b_1< \cdots< b_{2^k}$ where  $b_0=\bmin$ and $b_{2^k} = \bmax$.

Let $\dgamax{k}{\bmin,\bmax}$ denote the maximum increase in Alice's gain by asking the optimal $k$ binary division questions. Clearly $\dgamax{0}{\bmin,\bmax}=0$ and $\{\dgamax{k}{\bmin,\bmax}\}_{k=0,1,2,\cdots}$ is a non-decreasing sequence of numbers because the more questions, the better Alice can perform. Since the valuation vectors add up to one, the maximum possible utility of Alice is one, and hence the increase in Alice's utility by $k$ questions, $\dgamax{k}{\bmin,\bmax}$, is also bounded from above by one. 

We ask the following questions:
\begin{enumerate}
\item[Q1.] Does the value of spying questions depreciate over time? For instance, is it true that the extra gain we get by asking the third question is less than the gain we get by asking second question, i.e.,  $$\dgamax{2}{\bmin,\bmax}-\dgamax{1}{\bmin,\bmax}\overset{?}{\geq}  \dgamax{3}{\bmin,\bmax}-\dgamax{2}{\bmin,\bmax}.$$
In other words, is $\{\dgamax{k}{\bmin,\bmax}\}_{k=0,1,2,\cdots}$ a concave sequence of numbers?
\item[Q2.] Assuming that the answer to the above question cannot be answered affirmative in all cases, can we find a linear control on the growth of $\dgamax{k}{\bmin,\bmax}\leq k\tilde{\Delta}$, with $\tilde{\Delta}$ as small as possible?
\end{enumerate}

Before answering Q1 and Q2, let us discuss one of their implications. Let us assume for a moment that instead of playing a single Adjusted Winner game, we are playing $n$ repetitions of the game. More specifically, assume that during $n$ games, Bob's
valuation is i.i.d.\ random variables uniformly distributed in $[\bmin,\bmax]$
and $a$ is fixed in all games. If Alice is allowed to ask $R$
questions on average in each
game, or totally $nR$ questions, we are interested in finding bounds for the
Alice's expected improvement in gain \emph{averaged} over $n$ games. If the answer to Q1 is affirmative, then 
 the strategy of spying either $\lfloor R\rfloor$ or
$\lceil R\rceil$
questions in each game (with the average number of questions no larger than R)
maximizes the spying gain of Alice. To see this in the case of $R$ being an integer, observe that 
if Alice asks $t_i$ questions in game $i$, then her average
improvement in gain will be at most,
\begin{equation}
 \frac{1}{n} \sum_{i=1}^n \dgamax{t_i}{\bmin,\bmax} \leq \dgamax{\frac{1}{n}
\sum t_i}{\bmin,\bmax} \leq \dgamax{R}{\bmin,\bmax},	
\end{equation}
where we have used the concavity of the sequence. If the answer to Q2 is affirmative, then  the average improvement on Alice's
expected gain which is averaged over $n$ games is bounded by
 $R \tD$ since
if Alice
asks $t_i$ questions in the $i$th game, her maximum improvement is,
\begin{equation}
  \frac{1}{n} \sum_{i=1}^n \dgamaxhead_{t_i} \leq \frac{1}{n}\sum_{i=1}^n t_i
\tD = R\tD,
\end{equation}

For Q1, we have the following result:
\begin{thm}
\label{thm:ocz}
 For a fixed $1/2\leq\bmin<\bmax\leq 1$, if $a$ is outside the interval
$(\tau_l,\tau_u)$ where 
\begin{equation}
 \begin{split}
  \tau_u(\bmin,\bmax) &= \frac{2 \bmax^2+2 \bmax
\bmin}{\bmax+3 \bmin}, \\
  \tau_l(\bmin,\bmax) &= \frac{2 \bmax \bmin+2
\bmin^2}{3 \bmax+\bmin},
 \end{split}
\end{equation}
then the sequence $\left \{ \dgamax{k}{\bmin,\bmax} \right \}_{k=0}^\infty$
 is concave. Furthermore, the optimal cutting points $\{b_0, b_1, \cdots, b_{2^k}\}$ forms a geometric progression. More specifically, having asked $i-1$ questions and having ended up with $b\in[x, y]$ as the interval of Bob's valuation, Alice should ask whether $b\in [x, \sqrt{xy}]$ or $b\in[\sqrt{xy},y]$ as her $i$-th question.
\end{thm}
\begin{remark}The assumption $1/2\leq\bmin<\bmax\leq 1$ 
means that Alice knows which of
the two items Bob likes more, but she does not know his exact valuation. Also
note that the case which $\bmin<\bmax<1/2$ (the entire interval falls in the left
half) could be
reduced to this case by changing the order
of items.
\end{remark}
For Q2, we have the following result:
\begin{thm}
\label{thm:general-bound-differentiableNew}
 Assume that for $\bmin\leq x < y \leq \bmax$, $\dgamax{1}{x,y}$ is
differentiable with respect to $y$. Then $\tD(x,y)$ defined by 
\begin{equation}
 \tD = \max_{\bmin\leq \gamma_1 \leq \gamma_2 \leq \bmax} \Lambda(\gamma_1,\gamma_2),
\end{equation}
where $ \Lambda(x,y)=\frac{\partial}{\partial
y} \Gamma(x,y)$ defined as follows
\begin{equation}
 \Gamma(x,y) = \begin{cases}
                (y-x) \dgamax{1}{x,y} & y>x, \\
                0 & y=x,
               \end{cases}
\end{equation}
satisfies
$$\dgamax{k}{\bmin,\bmax}\leq k\tilde{\Delta}, \qquad \forall k.$$
\end{thm}


%
\section{Maximum Nash Collective Utility Function}\label{sec=MNC} 
In this section we consider an arbitrary society with a government who wants to
divide its several resources among the citizens. Each person assigns a value for
each of the resources available to the government, and we assume that the
government knows these valuations. The Nash collective utility function (Nash 
CUF) 
for a given division strategy is equal to the product of the gains of individual 
members of the society of that division strategy. Maximizing the Nash CUF for 
this society
implies a division policy for the government, specifying how much of each
resource should be allocated to each individual. For practical reasons the
government may want to divide the citizens into several clusters, say drivers,
teachers, etc, and apply the same division strategy uniformly to all people from
the same class. We consider the increase of Nash CUF
for a clustering refinement and draw conceptual links between this problem and
the portfolio
selection problem in stock
markets \cite{thomas1991elements}. 
\subsection{The model}
 
Assume that the population of the society is $n$, which is fixed.
The valuation vectors of all the individuals in the society is known to the
government. We assume that the government has partitioned the society into $k$
clusters $\mathcal{P}=(\mathcal{P}_1,\dots,\mathcal{P}_k)$. Let $n_i$ denote the
number of
people in cluster $\mathcal{P}_i$ and
$\alpha_i=\frac{n_i}{n}$. The government has decided to use a fixed division
strategy for all people in cluster $i$ which is denoted by $\bv_i$.
The sum of the portion each
individual receives should sum up to one, i.e.\
$\sum_{i=1}^k n_i \bv_i=\textbf{1}=\sum_{i=1}^k(n\alpha_i) \bv_i=1$. Let us
denote the the
valuation vector of people in cluster $i$ by $\vv_{i1}, \vv_{i2},\dots,
\vv_{in_i}$.

Based on the valuation vectors of individuals, the government wants to divide 
the
items so as to maximize the Nash CUD of the society, which is
\begin{equation}
W_{\mathcal{P}}=\max_{\bv_{1:k}}\prod_{i}\prod_{j=1}^{n_i}\bv_i^t \vv_{ij}.
\end{equation}

In the second scenario, the government divides one of the classes, say the
first class, into two subclasses $1a$ and $1b$ and uses different division
protocols for these subclasses. If $\mathcal{P}'$ denotes the new
partitioning and $W_{\mathcal{P}'}$ to be the maximum
Nash CUF in the new scenario,
\begin{equation}
 W_{\mathcal{P}'} = \max_{\bv'_{1a},\bv'_{1b}, \bv'_{2:k}} \prod_{\vv_{1a}}
\bv^{'t}_{1a}
\vv_{1a}
\prod_{\vv_{1b}} \bv^{'t}_{1b} \vv_{1b} \prod_{i=2}^k \prod_{j=1}^{n_i}
\bv^{'t}_i
\vv_{i,j}.
\end{equation}
By taking $\bv'_{1a} =\bv'_{1b} = \bv_1$ and $\bv'_i = \bv_i$ for $i>1$, we
realize that
$ W_{\mathcal{P}'}\geq  W_{\mathcal{P}}$. In fact by refining the
classification, the government can improve the social welfare, which was
expected. In this section, we are interested in finding an upper bound on the
possible improvement after this refinement.

Define $\Vv_i$ to be the random variable whose distribution is the
\emph{empirical
distribution} of the valuation vector of people in class $i$, i.e. for any set
$\mathcal{A}$
\begin{equation}
P(\Vv_i\in \mathcal{A})=\frac{|\#j: \vv_{ij}\in \mathcal{A}
|}{n_i},
\end{equation}
also define r.v.'s
$\Vv_{1a}$ and $\Vv_{1b}$ to be the random variables for \emph{empirical
distribution} of
subclasses $1a$ and $1b$. Values of $\alpha_{1a}$ and $\alpha_{1b}$ are defined
in a natural way by dividing the size of classes $1a$ and $1b$ to $n$. Note that
$$p(\Vv_{1}=\vv_{1})=\frac{\alpha_{1a}}{\alpha_1}
p(\Vv_{1a}=\vv_{1})+\frac{\alpha_{1b}}{\alpha_1} p(\Vv_{1b}=\vv_{1} ).$$
We can define a random variable $E$ indicating where a randomly chosen person
from class $1$ belongs to $1a$, or to $1b$. In this case
$p(E=0)=\alpha_{1a}/\alpha_1$ and
$p(E=1)=\alpha_{1b}/\alpha_1$. Also $p(\Vv_{1}=\vv_{1}|E=0)=p(\Vv_{1a}=\vv_{1})$
and
$p(\Vv_{1}=\vv_{1}|E=1)=p(\Vv_{1b}=v_{1})$, which is simply the Bayes rule. We
denote the support of $\Vv_1$ by the set $\mathcal{V}_1$ (i.e.
$p(\Vv_1=v_1)>0~~\iff
\vv_1\in \mathcal{V}_1$). Similarly we let $\mathcal{V}_{1a}$ and
$\mathcal{V}_{1b}$ to be the support of $\Vv_{1a}$ and $\Vv_{1b}$. Note that
$\mathcal{V}_{1a}\subset \mathcal{V}_{1}$ and $\mathcal{V}_{1b}\subset
\mathcal{V}_{1}$.

In a more generalized but similar case, we can assume that instead of dividing
cluster $\mathcal{P}_1$ into $2$ clusters, we divide it into $t$ clusters
$\mathcal{P}_{1,1},\dots, \mathcal{P}_{1,t}$ and show the new partitioning by
$\mathcal{P}'$. Exactly in the same way, we
define random variables $E$ and $\Vv_1$.

\begin{remark}
If we define the distance between two non-negative vectors
$\textbf{v}$ and $\textbf{w}$ by $-\log(\textbf{v}\cdot\textbf{w})$, we can
reexpress the problem of finding an optimal $k$-clustering as an unsupervised
clustering of the valuation vectors of the whole society into $k$ clusters.
Vectors $\textbf{b}_i$ will be the centers
which need to satisfy
$\sum_{i=1}^kn_i\textbf{b}_i=\textbf{1}$. Note that
$-\log(\textbf{v}\cdot\textbf{w})$ is not a metric.
\end{remark}

Our result is the following:
\begin{thm}
\label{thm:nash-main}
 With the above notations, if we refine the clustering $\mathcal{P}$ by
dividing cluster $\mathcal{P}_1$ into $t$ clusters resulting in a new
clustering $\mathcal{P}'$, we have,
\begin{equation}
 W_{\mathcal{P}} \leq W_{\mathcal{P}'} \leq W_{\mathcal{P}} 2^{n_1 I(\Vv_1;E)},
\end{equation}
\end{thm}

\begin{remark}
Since $\mathcal{P}'$ is a refined version of $\mathcal{P}$, the
lower bound on $W_{\mathcal{P}'}$ is expected. To intuitively understand the
upper
bound, note that a good clustering of $\mathcal{P}_{1}$ puts valuation vectors
that are geometrically close to each other into the same cluster. Therefore
knowing that a person is in a certain cluster $\mathcal{P}_{1,E}$ for some $E$
should provide some information about the geometrical location of the valuation
vector of the person. Thus $I(\Vv_1;E)$ is large for a good clustering. However
a large $I(\Vv_1;E)$ does not necessarily imply a good clustering.
Such information theoretic interpretation of clustering (traditionally a topic
of data mining and machine learning) may be new (we have not seen it) and it may
be of independent interest.
\end{remark}

\begin{remark}
 The distribution of $p(\textbf{v}_1)$ is the empirical distribution of
 valuation vectors in $\mathcal{P}_1$ and $p(\textbf{v}_1|E=e)$ is the empirical
 distribution of valuation vectors in $\mathcal{P}_{1,e}$. $I(\Vv_1;E)$ is
 computable from these empirical distributions.
\end{remark}

\begin{proof}[Proof of Theorem~\ref{thm:nash-main}]
As we have already discussed, $W_{\mathcal{P}'}\geq W_{\mathcal{P}}$, and it
remains to prove the other
side. For simplicity, we assume that $t=2$; the case of $t>2$ is similar. Note that maximizing $W_{\mathcal{P}}$ is equal to maximizing
\begin{align*}\frac{1}{n}\log
W_{\mathcal{P}}=&\max_{\bv_{1:k}}\sum_{i}\frac{1}{n}\sum_{j=1}^{n_i}\log(\bv_i^t
\vv_{ij})=\max_{ \bv_ {
1:k}}\sum_{i}\alpha_i\frac{1}{n_i}\sum_{j=1}^{n_i}\log(\bv_i^t \vv_{ij})
\\=&\max_{\bv_{1:k}}\sum_{i}\alpha_i\mathbb{E}\left[\log(\bv_i^t \Vv_i)\right],
\end{align*}
Similarly,
\begin{equation*}
 \begin{split}
\frac{1}{n}\log W_{\mathcal{P}'}=\max_{\bv'_{1a}, \bv'_{1b},
\bv'_{2:k}} & \Big
(
\alpha_{1a}\mathbb{E}\left[\log(\bv_{1a}^{'t}\Vv_{1a})\right]+\alpha_{1b}\mathbb
{E}
\left[\log(\bv_{1b}^{'t}\Vv_{1b})\right]\\
& +\sum_{i=2:k}\alpha_i\mathbb{E}\left[\log(\bv_i^{'t}\Vv_i)\right] \Big ).  
 \end{split}
\end{equation*}
The above equations show that Nash Collective Utility is equivalent to the mathematics of
portfolio
selection problem in stock markets (see \cite{thomas1991elements}), if we interpret division
into groups as the equivalent to side information. With this in mind, the benefit of adding extra groups
equivalent to adding additional side information. If we show the original divisions into groups (previous side information) by rv $E'$, the new division into groups will be corresponding to random variables $(\tilde E, E')$ where we assume that $\tilde E$ is a constant when $E'\neq 1$, and $\tilde E=E$ when $E'=1$. Thus, from the increase in the exponent of the growth \cite{thomas1991elements}, we get
$$\frac{1}{n}\log
W_{\mathcal{P}}-\frac{1}{n}\log W_{\mathcal{P}'}\leq I (V; \tilde E|E')=\alpha_1 I(V;\tilde{E}|E'=1)=\alpha_1 I(V;E)=\frac{n_1}{n}I(V;E).$$

\end{proof}


\section{Proofs for explicit Divide and Choose}

\subsection{Proof for one-shot spying 
}\label{sec:dc-achievability-one-shot}

\begin{proof}[Proof of Theorem \ref{thm:one-shot}]

We employ the technique of \cite{yassaee2013} to provide a lower bound on the expected utility of Alice after spying about Bob's valuation. Let $\mc=\{U(j)\}_{j=1}^{\sj}$ be a random product codebook, in which  are drawn independently from $q_{U}$. 
 Let $\mb:[1:\sj]\mapsto[1:\sm]$ be a random mapping (binning), mapping each element of $[1:\sj]$ uniformly and independently to the set $[1:\sm]$.   

We draw an index $j\in[1:\sj]$ with the probability
\[
P_{\enc}(j|v_B)=\dfrac{2^{\im_q(v_B;U(j))}}{\sum_{\tilde{j}}2^{\im_q(v_B;U(\tilde{j}))}}=
\dfrac{q_{V_B|U}(v_B|U(j))}{\sum_{\tilde{j}}q_{V_B|U}(v_B|U(\tilde{j}))}.
\]  
where we have use capital $P$ since the above conditional pmf is itself random (due to the random codebook generation). Then, we transmits $m=\mb(j)$ to Alice as the spying content. 
Alice uses $m$ to draw $\hat j$ from the following pmf
 \[
P_{\dec}(\hat{j}|m, v_A)=\dfrac{2^{\im_q(v_A;U(\hat{j}))}\mathbf{1}\{B(\hat{j})=m\}}{\sum_{\bar{j}}2^{\im_q(v_A;U(\bar{j}))}\mathbf{1}\{B(\bar{j})=m\}},
\]
where  $\mathbf{1}[\cdot]$ is the indicator function. 
Finally, Alice produces $D$ from $q_{D|UV_A}(d|U(\hat{j}),v_A)$.

Observe that the joint distribution of random variables factors as,
\begin{align*}
P(v_A, v_B, m,j,\hat{j},d)=q(v_A, v_B)P_{\enc}(j|v_B)\mathbf{1}\{B({j})=m\}P_{\dec}(\hat{j}|m, v_A)q(d|U(\hat{j}),v_A).
\end{align*}
We now compute the expected value of Alice's gain over the random codebook and random binning:
\begin{align}\mathbb{E}&[\gainwrt{A}{D}{V_A}{V_B}]= \e_{\mc,\mb}\sum_{v_A,v_B,d,m,j,\hat{j}}P(v_A, v_B, m,j,\hat{j},d)\gainwrt{A}{d}{v_A}{v_B}\\&\geq
 \e_{\mc,\mb}\sum_{v_A,v_B,d,m,j}P(v_A, v_B, m,j,\hat{J}=j,d)\gainwrt{A}{d}{v_A}{v_B}
\nonumber\\&= \e_{\mc,\mb}\sum_{v_A,v_B,d,m,j}q(v_A, v_B)P_{\enc}(j|v_B)\mathbf{1}\{B({j})=m\}P_{\dec}({j}|m, v_A)q(d|U({j}),v_A)\gainwrt{A}{d}{v_A}{v_B}\nonumber
\\&= \e_{\mc,\mb}\sum_{v_A,v_B,d}\sm \sj \cdot q(v_A, v_B)P_{\enc}(1|v_B)\mathbf{1}\{B({1})=1\}P_{\dec}({1}|1, v_A)q(d|U(1),v_A)\gainwrt{A}{d}{v_A}{v_B}\label{eqn:symmetryuse}
\\&= \e_{\mc,\mb}\sum_{v_A,v_B,d}\sm \sj \cdot q(v_A, v_B)\dfrac{2^{\im_q(v_B;U(1))}}{\sum_{\tilde{j}}2^{\im_q(v_B;U(\tilde{j}))}}\dfrac{2^{\im_q(v_A;U(1))}\mathbf{1}\{B(1)=1\}}{\sum_{\bar{j}}2^{\im_q(v_A;U(\bar{j}))}\mathbf{1}\{B(\bar{j})=1\}}q(d|U(1),v_A)\gainwrt{A}{d}{v_A}{v_B}\nonumber
\\&= \sum_{v_A,v_B,d}\sm \cdot q(v_A, v_B)
\e_{U(1), B(1)} \e_{\mc,\mb|U(1), B(1)}
\dfrac{\sj 2^{\im_q(v_B;U(1))}}{\sum_{\tilde{j}}2^{\im_q(v_B;U(\tilde{j}))}}\dfrac{2^{\im_q(v_A;U(1))}\mathbf{1}\{B(1)=1\}}{\sum_{\bar{j}}2^{\im_q(v_A;U(\bar{j}))}\mathbf{1}\{B(\bar{j})=1\}}\times\nonumber
\\&\qquad\qquad\times q(d|U(1),v_A)\gainwrt{A}{d}{v_A}{v_B}\nonumber
\\&\geq \sum_{v_A,v_B,d}\sm \cdot q(v_A, v_B)
\e_{U(1), B(1)}
\dfrac{\sj 2^{\im_q(v_B;U(1))}}{ \e_{\mc,\mb|U(1), B(1)}\sum_{\tilde{j}}2^{\im_q(v_B;U(\tilde{j}))}}\dfrac{2^{\im_q(v_A;U(1))}\mathbf{1}\{B(1)=1\}}{ \e_{\mc,\mb|U(1), B(1)}\sum_{\bar{j}}2^{\im_q(v_A;U(\bar{j}))}\mathbf{1}\{B(\bar{j})=1\}}\times\nonumber
\\&\qquad\qquad\times q(d|U(1),v_A)\gainwrt{A}{d}{v_A}{v_B}\label{eqn:jensenxy}
\\&\geq \sum_{v_A,v_B,d}\sm \cdot q(v_A, v_B)
\e_{U(1), B(1)}
\dfrac{\sj 2^{\im_q(v_B;U(1))}}{\sj-1+ 2^{\im_q(v_B;U(1))}}\dfrac{2^{\im_q(v_A;U(1))}\mathbf{1}\{B(1)=1\}}{ (\sj-1)\sm^{-1}+
2^{\im_q(v_A;U(1))}\mathbf{1}\{B(1)=1\}}\times\nonumber
\\&\qquad\qquad\times q(d|U(1),v_A)\gainwrt{A}{d}{v_A}{v_B}\label{eqn:jensenxy-after}
\\&= \sum_{v_A,v_B,d} q(v_A, v_B)
\e_{U(1)}
\dfrac{\sj 2^{\im_q(v_B;U(1))}}{\sj-1+ 2^{\im_q(v_B;U(1))}}\dfrac{2^{\im_q(v_A;U(1))}}{ (\sj-1)\sm^{-1}+
2^{\im_q(v_A;U(1))}}q(d|U(1),v_A)\gainwrt{A}{d}{v_A}{v_B}\nonumber
\\&= \sum_{v_A,v_B,d, u} q(v_A, v_B, u)
\dfrac{\sj}{\sj-1+ 2^{\im_q(v_B;u)}}\dfrac{2^{\im_q(v_A;u)}}{ (\sj-1)\sm^{-1}+
2^{\im_q(v_A;u)}}q(d|u,v_A)\gainwrt{A}{d}{v_A}{v_B}\label{eqn:ewithU}
\\&= \e\big[
\dfrac{1}{1-\sj^{-1}+\sj^{-1} 2^{\im_q(V_B;U)}}\dfrac{1}{1+ (\sj-1)\sm^{-1}2^{-\im_q(V_A;U)}
}\gainwrt{A}{D}{V_A}{V_B}\big],\nonumber
\end{align}
where equation \eqref{eqn:symmetryuse} is due to the symmetry of the random codebook generation; equation \eqref{eqn:jensenxy} follows from Jensen's inequality for the jointly convex function $f(x_1, x_2)=1/(x_1x_2)$;  equation \eqref{eqn:jensenxy-after}  follows from the fact that for any $\tilde{j}\neq 1$
\begin{align*}\e_{\mc,\mb|U(1), B(1)}2^{\im_q(v_B;U(\tilde{j}))}=\sum_{u}q(u)\frac{q(v_B,u)}{q(v_B)q(u)}=1.\end{align*}
Similarly for any $\bar{j}\neq 1$
\begin{align*}\e_{\mc,\mb|U(1), B(1)}2^{\im_q(v_A;U(\bar{j}))}\mathbf{1}\{B(\bar{j})=1\}=\sum_{u}q(u)\frac{q(v_A,u)}{q(v_A)q(u)}\sm^{-1}=\sm^{-1}.\end{align*}
Equation \eqref{eqn:ewithU} follows from the fact that $\e_{U(1)}$ is over $q(u)$ and $$2^{\im_q(v_B;u)}=\frac{q(u|v_B)}{q(u)}$$ and $ q(v_A, v_B,u)= q(v_A, v_B)q(u|v_B)$.

The derivation of Bob's gain is similar. Note that the spying strategy of Alice is randomized here. Based on a private randomness (the random codebook and random binning), she decides her spying function. 

Deriving the loosened bound is as follows:
\begin{align*}&\e\big[
\dfrac{1}{1+\sj^{-1} 2^{\im_q(V_B;U)}}\dfrac{1}{1+ \sj\sm^{-1}2^{-\im_q(V_A;U)}
}\gainwrt{A}{D}{V_A}{V_B}\big]
\\&\geq \e\big[
\dfrac{\mathbf{1}[\left\{\log\sj-\im(V_B;U)\ge\gamma,\ \im(V_A;U)-\log(\sj\sm^{-1})\ge\gamma\right\}
]}{(1+\sj^{-1} 2^{\im_q(V_B;U)})(1+ \sj\sm^{-1}2^{-\im_q(V_A;U)})
}\gainwrt{A}{D}{V_A}{V_B}\big]
\\&\geq \dfrac{1}{(1+2^{-\gamma})^2}\e\big[
\mathbf{1}[\left\{\log\sj-\im(V_B;U)\ge\gamma,\ \im(V_A;U)-\log(\sj\sm^{-1})\ge\gamma\right\}
]\gainwrt{A}{D}{V_A}{V_B}\big]
\\&\geq (1-3\times 2^{-\gamma})\e\big[
\mathbf{1}[\left\{\log\sj-\im(V_B;U)\ge\gamma,\ \im(V_A;U)-\log(\sj\sm^{-1})\ge\gamma\right\}
]\gainwrt{A}{D}{V_A}{V_B}\big]
\\&= (1-3\times 2^{-\gamma})\e\bigg[\big(1-
\mathbf{1}[\left\{\log\sj-\im(V_B;U)\ge\gamma,\ \mathsf{or}\ \im(V_A;U)-\log(\sj\sm^{-1})\ge\gamma\right\}
]\big)\gainwrt{A}{D}{V_A}{V_B}\bigg]
\\&\geq (1-3\times 2^{-\gamma})\e\big[
\gainwrt{A}{D}{V_A}{V_B}\big]-\bar g(1-3\times 2^{-\gamma})\mathsf{P}[\left\{\log\sj-\im(V_B;U)\ge\gamma,\ \mathsf{or}\  \im(V_A;U)-\log(\sj\sm^{-1})\ge\gamma\right\}]
\end{align*}
\end{proof}

\subsection{Proof of Theorem \ref{thm:DC-region}}\label{sec:dc-achievability}
For achievability, we use the existing results on the empirical
coordination, which is summarized as
follows. Assume two terminals have i.i.d.\~samples of random variables $X_1$ and $X_2$
with joint pmf $p(x_1,x_2)$, i.e.\ $p(x_1^n, x_2^n) = \prod_{i=1}^n
p(x_{1,i}, x_{2,i})$. The goal is
to simulate the channel $p(y_1, y_2|x_1,x_2)$ and generate $Y_1$ and $Y_2$ in
terminals $1$ and $2$ respectively. Since the first terminal has only
access to $X_1$ while $Y_1$ is dependent on both $X_1$ and $X_2$, which is the
same story for the second terminal, the two terminals need to communicate
with some rate in order to gain information about the other terminal so that
they can simulate the channel. This process could involve multiple rounds of communication in general. 
The two
terminals need to generate jointly typical sequences of $Y_1^n$ and
$Y_2^n$ with $X_1^n, X_2^n$ with high probability. Substituting $X_1$ by $V_A$,
$X_2$ by $V_B$, $Y_1$ by $D$ and $Y_2$ by a constant, say $0$, one can observe from the result of  \cite{yassaee2012} that empirical coordination holds if $R>I(V_B;U|V_A)$.  Using properties of typical sequences,
\begin{equation}
 |\tilde{G}_A - G_A| = |\frac{1}{n}\sum_{i=1}^n
\gainwrt{A}{D_i}{V_{A,i}}{V_{B,i}} - \ev{\gainwrt{A}{D}{V_A}{V_B}}| \leq
\delta G_A \leq \delta,
\end{equation}
since gains are bounded by $1$. The same inequality holds for $\tilde{G}_B$.
This proves the achievability.

The converse has much in common with the proof of the converse in
\cite{yassaee2012} by setting $X_1 = V_A$, $X_2 = V_B$, $D=Y_1$ and $Y_2 = 0$
in their terminology. Assuming an $(n,R)$ code with communication variable $C$ and division $D^n$, 
we define the auxiliary random variables $U$ and $D$ as follows: take
$Q$ to be a random variable independent from all other random variables and
uniformly distributed in $[1:n]$ and
\begin{equation}
 \begin{gathered}
  U = C V_{A_{[Q+1:n]}} V_{B_{[1:Q-1]}} Q, \\
  D =  D_Q,
 \end{gathered}
\end{equation}
note that since $Q$ is uniform and independent from all other random variables
and $V_A^n$ and $V_B^n$ are i.i.d.\ therefore $V_{A_Q} = V_A$ and $V_{B_Q} =
V_B$. Showing the Markov chains $V_A\rightarrow V_B\rightarrow U$ and $V_B\rightarrow V_AU\rightarrow D$ and the inequality $I(U;V_A|V_B)\leq H(C)/n$ 
are identical to that of \cite{yassaee2012} and thus omitted from here. Finally, note that 
\begin{equation}
 \begin{split}
  \left |\ev{\gainwrt{A}{D}{V_A}{V_B}} - G_A \right | &=
\left |\ev{\gainwrt{A}{D_Q}{V_{A,Q}}{V_{B,Q}}} - G_A\right | \\
&=\left | \frac{1}{n} \sum_{q=1}^{n} \ev{\gainwrt{A}{D_q}{V_{A,q}}{V_{B,q}}} -
G_A \right |\\
&= |\tilde{G}_A - G_A| < \delta.
 \end{split}
\end{equation}
Following a similar procedure $|\tilde{G}_B - G_B|<\delta$.

\subsection{Proof of Theorem \ref{thm:sharedrandomness-DC}}\label{APP:thm:sharedrandomness-DC}
\begin{proof}[Proof of Theorem \ref{thm:sharedrandomness-DC}] Let us fix some alphabet set $\mathcal{S}$ and pmf $p(s)$ for the shared randomness. Then, the strategy of Bob is $q(c|v_{B}^n,s)$ where $c$ is the message on the alphabet set $[1:2^{nR}]$. Alice's strategy is $q(d^n|v_{A}^n, c,s)$ where $d_i$ is the division by Alice in $i$-th game. For choosing $D_i$, Alice should only see how much it has learnt from $C, S$ and $V_{Ai}$ about $V_{Bi}$. This is due to the fact that the gain in the $i$-th stage depends only on the conditional pmf of  $V_{Ai}$ about $V_{Bi}$ given $C=c$ and $S=s$, \emph{i.e.,} only the marginal conditional distribution $q(v_{Ai}, v_{Bi}|c,s)$ matters. In other words, if we fix some action $q(c|v_B^n,s)$ for Bob, instead of the original identical distributions on $(V_{Ai}, V_{Bi})$, the two players play the cut and choose game with the modified pmfs $q_i(v_{A}, v_{B})=q(v_{Ai}, v_{Bi}|c,s)$ for $i\in[n]$. Observe that the Markov chain $C\rightarrow V_{B}^nS\rightarrow V_{A}^n$, independence of $S$ from $V_{B}^n, V_{A}^n$,  and the fact that $(V_{B}^n, V_{A}^n)$ are i.i.d., imply that $CS\rightarrow V_{Bi}\rightarrow V_{Ai}$ for $i\in[1:n]$. Hence,
$$q(v_{Ai}, v_{Bi}|cs)=q(v_{Bi}|cs)(v_{Ai}| v_{Bi}, cs)=q(v_{Bi}|cs)q(v_{Ai}| v_{Bi}).$$
Therefore, conditioning on $c$ and $s$ changes only the marginal distribution on $V_{Bi}$, and the conditional distribution of $V_{Ai}$ given $V_{Bi}$ is not affected.

Fix some action $p(c|v_B^n, s)$ for Bob. After revealing $C=c, S=s$, the marginal distribution of $(V_{Ai}, V_{Bi})$ reduces to $q(v_{Ai}, v_{Bi}|c,s)$. Alice plays her best response in the $i$-th game, and a rate pairs from  $\mathsf{G}(q(v_{Ai}, v_{Bi}|c,s))$ will occur.
We assume that Alice's choice of her best response is such that $$(G_A(q(v_{Ai}, v_{Bi}|c,s)), G^*_B(q(v_{Ai}, v_{Bi}|c,s))) \in \mathsf{G}(q(v_{Ai}, v_{Bi}|c,s))$$ occurs.  Then, let $r(c|v_B^n,s)$ be a maximizer of  
\begin{align}\max_{r(c|v_B^ns)}\frac 1n\sum_{i=1}^n\sum_{c,s}r(c,s)G^*_B(r(v_{Ai}, v_{Bi}|c,s))\label{eqn2gt},\end{align}
where the expression is computed with respect to $r(c,s,v_A^n,v_B^n)=r(c|v_B^n,s)q(v_A^n,v_B^n)p(s)$. Then $r(c|v_B^n,s)$ will lead to an equilibrium; Alice is always playing one of her best responses and Bob has maximized his payoff by choosing  $r(c|v_B^n,s)$. The gain of Alice will then become
\begin{align}\frac 1n\sum_{i=1}^n\sum_{c,s}r(c,s)G_A(r(v_{Ai}, v_{Bi}|c,s))\label{eqn2gtAlice}.\end{align}
It remains to show that the payoffs given in equations \eqref{eqn2gt} and \eqref{eqn2gtAlice} can be related to the ones given in the statement of the theorem in equation \eqref{eqn;rate;gain;pair}. We argue that the following two claims establish our desired result: Claim 1 is that for any $n$, any $p(s)$, and any $r(c|v_B^n,s)$ for $ |\mathcal C|\leq 2^{nR}$, we have that 
\begin{align}\frac 1n\sum_{i\in[n]}\sum_{c,s}r(c,s)G^*_B(r(v_{Ai}, v_{Bi}|c,s))\leq G_{Bmax}.\end{align} 
Therefore, the payoff of Bob in the Nash equilibriums that we have defined in  equation \eqref{eqn2gt} cannot exceed $G_{Bmax}$. Claim 2 is that
for any $\epsilon>0$, any arbitrary $q(\tilde c|v_B)$ where $I(\tilde C;V_B)= R-\epsilon$, we can find a sufficiently large $n$, shared randomness $p(s)$, and some $q(c|v_B^n,s)$ where $|\mathcal C|\leq 2^{nR}$ such that the gain of player $B$, 
$$Gain_B=\frac 1n\sum_{i=1}^n\sum_{c,s}q(c,s)G^*_B(q(v_{Ai}, v_{Bi}|c,s))$$
 satisfies
\begin{align}|Gain_B-\sum_{\tilde c}q(\tilde c)G^*_B(q( v_{A},  v_{B}|\tilde c))|\leq \epsilon,\label{eqn:eq5gt}\end{align}
and the gain of the player A,
$$Gain_A=\frac 1n\sum_{i=1}^n\sum_{c,s}p(c,s)G_A(q(v_{Ai}, v_{Bi}|c,s))$$
 satisfies
\begin{align}|Gain_A-\sum_{\tilde c}q(\tilde c)G_A(q( v_{A},  v_{B}|\tilde c))|\leq \epsilon.\label{eqn:eq5bgt}\end{align}
The above two claims prove our result. The reason is that by Claim 1, Bob can never expect to have a payoff larger than $G_{Bmax}$. Then if we choose some $q(\tilde c|v_B)$ where $I(\tilde C;V_B)\leq R-\epsilon/2$ such that $\sum_{\tilde c}q(\tilde c)G^*_B(q( v_{A},  v_{B}|\tilde c))$ is within $\epsilon/2$ of $G_{Bmax}$, Claim 2 shows that the resulting strategy of Bob will be within $\epsilon$ of $G_{Bmax}$. Therefore, it has to be an $\epsilon$-equilibrium from the perspective of Bob. Since Alice is always performing her best response cut (as we consider $G_A(q(v_{Ai}, v_{Bi}|c,s))$), she does not have any incentive to change her actions. Therefore, the payoffs given in equations \eqref{eqn2gt} and \eqref{eqn2gtAlice} can be made within $\epsilon$ distance of the rate pair given in equation \eqref{eqn;rate;gain;pair}.

\emph{Proof of Claim 1:}  Let $Q$ be a time-sharing variable, uniform on $[1:n]$ and independent of previously defined variables. Setting $\tilde{C}=(C,Q,S)$, $\tilde{V}_B=V_{BQ}, \tilde{V}_A=V_{AQ}$ we get that
$$\sum_{\tilde c}q(\tilde c)G^*_B(q(\tilde v_{A}, \tilde v_{B}|\tilde c))=\frac 1n\sum_{i\in[n]}\sum_{c,s}q(c,s)G^*_B(q(v_{Ai}, v_{Bi}|c,s)).$$
Furthermore, the joint pmf of $(\tilde{V}_A, \tilde{V}_B)$ is $q(v_A, v_B)$. Finally, the inequality $R\geq  I(\tilde C;\tilde V_{B})$ holds since \begin{align}nR&\geq H(C)\nonumber\\&\geq I(C;V_B^n|S)\nonumber
\\&=I(CS;V_B^n)\label{eqn:usseS}
\\&=\sum_{i=1}^nI(CS;V_{Bi}|V_{B1:i-1})\nonumber
\\&=\sum_{i=1}^nI(CSV_{B1:i-1};V_{Bi})\nonumber
\\&\geq \sum_{i=1}^nI(CS;V_{Bi})\nonumber
\\&=n\cdot I(CS;V_{BQ}|Q)\nonumber
\\&=n\cdot I(CSQ;V_{BQ})\nonumber
\\&=n\cdot  I(\tilde C;\tilde V_{B}),\nonumber\end{align}
where \eqref{eqn:usseS} follows from the independence of shared randomness $S^n$ from $V_B^n$.
\emph{Proof of Claim 2:} Let $(\tilde C^n, V_A^n, V_B^n)$ be $n$ i.i.d. repetitions according to $q(\tilde c, v_A, v_B)=q(\tilde c, v_B)q(v_A|v_B)$. Then clearly,
\begin{align}\sum_{\tilde c}q(\tilde c)G_A(q(v_{Ai}, v_{Bi}|\tilde c))&=\frac 1n\sum_{i=1}^n\sum_{\tilde c^n}q(\tilde c^n)G_A(q(v_{Ai}, v_{Bi}|\tilde c_i)),\label{eqnlast-2}
\\\sum_{\tilde c}q(\tilde c)G^*_B(q(v_{Ai}, v_{Bi}|\tilde c))&=\frac 1n\sum_{i=1}^n\sum_{\tilde c^n}q(\tilde c^n)G^*_B(q(v_{Ai}, v_{Bi}|\tilde c_i)).\label{eqnlast-1}
\end{align}
The problem is that the alphabet set of $\tilde C^n$ can be much larger than $2^{nR}=2^{n(I(V_B;\tilde C)+\epsilon)}$. In the rest of the proof, we show how to reduce the cardinality of the message to around $2^{n(I(V_B;\tilde C)+\epsilon)}$.

Let $B_1$ and $B_2$ be two random bin index of $\tilde C^n$ at rates $2^{n(H(\tilde C|V_B)-\frac{\epsilon}{2})}$ and $[1:2^{n(I(V_B;\tilde C)+\epsilon)}]$ respectively. Then, for almost all choices of the random binning, one can recover $\tilde C^n$ from $(B_1, B_2)$ with probability $1-\epsilon$ (Slepian-Wolf) via a deterministic mapping $q^{SW}( \tilde c^n|b_1,b_2)$ such that
\begin{align}\|q^{SW}( \tilde c^n|b_1,b_2)q(b_1, b_2)-q(\tilde c^n, b_1, b_2)\|\leq \epsilon.\label{eqn:fina1}\end{align} Further by the OSRB lemma \cite[Theorem 1]{OSRB}, $B_2$ is almost independent of $V_B^n$:
\begin{align}\|q(b_2,v_B^n)-q(b_2)q(v_B^n)\|\leq \epsilon.\label{eqn:fina2}\end{align}
Furthermore, because $B_1$ and $B_2$ are functions of $\tilde{C}^n$, we have $q(b_1,b_2,v_A^n, v_B^n)=q(b_1,b_2, v_B^n)q(v_A^n|v_B^n)$. From equation  \eqref{eqn:fina2}, we have
\begin{align}\|q(b_1,b_2, v_A^n, v_B^n)-q(b_2)q(v_A^n,v_B^n)q(b_1|v_B^nb_2)\|\leq \epsilon.\label{2epseqn-1}\end{align}
Then, from equations \eqref{eqn:fina1} and the fact that $q^{SW}( \tilde c^n|b_1,b_2)$ is a deterministic mapping, we get
\begin{align}\|q(b_1,b_2, \tilde c^n, v_A^n, v_B^n)-q(b_2)q(v_A^n,v_B^n)q(b_1|v_B^nb_2)q^{SW}( \tilde c^n|b_1,b_2)\|\leq 2\epsilon.\label{2epseqn}\end{align}

Let us assume the following alternative desirable scenario with $B_2$ as a shared randomness between the two players. Shared randomness $B_2$ is independent of $(V_A^n, V_B^n)$ and jointly distributed according to $q(b_2)q(v_A^n,v_B^n)$. Message $B_2$ is created by Bob from $(V_B^n, B_2)$ according to $q(b_1|v_B^nb_2)$ and sent to Alice; the rate of this message is $R$. Alice uses $B_1$ and $B_2$ to recover $\tilde{C}^n$ via the deterministic function $q^{SW}( \tilde c^n|b_1,b_2)$. The joint pmf induced by the alternative scenario will be $q(b_2)q(v_A^n,v_B^n)q(b_1|v_B^nb_2)q^{SW}( \tilde c^n|b_1,b_2)$, which is in $2\epsilon$ total variation distance of the original i.i.d.\ pmf by equation \eqref{2epseqn}. With probability $1-2\epsilon$, the two scenarios are not statistically distinguishable. Hence,
\begin{align}\label{eqnlast0}\bigg|\frac 1n\sum_{i=1}^n\sum_{b_1, b_2}q(b_1, b_2)G^*_B(q(v_{Ai}, v_{Bi}|b_1, b_2))-\frac 1n\sum_{i=1}^n\sum_{\tilde c^n}q(\tilde c^n)G^*_B(q(v_{Ai}, v_{Bi}|\tilde c))\bigg|\leq 2\epsilon \bar{g}_B,\end{align}
where $\bar{g}_B$ is a universal upper bound on $G^*_B(p(\cdot, \cdot))$ for all pmfs $p(\cdot, \cdot)$. 
Using equation  \eqref{eqnlast-1}, we obtain
\begin{align}\label{eqnlast0}\bigg|\frac 1n\sum_{i=1}^n\sum_{b_1, b_2}q(b_1, b_2)G^*_B(q(v_{Ai}, v_{Bi}|b_1, b_2))-\sum_{\tilde c}q(\tilde c)G^*_B(q(v_{Ai}, v_{Bi}|\tilde c))\bigg|\leq 2\epsilon \bar{g}_B.\end{align}
Similarly, we have
\begin{align}\bigg|\frac 1n\sum_{i=1}^n\sum_{b_1, b_2}q(b_1, b_2)G_A(q(v_{Ai}, v_{Bi}|b_1, b_2))- \sum_{\tilde c}q(\tilde c)G_A(q(v_{Ai}, v_{Bi}|\tilde c))\bigg|\leq 2\epsilon \bar{g}_A,\label{eqnlast1n}\end{align}
where $\bar{g}_A$ is a universal upper bound on $G_A(p(\cdot, \cdot))$ for all pmfs $p(\cdot, \cdot)$. The alternative scenario works for us because by taking communication variable as $B_1$ and shared randomness $S$ as $B_2$, we obtain the desired result.

\end{proof}

\section{Proof of Trembling Hand Perfect Equilibrium (THP) 
\label{sec:app-THP}}
Let $\Sab=(\Sa,\Sb)$ denote the strategy given in Section \ref{sec:THPMR}. In order to show that this pair of strategy is THP, we should introduce a
sequence of completely mixed strategies converging to $\Sab$ where $\Sab$ 
should be the best response at all information sets for every element of the 
sequence. Define completely mixed strategy
pair $\Se=(\Sae,\Sbe)$ as follows. At any given information set, the player who
should continue the game chooses the action given by $\Sab$ with probability
$1-\epsilon$ and the other possible action with probability $\epsilon$. Then
with $\epsilon \rightarrow 0$, $\Se$ converges to $\Sab$. We will show that for
$\epsilon$ small enough, $\Sab$ is the best response at any information set
given $\Se$ for other information sets, yielding the desired sequence of
completely mixed strategies.

For doing so, for any given information set $\is$, we fix strategy
$S^\epsilon$ for  information sets $\hat{\is} \neq \theta$ and find the
optimal action at $\is$. This action turns out to be the action given by
strategy $S$. 
$\theta$ can be an information set of Alice or Bob.
In the following, we first analyze Alice's information sets and 
then we will go through Bob's information sets in Section~\ref{sec:Bob-IS}.

\subsection{Alice's Information Sets \label{sec:Alice-IS}}
Assume $\is = (g_A, a^{k-1},b^{k-1})$ is a given information set for
Alice. Assume $n_g(k)$ and $n_l(k)$ are the number of gains and losses of Alice in 
this information set as was defined in \eqref{eq:ng-nl} which could be 
computed by having $\theta$. We shall fix strategy 
$S^\epsilon$ for all other information sets and find the optimal strategy in 
$\is$. Such a strategy would be relevant 
only if we pass through $\theta$.
Upon reaching $\theta$, Alice's strategy would be a combination of playing 
$\risky$ and
$\nonrisky$. However, since the resulting maximization is linear and hence 
could be restricted to 
pure strategies in this information set, it suffices to show that the 
action given by strategy $S$ in $\is$ is optimal. 
Also, note that since $\theta$ is given, Alice's gains at stages $1$ through 
$k-1$ could be deterministically calculated. Therefore, we only need to consider 
gains at stages $k$ through $n$ and show that playing according to the strategy in $\is$ is optimal. In other words, we want to show that for sufficiently small $\epsilon$, the maximum
\[
 \max_{a_k \in \{ \risky , \nonrisky\}} \evwrt{S^\epsilon}{\ga{ (a_k, 
A_{[k+1:n]}), B_{[k:n]}}| \theta}
\]
is taken for $a_k = \risky$ if $n_g(k) \geq n_l(k)$, and $a_k = \nonrisky$ if $n_g(k)> 
n_l(k)$ where 
$n_g(k)$ and $n_l(k)$ are the gains and losses of Alice which can be computed  from 
$\theta = (g_A, a^{k-1}, b^{k-1})$. Let
\begin{align}
\begin{split}
 f_\risky(\epsilon) &= 
\evwrt{S^\epsilon}{\ga{(A_k=\risky,
A_{[k+1:n]}),B_{[k:n]}}|\theta} \\
f_\nonrisky(\epsilon) &= 
\evwrt{S^\epsilon}{\ga{(A_k=\nonrisky,
A_{[k+1:n]}),B_{[k:n]}}|\theta}.
\end{split}
\end{align}
Observe that in $f_\nonrisky(\epsilon)$, since $A_k = \nonrisky$, Bob has no choice in that stage of the game and the value of $B_k$ is irrelevant. We need to compare $f_\nonrisky(\epsilon)$ and $f_\risky(\epsilon)$ for $\epsilon$ sufficiently small. 

Observe that
\[
\begin{split}
 f_\risky(\epsilon) &= p(T =0 | \theta) 
\evwrt{S^\epsilon}{\ga{(A_k=\risky,
A_{[k+1:n]}),B_{[k:n]}}|\theta, T=0}
\\
& \qquad + p(T =1 | \theta) 
\evwrt{S^\epsilon}{\ga{(A_k=\risky,
A_{[k+1:n]}),B_{[k:n]}}|\theta, T=1}
\end{split}
\]
A similar expression could be written for $f_\nonrisky(\epsilon)$. In order 
to make the comparison, we need the 
following lemma:
\begin{lem}\label{lemma:NA1}
 If $n_g(k) < n_l(k)$ then 
 \[
  \lim_{\epsilon \rightarrow 0} \frac{p(T=1|\theta)}{p(T=0|\theta)} = 0.
 \]
Also if $n_l(k) < n_g(k)$ then
\[
  \lim_{\epsilon \rightarrow 0} \frac{p(T=0|\theta)}{p(T=1|\theta)} = 0,
\]
and finally if $n_g(k) = n_l(k)$ then for all values of $\epsilon$ we have 
\[
\frac{p(T=0|\theta)}{p(T=1|\theta)} = \frac{p(T=0|g_A)}{p(T=1|g_A)}.
\]
\end{lem}
We provide the proof of this lemma later but observe that this lemma has an intuitive interpretation: when $n_g(k) > n_l(k)$, it is 
more probable that $T=1$, since $T=0$ results in more losses to Alice, 
therefore it suffices for Alice to assume $T=1$ and perform her strategy 
accordingly, which is $\risky$ in this case. On the other hand, when $n_g(k) < 
n_l(k)$, Alice is better off assuming $T=0$ and play $\nonrisky$. When $n_g(k) = 
n_l(k)$, the posteriors $p(T=0|g_A)$ and $p(T=1|g_A)$ is Alice's belief on $T$.

This lemma implies that to compare $f_\risky(\epsilon)$ with $f_\nonrisky(\epsilon)$ 
for small $\epsilon$, when $n_g(k) > n_l(k)$ we need to compare 
\[
 \hat{f}_X(\epsilon) = \evwrt{S^\epsilon}{\ga{(A_k=X, 
A_{[k+1:n]}),B_{[k:n]}}|\theta, T=1},
\]
for two value of $X = \risky$ and $X = \nonrisky$. Similarly, when $n_g < n_l$ 
we need to compare
\[
 \hat{f}_X(\epsilon) = \evwrt{S^\epsilon}{\ga{(A_k=X,
A_{[k+1:n]}),B_{[k:n]}}|\theta, T=0},
\]
and when $n_g = n_l$ we should compare
\[
\begin{split}
 \hat{f}_X(\epsilon) &= p(T=0|g_A) 
\evwrt{S^\epsilon}{\ga{(A_k=X,
A_{[k+1:n]}),B_{[k:n]}}|\theta, T=0} \\
& \qquad +
p(T=1|g_A) \evwrt{S^\epsilon}{\ga{(A_k=X,
A_{[k+1:n]}),B_{[k:n]}}|\theta, T=1}.
\end{split}
\]
In the sequel, we do the comparisons for these three cases separately.

\subsubsection{$n_g(k) = n_l(k)$}
In this case, we show that $\hat{f}_\risky(0) > \hat{f}_\nonrisky(0)$ which 
implies that for small values of $\epsilon$, we have
$\hat{f}_\risky(\epsilon) > \hat{f}_\nonrisky(\epsilon)$. To do so, note that setting $\epsilon=0$ is equivalent to setting 
$A_{[k+1:n]}$ and $B_{[k+1:n]}$ to the path given by strategy $S$.

In order to compute $\hat{f}_\risky(0)$, note that  
such a path could be obtained by noting that when $T=0$, Alice plays $\risky$ 
in stage $k$ which results to a loss (since Bob plays $\selfish$ and $T=0$) and 
from then on, Alice plays $\nonrisky$ resulting in an overall gain of $\gamin + 
(n-k)/2$. On the other hand, when $T=1$, Alice will always gain and $n_g$ will 
always remain greater than $n_l$, yielding an overall gain of $(n-k+1)\gamax$. 
Therefore
\[
 \hat{f}_\risky(0) = p(T=0|g_A) \left ( \gamin + \frac{n-k}{2} \right ) + 
p(T=1|g_A)  (n-k+1) \gamax.
\]

$\hat{f}_\nonrisky(0)$ could be computed similarly. Note that when $T=0$, since 
we have assumed that Alice plays $\nonrisky$ in stage $k$, he gains $1/2$ in 
that stage, leaving $n_g = n_l$. If $k<n$, Alice risks in stage $k+1$ and will 
play $\nonrisky$ in the remaining stages. This will give her a total gain of
\[
\frac{1}{2} + \one{k<n} \gamin + \frac{1}{2} \pp{n-k-1}. 
\]
On the other hand, when $T=1$, Alice gains $1/2$ in stage $k$ and then will 
risk at the remaining stage and will gain $\gamax$ yielding an overall gain of 
$\frac{1}{2} + (n-k) \gamax$. Hence,
\[
 \hat{f}_\nonrisky (0) = p(T=0|g_A) \left( \frac{1}{2} + \one{k<n} \gamin + 
\frac{1}{2} \pp{n-k-1} \right) + p(T=1|g_A) \left( \frac{1}{2} + (n-k) \gamax 
\right).
\]

In order to compare these two values, note that
\[
\begin{split}
 \hat{f}_\risky(0) - \hat{f}_\nonrisky(0) &= p(T=0|G_A = g_A) \left ( \one{k=n} 
\gamin + 
\frac{1}{2}\np{n-k-1} \right ) + p(T=1|G_A = g_A) \left (\gamax - 
\frac{1}{2}\right) \\
&= \begin{cases}
    p(T=0|G_A = g_A) \left ( \gamin - \frac{1}{2} \right ) + p(T=1|G_A = g_A) 
\left ( \gamax - \frac{1}{2} \right ) & k=n \\
  p(T=1|G_A = g_A) \left ( \gamax - \frac{1}{2} \right ) & k < n
  \end{cases} \\
  & \stackrel{(a)}{\geq} p(T=0|G_A = g_A) \left (\gamin - \frac{1}{2} 
\right ) + p(T = 1 | G_A = g_A) \left ( \gamax - \frac{1}{2} \right ) \\
 & \stackrel{(b)}{>} 0
\end{split}
\]
where $(a)$ exploits the fact that $\gamin \leq 1/2$ and $(b)$ uses our 
main assumption
\eqref{eq:one-stage-risk}. Since this is strictly greater than zero, we are 
done in this case.

\subsubsection{$n_g(k) > n_l(k)$}
As was discussed before, we should only consider the terms corresponding to 
$T=1$. Quite similar to the discussion of the last section, we have
\[
 \hat{f}_\risky(0) = (n-k+1) \gamax,
\]
and
\[
 \hat{f}_\nonrisky(0) = \frac{1}{2} + (n-k) \gamax.
\]
Therefore
\[
\hat{f}_\risky(0) - \hat{f}_\nonrisky(0) = p(T = 1 | G_A = g_A) \left ( \gamax 
- 
\frac{1}{2} \right ) > 0
\]
which is strictly greater than zero and we are done.

\subsubsection{$n_g(k) < n_l(k)$}
Considering terms corresponding to $T=0$ we have
\[
 \hat{f}_\risky(0) = \gamin + \frac{1}{2} (n-k),
\]
and
\[
 \hat{f}_\nonrisky(0) = \frac{1}{2} (n-k+1),
\]
hence,
\[
 \hat{f}_\nonrisky(0) - \hat{f}_\risky(0) = p(T=0 | G_A = g_A) \left ( 
\frac{1}{2} - 
\gamin \right ) > 0
\]
and we are done.

\begin{proof}[Proof of Lemma \ref{lemma:NA1}]
First assume $n_g(k) < n_l(k)$. Using Bayes rule we have
\begin{equation}
\label{eq:frac-limit-simplified}
\frac{p(T=1|\theta)}{p(T=0|\theta)} = \frac{p(T=1|g_A, a^{k-1}, 
b^{k-1})}{p(T=0|g_A, a^{k-1}, b^{k-1})} = \frac{p(a^{k-1}, b^{k-1}|g_A, 
T=1)}{p(a^{k-1}, b^{k-1}|g_A, T=0)} \frac{p(T=1|g_A)}{p(T=0|g_A)}.
\end{equation}
Now for $t \in \{0,1\}$ we have
\[
 p(a^{k-1}, b^{k-1} | g_A, T=t) = \prod_{i=1}^{k-1} p(a_i| g_A, T=t, a^{i-1}, 
b^{i-1}) p(b_i|g_A, T=t, a^{i}, b^{i-1}).
\]
Note that the term $p(a_i| g_A, T=t, a^{i-1}, 
b^{i-1})$ corresponding to Alice's strategy is independent from the value of 
$t$, since Alice only looks at $g_A$ and determines the number of gains and 
losses from the sequence of actions to determine her action. Therefore
\[
 \frac{p(a^{k-1}, b^{k-1} | g_A, T=1)}{p(a^{k-1}, b^{k-1} | g_A, T=0)} = 
\frac{\prod_{i=1}^{k-1} p(b_i|g_A, T=1, a^{i}, b^{i-1})}{\prod_{i=1}^{k-1} 
p(b_i|g_A, T=0, a^{i}, b^{i-1})}.
\]
We know that when Alice plays $\risky$, Bob always plays $\selfish$ with 
probability $1-\epsilon$ and $\nonselfish$ with probability $\epsilon$. As a 
result, in case $T=0$, we can conclude that when Alice looses, Bob has played 
$\selfish$. The number of such stages is $n_l$, contributing a term 
$(1-\epsilon)^{n_l(k)}$ to the probability. Furthermore, when Alice gains, Bob has 
played $\nonselfish$ which contributes a term $\epsilon^{n_g(k)}$. The case $T=1$ is 
similar. Therefore
\[
 \frac{\prod_{i=1}^{k-1} p(b_i|g_A, T=1, a^{i}, b^{i-1})}{\prod_{i=1}^{k-1} 
p(b_i|g_A, T=0, a^{i}, b^{i-1})} = \frac{\epsilon^{n_l(k)} 
(1-\epsilon)^{n_g(k)}}{\epsilon^{n_g(k)} (1-\epsilon)^{n_l(k)}}.
\]
Putting this into \eqref{eq:frac-limit-simplified} and sending $\epsilon 
\rightarrow 0$ we get the desired result. The two other cases are similar.
\end{proof}

\subsection{Bob's Information Sets \label{sec:Bob-IS}}
Assume an information set $\is = (\vb, a^k, b^{k-1})$ of Bob is given. If $a_k 
= 
\nonrisky$, then Bob has no choice, therefore assume that $a_k  = \risky$. 
We shall fix strategy 
$S^\epsilon$ for all other information sets and find the optimal strategy in 
$\is$.
Such a strategy would be relevant 
only if we pass through $\theta$.
Upon reaching $\theta$, Bob's strategy would be a combination of playing 
$\selfish$ and
$\nonselfish$. However, following the same discussion we had for Alice, since the resulting maximization is linear and hence 
could be restricted to 
pure strategies in this information set, it suffices to show that the 
action given by strategy $S$ in $\is$ is optimal. 
Also, note that since $\theta$ is given, Bob's gains at stages $1$ through 
$k-1$ could be deterministically calculated. Therefore, as before we only need to consider 
gains at stages $k$ through $n$ and show that playing selfishly is optimal.
In other words we should solve the following optimization problem
\begin{align}
 \max_{b_k\in \{\selfish, \nonselfish\}} 
\evwrt{\Se}{\gbwrt{A_{[k:n]}, (b_k,B_{[k+1:n]})}{g_B} | \is},
\end{align}
for small values of $\epsilon$, where $A_k=a_k$ (since we have conditioned on $\theta$) and the future actions $A_{[k+1:n]},B_{[k+1:n]}$ follow then one given by $\Sab^{\epsilon}$ (i.e. $\Sab$ with probability
$1-\epsilon$ and the other possible action with probability $\epsilon$). We need to show that the answer to the above maximization problem is $b_k = \selfish$. 
In fact we should compare the following two functions for small values of 
$\epsilon$:
\begin{equation}
\label{eq:f-epsilon-s-ns}
  f_X(\epsilon) = \evwrt{\Se}{\gbwrt{A_{[k:n]}, (B_k=X,B_{[k+1:n]})}{g_B} | 
\is}, \qquad X \in \{ \selfish, \nonselfish\}.
\end{equation}
Note that we can expand
\[
\begin{split}
 f_X(\epsilon) &= p(T=0|\is) \evwrt{\Se}{\gbwrt{A_{[k:n]}, (B_k=X,B_{[k+1:n]})}{g_B} | 
\is,T=0} \\
&\qquad + p(T=1|\is) \evwrt{\Se}{\gbwrt{A_{[k:n]}, (B_k=X,B_{[k+1:n]})}{g_B} | 
\is,T=1},
\qquad X \in \{ \selfish, \nonselfish\}.
\end{split}
\]
Hence, it suffices to show that both conditioned on $T=0$ and $T=1$, $\selfish$ 
is dominant. In other words, if we define
\begin{equation}
\label{eq:f-epsilon-s-ns-t}
 f_{X,t}(\epsilon) = \evwrt{\Se}{\gbwrt{A_{[k:n]}, (B_k=X,B_{[k+1:n]})}{g_B} | 
\is,T=t}, X\in \{\selfish, \nonselfish\}, t\in\{0,1\}
\end{equation}
we will show that for any value of $T=t\in\{0,1\}$ we have
$f_{\selfish,t}(\epsilon) \geq f_{\nonselfish,t}(\epsilon)$ for $\epsilon$ 
sufficiently small.

\subsubsection{Case $T=1$: Showing that $f_{\selfish,1}(\epsilon) \geq f_{\nonselfish,1}(\epsilon)$ for small $\epsilon$ 
}
Since $T$ and $g_B$ are fixed, similar to Section~\ref{sec:Alice-IS} we can 
define $n_g$ and $n_l$ with 
respect to Alice at 
each stage.
It is more convenient to look at terms in \eqref{eq:f-epsilon-s-ns-t} as a one 
dimensional random walk in the following way. Note that at stage $j$, if we 
define $\delta(j) = n_g(j) - n_l(j)$ to be the difference between gains and 
losses of Alice before stage $j$, action pair $a_j,b_j$ will either add one 
unit to 
this value after playing at stage $j$, subtract one unit or leave it unchanged.
More precisely, if $\delta \geq 0$, there are three possible moves: $(i)$ $a_j 
= 
\risky, b_j = 
\selfish$ which increases $\delta$ by one (since $T=1$) which is 
equivalent to one move to the right in the random walk with probability 
$(1-\epsilon)^2$, this action has a gain of $\Gbmax$ for Bob, $(ii)$ $a_j = 
\risky, 
 b_j = \nonselfish$ which is equivalent to a left move with probability 
$\epsilon(1-\epsilon)$ and gain $\Gbmin$ and $(iii)$ $a_j = \nonrisky$ 
which is equivalent to no move with probability $1-\epsilon$ and gain $1/2$. 
Similarly one can define transition probabilities and gains for $\delta < 0$ 
which is depicted in Figure~\ref{fig:va-friend-vb-random-walk}.

\begin{figure}
\centering
\begin{tikzpicture}
  \footnotesize
 \tikzstyle{place}=[fill=blue!20,draw=blue!50];
 \tikzstyle{shaddow}=[circle,inner sep=2mm];
 \tikzstyle{my loop}=[->,to path={ .. controls +(310:1.3) and +(230:1.3) .. 
(\tikztotarget) \tikztonodes}]
 \node [shaddow] (a0) at (0,0) {};
 \node [shaddow] (a1) at (3,0) {};
 \node [shaddow] (a-1) at (-3,0) {};
 \node [shaddow] (a-2) at (-6,0) {};
 \node [shaddow] (a-3) at (-9,0) {};
 
 \draw [->,blue!50,very thick] (-10,0) -- (4,0);
 
 \node at (4.5,0) {$\delta$};
 
 \draw [place] (a0) circle (2pt);
 \draw [place] (a1) circle (2pt);
 \draw [place] (a-1) circle (2pt);
 \draw [place] (a-2) circle (2pt);
 \draw [place] (a-3) circle (2pt);
 
 \draw (a0) edge [->,bend left] node[above] 
{$(1-\epsilon)^2$, \color{red}{$\gbmax$}} (a1);
 \draw (a0) edge [my loop] node[below] {$\epsilon$, \color{red}{$1/2$}} (a0);
 \draw (a0) edge [->,bend right] node [above] 
{$\epsilon(1-\epsilon)$, \color{red}{$\gbmin$}} (a-1);

  \draw (a-2) edge [->,bend left] node[above] 
{$\epsilon(1-\epsilon)$, \color{red}{$\gbmax$}} (a-1);
  \draw (a-2) edge [my loop] node[below] {$1-\epsilon$, \color{red}{$1/2$}} 
(a-2);
  \draw (a-2) edge [bend right] node[above] 
{$\epsilon^2$, \color{red}{$\gbmin$}} (a-3);
 
 \node[text=blue!50] at (a0.south) {0};
 \node[text=blue!50] at (a1.south) {1};
 \node[text=blue!50] at (a-1.south) {-1};
 \node[text=blue!50] at (a-2.south) {-2};
 \node[text=blue!50] at (a-3.south) {-3};
\end{tikzpicture}
\caption{\label{fig:va-friend-vb-random-walk}Auxiliary random walk for the case 
$T=1$. Transition probabilities as well as Bob's gain in red are 
shown for two values of $\delta=0$ and $\delta=-2$. Note that the values for 
$\delta\geq 0$ are identical to that of $\delta = 0$ and values for $\delta <0$ 
are all the same as $\delta = -2$.}
\end{figure}
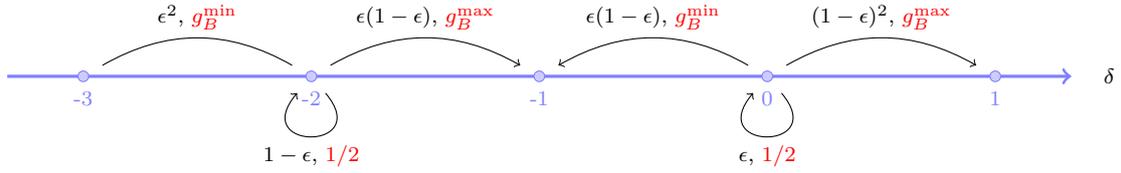

From now on, we shall continue our argument solely on this random walk. Define 
$F^1_{\delta,l}(\epsilon)$ to be the expected value of sum 
of the gains one would observe if he started at position $\delta$ and moved $l$ 
times, which is a polynomial in $\epsilon$. Note that
\[
 f_{\selfish, 1}(\epsilon) = \gamax + F^1_{\delta(k)+1, n-k}(\epsilon),
\]
and
\[
 f_{\nonselfish, 1}(\epsilon) = \gamin + F^1_{\delta(k)-1, n-k}(\epsilon).
\]
Now we claim that for all values of $\epsilon$, we have $f_{\risky, 
1}(\epsilon) \geq f_{\nonrisky,1}(\epsilon)$, showing that playing $\risky$ at 
this information set is dominant. In order to show this, we use the idea of 
\emph{coupling} in this random walk. When we choose to play $\selfish$ at stage $k$, 
we move man number 1, say the selfish man, from position $\delta(k)$ to 
$\delta(k) + 1$ and from then on, he moves randomly $n-k$ times. On the other 
hand, when we decide to play $\nonselfish$, we move man number 2, say the 
non-selfish man, from position $\delta(k)$ to $\delta(k) -1$ and let him move 
$n-k$ times. Assume $\Delta_\selfish(j)$ for $j>0$ be the position of the 
selfish man at stage $k+j-1$; hence, $\Delta_\selfish(1) = \delta(k) +1$ and 
$\Delta_\selfish(j)$ for $j>0$ is a random variable. Define 
$\Delta_\nonselfish(j)$ similarly for the non-selfish man. Also, 
$\Delta_\selfish(j)$ is independent from $\Delta_\nonselfish(k)$ since their 
moves are independent.

For each of the two men, $X \in \{\selfish, \nonselfish\}$, define $M_X(j) = 
\Delta_X(j) - \Delta_X(j-1)$ with $\Delta_\selfish(0) = \Delta_\nonselfish(0) 
= \delta(k)$ to be their move at stage $k+j-1$, hence $M_X(j) \in 
\{-1,1,0\}$. Note that since $\gamax + \gamin = 1$, for each of the 
two men, gain at stage $j$ is equal to
\[
 \frac{1}{2} + M_X(j) \left ( \gbmax - \frac{1}{2} \right ).
\]
Therefore the overall gain in stages $k$ through $n$ would be
\[
\begin{split}
 \sum_{j=1}^{n-k+1}\left( \frac{1}{2} + M_X(j) \left( \gbmax - \frac{1}{2} \right)\right) &= 
\frac{n-k+1}{2} + \left( \gbmax - \frac{1}{2} \right) \sum_{j=1}^{n-k+1} M_X(j) 
\\ 
& = \frac{n-k+1}{2} + \left( \gbmax - \frac{1}{2} \right) 
(\Delta_X(n-k+1)-\delta(k) ) \qquad X \in \{\selfish, \nonselfish\},
\end{split}
\]
and
\[
 f_X(\epsilon) = \frac{n-k+1}{2}- \delta(k) \left ( \gbmax - \frac{1}{2} \right 
)  + \ev{\Delta_X(n-k+1) } \left( \gbmax  - \frac{1}{2} \right) \qquad \qquad X 
\in \{ \selfish, \nonselfish \}.
\]
This suggests that the overall gain is only a function of the moves through the overall displacement at stage $n$ (i.e. the final location of the two men at the final stage). This is a result of the fact that the average of gains of moving 
right and left in the random walk is equal to the gain of no move. 

Now define $\tilde{\Delta}_X(j)$ ($X \in \{\selfish, \nonselfish\}$) to be the coupled placement of the two men as follows:
The two men move independently until the first time they reach the same node. From that point on, the two men are coupled and forced to move together on the Markov chain (i.e. one of the men walks randomly on the chain with the other man mimicking his moves). It is evident that the marginal distribution of 
$\tilde{\Delta}_X(j)$ is equal to the distribution of $\Delta_X(j)$, since the two 
men are moving on the same random walk. Since the non-selfish man starts to the left of the selfish man, i.e.
$\tilde{\Delta}_\nonselfish(1) = \delta(k) - 1 < \delta(k) + 1 = 
\tilde{\Delta}_\selfish(1)$, and they get coupled when they reach at the same node, the non-selfish man cannot go to the right of the selfish man. In other words with probability one we have
\begin{equation}
\label{eq:tilde-delta-comp-T=1}
 \tilde{\Delta}_\nonselfish(j) \leq \tilde{\Delta}_\selfish(j) \qquad 1 \leq 
j \leq n-k+1.
\end{equation}
Therefore
\[
\begin{split}
 f_{\nonselfish,1}(\epsilon) &= \frac{n-k+1}{2}- \delta(k) \left ( \gbmax - 
\frac{1}{2} \right 
)  + \ev{\Delta_\nonselfish(n-k+1) } \left( \gbmax  - \frac{1}{2} \right) \\
& = \frac{n-k+1}{2}- \delta(k) \left ( \gbmax - 
\frac{1}{2} \right 
)  + \ev{\tilde{\Delta}_\nonselfish(n-k+1) } \left( \gbmax  - \frac{1}{2} 
\right) \\
& \stackrel{(a)}{\leq} \frac{n-k+1}{2}- \delta(k) \left ( \gbmax - 
\frac{1}{2} \right 
)  + \ev{\tilde{\Delta}_\selfish(n-k+1) } \left( \gbmax  - \frac{1}{2} 
\right) \\
& = \frac{n-k+1}{2}- \delta(k) \left ( \gbmax - 
\frac{1}{2} \right 
)  + \ev{\Delta_\selfish(n-k+1) } \left( \gbmax  - \frac{1}{2} 
\right) \\
&= f_{\selfish,1}(\epsilon),
\end{split}
\]
where $(a)$ exploits equation \eqref{eq:tilde-delta-comp-T=1} and the fact that $\gbmax 
\geq 1/2$.
This completes the proof for this case.

\subsubsection{Case $T=0$: Showing that $f_{\selfish, 0}(\epsilon) > f_{\nonselfish, 0}(\epsilon)$ for small $\epsilon$ }
Similar to the previous case, we can define a similar yet different random 
walk which is depicted in Figure~\ref{fig:va-enemy-vb-random-walk}. 
$F_{\delta,l}^{0}(\epsilon)$ is defined 
similarly.
Quite similar to the previous case, consider selfish and non-selfish men and 
denote their placement by $\Delta_\selfish(j)$ and $\Delta_\nonselfish(j)$ for 
$j>0$ and $\Delta_\selfish(0) = \delta_\nonselfish(0) = \delta(k)$. Also 
$\Delta_\selfish(1) = \delta(k) - 1$ and $\Delta_\nonselfish(1) = \delta(k) + 
1$. Furthermore define $M_X(j) = \Delta_X(j) - \Delta_X(j-1)$. Then gain of max 
$X$ at stage $k+j-1$ is 
\[
 \frac{1}{2} + M_X(j) \left ( \gbmin - \frac{1}{2} \right).
\]
It is easily verified that all expressions are similar to those of case $T=1$ 
by substituting $\gbmax$ with $\gbmin$, hence
\[
 f_{X,0}(\epsilon) = \frac{n-k+1}{2}- \delta(k) \left ( \gbmax - \frac{1}{2} 
\right 
)  + \ev{\Delta_X(n-k+1) } \left( \gbmin  - \frac{1}{2} \right) \qquad \qquad X 
\in \{ \selfish, \nonselfish \}.
\]
Again, define $\tilde{\Delta}_X(j)$ to be the coupled random variable 
representing coupled placement of the two men. Since 
\[
 \tilde{\Delta}_\selfish(1) = \delta(k) - 1 \leq \delta(k) + 1 = 
\tilde{\Delta}_\nonselfish(1),
\]
therefore we have
\begin{equation}
\label{eq:tilde-delta-comp-T=0}
 \tilde{\Delta}_\selfish(j) \leq \tilde{\Delta}_\nonselfish(j) \qquad 1 \leq j 
\leq n-k+1.
\end{equation}
Using this and the fact that the marginal distribution of $\tilde{\Delta_X}$ 
is equal to the distribution of $\Delta_X$, we have
\[
\begin{split}
 f_{\nonselfish,0}(\epsilon) &= \frac{n-k+1}{2}- \delta(k) \left ( \gbmin - 
\frac{1}{2} \right 
)  + \ev{\Delta_\nonselfish(n-k+1) } \left( \gbmin  - \frac{1}{2} \right) \\
& = \frac{n-k+1}{2}- \delta(k) \left ( \gbmax - 
\frac{1}{2} \right 
)  + \ev{\tilde{\Delta}_\nonselfish(n-k+1) } \left( \gbmin  - \frac{1}{2} 
\right) \\
& \stackrel{(a)}{\leq} \frac{n-k+1}{2}- \delta(k) \left ( \gbmax - 
\frac{1}{2} \right 
)  + \ev{\tilde{\Delta}_\selfish(n-k+1) } \left( \gbmin  - \frac{1}{2} 
\right) \\
& = \frac{n-k+1}{2}- \delta(k) \left ( \gbmin - 
\frac{1}{2} \right 
)  + \ev{\Delta_\selfish(n-k+1) } \left( \gbmin  - \frac{1}{2} 
\right) \\
&= f_{\selfish,0}(\epsilon),
\end{split}
\]
where $(a)$ exploits equation \eqref{eq:tilde-delta-comp-T=0} and the fact that $\gbmin 
\leq 1/2$. This completes the proof by showing that playing $\selfish$ is 
dominant.

\begin{figure}
\centering
\begin{tikzpicture}
  \footnotesize
 \tikzstyle{place}=[fill=blue!20,draw=blue!50];
 \tikzstyle{shaddow}=[circle,inner sep=2mm];
 \tikzstyle{my loop}=[->,to path={ .. controls +(310:1.3) and +(230:1.3) .. 
(\tikztotarget) \tikztonodes}]
 \node [shaddow] (a0) at (0,0) {};
 \node [shaddow] (a1) at (3,0) {};
 \node [shaddow] (a-1) at (-3,0) {};
 \node [shaddow] (a-2) at (-6,0) {};
 \node [shaddow] (a-3) at (-9,0) {};
 
 \draw [->,blue!50,very thick] (-10,0) -- (4,0);
 
 \node at (4.5,0) {$\delta$};
 
 \draw [place] (a0) circle (2pt);
 \draw [place] (a1) circle (2pt);
 \draw [place] (a-1) circle (2pt);
 \draw [place] (a-2) circle (2pt);
 \draw [place] (a-3) circle (2pt);
 
 \draw (a0) edge [->,bend left] node[above] 
{$\epsilon(1-\epsilon)$, \color{red}{$\gbmin$}} (a1);
 \draw (a0) edge [my loop] node[below] {$\epsilon$, \color{red}{$1/2$}} (a0);
 \draw (a0) edge [->,bend right] node [above] 
{$(1-\epsilon)^2$, \color{red}{$\gbmax$}} (a-1);

  \draw (a-2) edge [->,bend left] node[above] 
{$\epsilon^2$, \color{red}{$\gbmin$}} (a-1);
  \draw (a-2) edge [my loop] node[below] {$1-\epsilon$, \color{red}{$1/2$}} 
(a-2);
  \draw (a-2) edge [bend right] node[above] 
{$\epsilon(1-\epsilon)$, \color{red}{$\gbmax$}} (a-3);
 
 \node[text=blue!50] at (a0.south) {0};
 \node[text=blue!50] at (a1.south) {1};
 \node[text=blue!50] at (a-1.south) {-1};
 \node[text=blue!50] at (a-2.south) {-2};
 \node[text=blue!50] at (a-3.south) {-3};
\end{tikzpicture}
\caption{\label{fig:va-enemy-vb-random-walk}Auxiliary random walk for the case 
$T=0$. Transition probabilities as well as Bob's gain in red are 
shown for two values of $\delta=0$ and $\delta=-2$. Note that the values for 
$\delta\geq 0$ are identical to that of $\delta = 0$ and values for $\delta <0$ 
are all the same as $\delta = -2$.}
\end{figure}
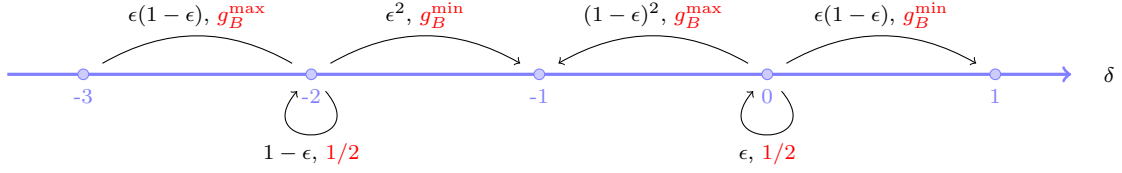


\section{Proofs for Adjusted Winner}
The
division given by the AW algorithm are functions of
algorithm's inputs, $a$ and $b$, when $m=2$.
In this case, $\awwrt{A}{a,b}$ is a vector of length $2$, say
$(d^A_1,d^A_2)$, indicating the portion of goods given to Alice. Similarly
$\awwrt{B}{a,b}=(d^B_1,d^B_2)$ denotes the portion of goods given to Bob. Since
we divide the
goods between parties, $d^A_i + d^B_i = 1,i=1,2$. In the following, since we are
interested in
Alice's gain, we use $\aw{a,b}$ for $\awwrt{A}{a,b}$ unless otherwise stated. The
following formulation for $\aw{a,b}$, as derived in
Section ~\ref{sec:binary-aw-proofs}:
\begin{equation}
\label{eq:aw-four-case}
 \aw{a,b} = \begin{cases}
             \left(0,\frac{1}{2-a-b} \right ) & 0 \leq a \leq \min(1-b,b), \\
             \left(\frac{1}{a+b},0 \right ) & \max(1-b,b) \leq a \leq 1, \\
             \left(1,\frac{1-a-b}{2-a-b} \right ) & b < a < 1-b \wedge b\leq
1/2, \\
             \left(1-\frac{1}{a+b},1 \right ) & 1-b < a < b \wedge b> 1/2.
            \end{cases}
\end{equation}

\begin{definition}
\label{def:daw}
 Alice's gain when she announces valuation $(\ta,1-\ta)$ while her true
valuation is $(a,1-a)$ in the case that Bob's true valuation is $(b,1-b)$ which
is equal to his announced valuation is denoted by $\daw{\ta}{a}{b}$ which is
equal to,
\begin{equation}
 \daw{\ta}{a}{b} = \aw{\ta,b}\cdot (a,1-a),
\end{equation}
where $\cdot$ is the inner product operation.
\end{definition}
Using \eqref{eq:aw-four-case} we can write the exact expression of this
function as we will see later. An example of $\dawhead$ is presented in
Figure~\ref{fig:psi-example}.

\begin{figure}
 \centering
  \includegraphics[width=0.5\textwidth]{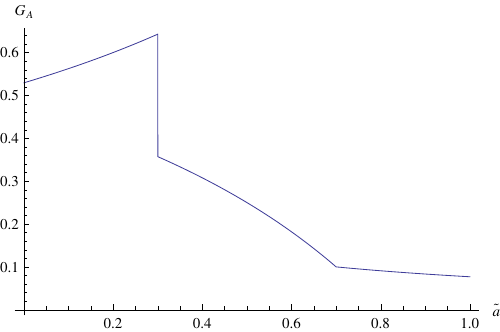}
  \caption{An example of $\dawhead$ as a function of $\ta$ for $b=0.3$
and $a=0.1$. Note the discontinuity of the function at $\ta=b$, its convexity in 
the intervals
$(0,b)$ and $(1-b,1)$ and its concavity in the interval $(b,1-b)$.
\label{fig:psi-example}}
\end{figure}

In the case where $b$ is uniformly distributed in $[\bmin,\bmax]$, the gain
associated with $\ta$ is an integral with respect to $b$, which is discussed in
the following definition.

\begin{definition}
  Alice's expected gain when she announces valuation $(\ta,1-\ta)$ while her
true
valuation is $(a,1-a)$ in the case that Bob's true valuation is uniformly
distributed in $[\bmin,\bmax]$ and he acts honestly is denoted by
$\daw{\ta}{a}{\bmin,\bmax}$ and is equal to,
\begin{equation}
\label{eq:psi-integral}
\daw{\ta}{a}{\bmin,\bmax} =\frac{1}{\bmax - \bmin} \int_{\bmin}^{\bmax}
\daw{\ta}{a}{b} \mathrm{d} b.
\end{equation}
\end{definition}

Note that $\ta$ and $b$ are the two inputs to the Adjusted Winner algorithm.
When we integrate over $b$, at one point in the integration $b=\ta$. As is
discussed in Section ~\ref{sec:binary-aw-proofs} in the case where the two
inputs
to the Adjusted Winner algorithm are identical, there are two possible divisions
of the cake as the output of the algorithm. If both players had announced their
valuations truly, these two divisions would give them the same gains; however,
in our scenario, Alice announces an untrue valuation. Thus, when $\ta=b$, these
two valuations result in two different gains for Alice and the function under
integration is not defined in this one point. However, since the integral is not
dependent on the value of one point, we can omit it.


\begin{definition}
\label{def:dawmax}
 The maximum expected value of Alice's gain with above conditions is defined as,
\begin{equation}
\label{eq:dawmax-definition}
\begin{split}
 \dawmax{a}{\bmin,\bmax} &= \max_{0\leq\ta\leq1}
\daw{\ta}{a}{\bmin,\bmax}\\
&=\max_{\bmin\leq\ta\leq\bmax}
\daw{\ta}{a}{\bmin,\bmax}.
\end{split}
\end{equation}
\end{definition}

The second line in \eqref{eq:dawmax-definition} suggests that the optimum value
of $\ta$ for
$\dawhead$
falls in the interval $[\bmin,\bmax]$ which is justified in
Corollary~\ref{cor:def-dawmax-justification}.

\begin{definition}
  \label{def:dgamax}
 For a fixed value of Alice's valuation, $a$, and Bob's valuation $b$ uniform in
$[\bmin,\bmax]$ and a series of dividing
points for $k$ questions $\bmin=b_0\leq\dots\leq b_{2^k} = \bmax$, the
improvement of Alice's gain by asking this set of questions is denoted by
\begin{equation}
\label{eq:dga-definition}
 \dga{k}{b_0,\dots,b_{2^k}} = \sum_{i=1}^{2^k} \frac{b_i-b_{i-1}}{b_{2^k}-b_0}
\dawmax{a}{b_{i-1},b_i} - \dawmax{a}{\bmin,\bmax}.
\end{equation}
The maximum improvement by asking $k$ questions is
\begin{equation}
 \dgamax{k}{\bmin,\bmax} = \max_{
\bmin=b_0\leq\dots\leq b_{2^k} = \bmax} \dga{k}{b_0,\dots,b_{2^k}}.
\end{equation}

\end{definition}

Note that the term $\frac{b_i-b_{i-1}}{b_{2^k}-b_0}$ in
\eqref{eq:dga-definition} is the probability of the event $b\in [b_{i-1},b_i]$;
in fact, $\dga{k}{b_0,\dots,b_{2^k}}$ is the expected value of Alice's
improvement in gain having the fact that $b$ is uniformly distributed in
$[\bmin,\bmax]$.

For the
sake of simplicity we have assumed the maximums in
Definitions~\ref{def:dawmax} and \ref{def:dgamax} exist. One can check that if
we replace maximums by supremum   and taking suboptimal points, the same results
hold.

Note that $\dgamaxhead$ is defined on intervals. When we want
to
prove upper bounds on gain improvement, it will be convenient to work with a
special set of these functions, which we name \emph{interval concave}. Note
that this terminology is not related to the concept of concavity and is used
simply because the condition has similarities to what we have for concave
functions.

\begin{definition}
 \label{def:interval-concave}
  If $\mathcal{M}$ is the set of all pairs $(x,y) \in \reals^2$ such that
$\bmin\leq x<y\leq \bmax$, a function $\Delta:\mathcal{M}\rightarrow \reals$ is
said to be interval concave in $[\bmin,\bmax]$ if for all $(x,y) \in
\mathcal{M}$ and $x\leq t \leq y$ we have,
\begin{equation}
  \frac{t-x}{y-x} \Delta(x,t) + \frac{y-t}{y-x} \Delta(t,y) \leq \Delta(x,y) .
\end{equation}
\end{definition}

\begin{thm}
\label{thm:upper-bound-general}
 Assume $a$ is fixed, $b$ is uniformly distributed in $[\bmin,\bmax]$ and
we have an interval concave $\tD$ in $[\bmin,\bmax]$ which is
an upper bound for
$\dgamaxhead_1$ in this interval, i.e.\
\begin{equation}
\label{eq:thm:upper-bound-general-k=1}
 \forall x,y \:\: \bmin \leq x < y \leq \bmax \qquad \dgamax{1}{x,y} \leq
\tD(x,y),
\end{equation}
then for all $k\geq 1$ we have
\begin{equation}
 \dgamax{k}{\bmin,\bmax} \leq k \tD(\bmin,\bmax).
\end{equation}
\end{thm}

The proof of the above theorem is given in
Section ~\ref{proofsAW2}. The proof of the other main theorems are given in appendices \ref{secProofThem5AW} and  \ref{secProofThem5AW2}.

\subsection{Deriving Adjusted Winner formulation for two goods
\label{sec:binary-aw-proofs}}
When there are only two goods, by changing their ordering, we
realize that $\aw{a,b}$, which is vector of size $2$, is the reverse of
$\aw{1-a,1-b}$. Therefore it suffices to analyze the case when $b\leq 1/2$. We
will take three cases:

\textbf{Case I,} $0\leq a \leq b$: Since the valuation of Bob is more than
Alice in the first good and the valuation of Alice is more in the second good,
the initial allocation is 
\begin{equation*}
\left[ \begin{array}{c c c} d^A_1 & d^B_1 \\
d^A_2 & d^B_2
\end{array}
\right ] = 
\left [
\begin{array}{c c c} 0 & 1 \\ 1 & 0 
\end{array} \right ],
\end{equation*}
 where Alice's gain is $1-a$ and
Bob's is $b$. Since Alice's gain is more, a
portion of the second good should be given to Bob. Solving the equations, the
final allocation would be,
\begin{equation*}
\left [
\begin{array}{c c c} 0 & 1 \\ \frac{1}{2-a-b} & \frac{1-a-b}{2-a-b}
\end{array} \right ],
\end{equation*}
It should be noted that in the case of $a=b$, there is no unique allocation,
since in that case the initial allocation is giving all the goods to Alice, but
we can start with either the first good to give to Bob or the second one,
therefore any of the following allocations is feasible,
\begin{equation*}
\left [
\begin{array}{c c c} 0 & 1 \\ \frac{1}{2-2b} & \frac{1-2b}{2-2b}
\end{array} \right ]
\qquad
\left [
\begin{array}{c c c} 1 & 0 \\ \frac{1-2b}{2-2b} & \frac{1}{2-2b}
\end{array} \right ],
\end{equation*}
which give us exactly the same gain. We note that we would get the second
allocation instead of the first if we took the case of $a=b$ in Case II
(discussed below).
Therefore the AW function is not well defined when the valuations are identical.

\textbf{Case II,} $b<a\leq 1-b$: the initial allocation is
\begin{equation*}
\left [
\begin{array}{c c c} 1 & 0 \\ 0 & 1
\end{array} \right ],
\end{equation*}
where Bob's gain is $1-b$ which is greater than that of Alice which is $a$,
therefore a portion of the second good should be given to Alice. Solving for
equality we get the following final allocation
\begin{equation*}
\left [
\begin{array}{c c c} 1 & 0 \\ \frac{1-a-b}{2-a-b} & \frac{1}{2-a-b}
\end{array} \right ].
\end{equation*}

\textbf{Case III,} $1-b<a\leq 1$: the initial allocation is
\begin{equation*}
\left [
\begin{array}{c c c} 1 & 0 \\ 0 & 1
\end{array} \right ],
\end{equation*}
where Alice's gain is $a$ which is greater than that of Bob which is $1-b$,
therefore a portion of the first good should be given to Bob. Solving for
equality, the final allocation would be
\begin{equation*}
\left [
\begin{array}{c c c} \frac{1}{a+b} & 1-\frac{1}{a+b} \\ 0 & 1
\end{array} \right ].
\end{equation*}

When $b>1/2$, by considering $\aw{1-a,1-b}$ and reversing the answer, we can
find the allocation in general:
\begin{equation}
 \aw{a,b} = \begin{cases}
                      \left (0 , \frac{1}{2-a-b}\right )
                      & a\leq b \wedge b\leq 1/2, \\
                     \left (1 , \frac{1-a-b}{2-a-b} \right )
                      & b<a\leq 1-b \wedge b\leq 1/2, \\
                     \left (\frac{1}{a+b} , 0 \right )
                     & 1-b<a\leq 1 \wedge b\leq 1/2, \\
		     \left (\frac{1}{a+b} , 0 \right )
                     & b\leq a \leq 1 \wedge b>1/2, \\
		      \left (1-\frac{1}{a+b} , 1 \right )
                     & 1-b\leq a <b \wedge b>1/2, \\
		      \left (0, \frac{1}{2-a-b} \right )
                     & a<1-b \wedge b>1/2. \\
            \end{cases}
\end{equation}
Taking the similar terms together and neglecting the cases when $a=b$ which is
not well defined as discussed before, we get the following simplified
formulation:
\begin{equation}
 \aw{a,b} = \begin{cases}
             \left(0,\frac{1}{2-a-b} \right ) & 0 \leq a \leq \min(1-b,b), \\
             \left(\frac{1}{a+b},0 \right ) & \max(1-b,b) \leq a \leq 1, \\
             \left(1,\frac{1-a-b}{2-a-b} \right ) & b < a < 1-b \wedge b\leq
1/2, \\
             \left(1-\frac{1}{a+b},1 \right ) & 1-b < a < b \wedge b> 1/2.
            \end{cases}
\end{equation}
Note that as discussed before, the special case when $a=b$ does not result in a
unique division, and we have taken one of the possible cases. However, as we
will see later,
the case of $a=b$ is not interesting for us, therefore this conflict is
acceptable for the purpose of our study.

An interesting fact is that, the four above cases are not independent. In fact
the two following
equalities (which are true, even when $m>2$) relate these four cases:
\begin{equation}
 \begin{split}
  \aw{1-(a,b)} &= \aw{a,b}^r, \\
  \aw{(a,b)^r} &= 1- \aw{a,b},
 \end{split}
\end{equation}
where the reverse operator acts as $(\alpha, \beta)^r = (\beta, \alpha)$. Note
that these are simply the case where the ordering of players or the placement
of items are altered.


\subsection{Some lemmas \label{sec:general-uniform-proofsNN}}
First we start by the following observation regarding the $\dawmaxhead$
function. The optimal value of $\ta$ when Alice
knows Bob's valuation has been analyzed formerly, a discussion could be found in
\cite{brams:cake:1996}.

\begin{prop}
  \label{prop:dawmax-observation}
 Assume $a$ and $b$ are fixed. Then
$\daw{\ta}{a}{b}$ is a concave function of $\ta$ if $\min(b,1-b) < \ta <
\max(b,1-b)$ and convex when $\ta<\min(b,1-b)$ or $\ta>\max(b,1-b)$. Also it is
increasing when $\ta<b$ and decreasing when $\ta>b$.
Furthermore,
\begin{equation}
\label{eq:prop-daw}
 \supremum_{\ta} \daw{\ta}{a}{b} = \begin{cases}
                                    \lim_{x\rightarrow b^{+}} \daw{x}{a}{b} &
a>b, \\
                                    \lim_{x\rightarrow b^{-}} \daw{x}{a}{b} &
a<b, \\
				    \daw{b}{b}{b} & a=b.
                                   \end{cases}
\end{equation}
\end{prop}

In fact this shows that if $a>b$, the optimal value of $\ta$ is $b^+ = b+
\epsilon$, and when $a<b$, the optimal value is $\ta=b^- = b-\epsilon$. In fact
in these two cases $\dawhead$ does not have a maximum. It should be noted that
 the AW function is not well defined when $\ta=b$ and
$a\neq b$. 

\begin{proof}
 First we give the exact formulation of $\dawhead$ using
\eqref{eq:aw-four-case} and Definition~\ref{def:daw},
\begin{equation}
 \daw{\ta}{a}{b} = \begin{cases}
                    \frac{1-a}{2-\ta-b} & 0 \leq \ta \leq \min(1-b,b), \\
		    \frac{a}{\ta+b} & \max(1-b,b) \leq \ta \leq 1, \\
		    a + (1-a) \frac{1-\ta-b}{2-\ta-b} & b < \ta < 1-b \wedge b
\leq 1/2, \\
		    1-\frac{a}{\ta+b} & 1-b < \ta < b \wedge b > 1/2.
                   \end{cases}
\end{equation}
First assume $b\leq 1/2$. In this case,
\begin{equation}
  \daw{\ta}{a}{b} = \begin{cases}
                     \frac{1-a}{2-\ta-b} & 0 \leq \ta \leq b, \\
		      a + (1-a) \frac{1-\ta-b}{2-\ta-b} & b < \ta < 1-b, \\
                     \frac{a}{\ta+b} & 1-b \leq \ta \leq 1,
                    \end{cases}
\end{equation}
which is increasing in $\ta\leq b$, decreasing in $b<\ta<1-b$ and $1-b\leq
\ta$, also the limit of the second case when $\ta$ goes to $1-b$ from left is
equal to $a$ which is equal to the value of the third case for $\ta = 1-b$.
Therefore the function is continuous everywhere expect possibly in $b$. The
left and right limits at $b$ are $(1-a)/(2-2b)$ and $(1+a-2b)/(2-2b)$
respectively. We see that the left limit is greater when $a<b$, they
are equal when $a=b$ and the right limit is greater when $b<a$, which shows
\eqref{eq:prop-daw} in this special case. The concave/convex statements are
evident from the expression. 

Now assume $b>1/2$, we have,
\begin{equation}
 \daw{\ta}{a}{b} = \begin{cases}
                    \frac{1-a}{2-\ta-b} & 0\leq \ta \leq 1-b, \\
                    1-\frac{a}{\ta+b} & 1-b<\ta<b, \\
                    \frac{a}{\ta+b} & b \leq \ta \leq 1,
                   \end{cases}
\end{equation}
which is increasing in $0\leq \ta \leq 1-b$ and $1-b<\ta <b$ and decreasing in
$b<\ta$. The limit of the second case and the value of the first case are both
equal to $1-a$ at $\ta = 1-b$, therefore the function is equal at that point.
The
left and right limits at $b$ are $1-a/2b$ and $a/2b$ respectively, therefore
left limit is greater when $b>a$, the right limit is greater when $b<a$ and
they are equal when $a=b$, which again verifies \eqref{eq:prop-daw}. Again, the
concave/convex statement are
evident from the expression. 
\end{proof}

Using this Proposition, we can conclude the following statement which justifies
Definition~\ref{def:dawmax}.

\begin{cor}
\label{cor:def-dawmax-justification}
 The optimum value for $\ta$ for $\daw{\ta}{a}{\bmin}{\bmax}$ falls in
$[\bmin,\bmax]$, i.e.\
\begin{equation*}
\max_{0\leq\ta\leq1}
\daw{\ta}{a}{\bmin,\bmax}
=\max_{\bmin\leq\ta\leq\bmax}
\daw{\ta}{a}{\bmin,\bmax}.
\end{equation*}
\end{cor}

\begin{proof}
 Assume $\ta \notin [\bmin,\bmax]$. First assume $\ta < \bmin$. As we have
shown in Proposition~\ref{prop:dawmax-observation}, $\daw{\ta}{a}{b}$ is
increasing in $[\ta,b)$ for all $b\in [\bmin,\bmax]$. Therefore,
\begin{equation}
 \begin{split}
  \daw{\ta}{a}{\bmin,\bmax} &=\frac{1}{\bmax - \bmin} \int_{\bmin}^{\bmax}
\daw{\ta}{a}{b} \mathrm{d} b \\
  &< \frac{1}{\bmax - \bmin} \int_{\bmin}^{\bmax}
\daw{\bmin}{a}{b} \mathrm{d} b \\
  &= \daw{\bmin}{a}{\bmin,\bmax},
 \end{split}
\end{equation}
hence the maximum can not happen at this $\ta$. The proof for the case where
$\ta > \bmax$
is similar using the fact that $\daw{\ta}{a}{b}$ is
decreasing in $(b,\ta]$ for all $b\in [\bmin,\bmax]$.
\end{proof}

We
expect that by asking a number of questions, the expected gain for
Alice increases, and the more questions she asks, the more is this improvement.
The following proposition states this.
\begin{prop}
\label{prop:delta-positive}
 Assume $a$ is fixed and $b$ is uniformly distributed in $[\bmin,\bmax]$. Then
the sequence $\dgamax{k}{\bmin,\bmax}$ for $k\geq 1$ is
nonnegative, nondecreasing and bounded by $1$, i.e.\
\begin{equation}
0\leq \dgamaxhead_1 \leq \dgamaxhead_2 \leq \dots \leq 1.
\end{equation}
\end{prop}

\subsection{Proof of Theorem~\ref{thm:ocz}}\label{secProofThem5AW}
First we prove some tools. In this special case when $1/2\geq \bmin$, the
integral in
\eqref{eq:psi-integral} could be computed
and the following properties could be easily derived by taking the first and
second derivatives. 

\begin{lem}
 If $1/2\leq\bmin$, then for $\bmin\leq\ta\leq\bmax$ we have,
\begin{equation}
 \daw{\ta}{a}{\bmin,\bmax} = 
\frac{a \log \left(\frac{4 \ta^2}{(\ta+\bmax)
(\ta+\bmin)}\right)-\ta+\bmax}{\bmax-\bmin},
\end{equation}
is concave in $\ta$, therefore it has a unique maximum. Furthermore if
$a\geq \tau_u(\bmin,\bmax)$ then the derivative is positive inside the
interval and therefore the maximum happens at $\bmax$ and if $a\leq
\tau_l(\bmin,\bmax)$ the derivative is negative inside the interval and
therefore the maximum happens at $\bmin$.
\end{lem}

\begin{proof}
 Using the expressions in \eqref{eq:prop-daw}, we have
\begin{equation}
 \begin{split}
  \daw{\ta}{a}{\bmin,\bmax} &= \frac{1}{\bmax-\bmin} \int_{\bmin}^{\bmax}
\daw{\ta}{a}{b} \mathrm{d} b \\
&= \frac{1}{\bmax-\bmin} \left ( \int_{\bmin}^{\ta} \frac{a}{\ta + x} \mathrm{d}
x + \int_{\ta}^{\bmax} \left(1- \frac{a}{\ta+x}\right) \mathrm{d} x \right ) \\
&= \frac{a \log \left(\frac{4 \ta^2}{(\ta+\bmax)
(\ta+\bmin)}\right)-\ta+\bmax}{\bmax-\bmin}.
 \end{split}
\end{equation}
Omitting the linear terms, we need to show that $\log \left(\frac{\ta^2}{(\ta +
\bmin)
(\ta+\bmax)}\right)$ is concave in $\ta$, the second derivative is equal to
\begin{equation}
\begin{gathered}
 \frac{1}{(\bmin+\ta)^2} + \frac{1}{(\bmax+\ta)^2} - \frac{2}{\ta^2} \\
= \left ( \frac{1}{(\bmin+\ta)^2} - \frac{1}{\ta^2} \right ) 
 + \left ( \frac{1}{(\bmax+\ta)^2} - \frac{1}{\ta^2} \right ) \leq 0,
\end{gathered}
\end{equation}
which shows the concavity.

Now assume that $a\geq \tau_u$. Since the function is concave, it suffices to
show that the derivative is positive at $\ta = \bmax$. The first derivative is
equal to,
\begin{equation}
\frac{ \frac{\ta(\ta + \bmax) - 2 a \bmax}{\ta(\bmin - \bmax)} -
\frac{a}{\ta+\bmin}}{\ta+\bmax}.
\end{equation}
Substituting $\ta = \bmax$,
\begin{equation}
 \frac{ -2\bmax(\bmax + \bmin) + a(\bmax+3\bmin)}{2\bmax(\bmax^2 - \bmin^2)}.
\end{equation}
Note that the denominator is positive since $\bmax>\bmin>0$, therefore
expression is greater than or equal to zero if and only if $a\geq \tau_u$. For
the second case, again since the function is concave, in order to show that the
maximum happens at $\bmin$, it suffices to check the derivative at $\ta = \bmin$
which is equal to,
\begin{equation}
 \frac{a(3\bmax + \bmin) - 2\bmax \bmin - 2\bmin^2}{2\bmin(\bmax^2 -\bmin^2)}.
\end{equation}
Again since the denominator is positive, the first derivative is less than or
equal to zero if and only if $a\leq \tau_l$.
\end{proof}

\begin{lem}
\label{lem:taus}
 (a) If $1/2\leq \bmin < \bmax \leq 1$, the thresholds $\tau_l$ and $\tau_u$
satisfy
\begin{equation}
 \tau_l \leq \bmin \leq \bmax \leq \tau_u,
\end{equation}
(b) If $[b_0,b_1]$ is a subinterval of $[\bmin,\bmax]$, i.e. $\bmin\leq
b_0$ and $b_1 \leq \bmax$, then
\begin{equation}
 \tau_l(\bmin,\bmax) \leq \tau_l(b_0,b_1) \leq \tau_u(b_0,b_1) \leq
\tau_u(\bmin,\bmax).
\end{equation}
\end{lem}

\begin{proof}
 If $s$ denotes the ratio of endpoints, $\bmax/\bmin$, we see that,
\begin{equation}
\label{eq:tau-s}
 \begin{split}
\tau_u &=  \bmax \frac{2(s+1)}{s+3}, \\
\tau_l  &= \bmin \frac{2(s+1)}{3s+1}.
\end{split}
\end{equation}

For part (a), note that $2(s+1)/(s+3) \geq 1$ for $s\geq 1$. Since
$s=\bmax/\bmin\geq 1$ and  $\tau_u \geq \bmax$. Similarly, for $s\geq 1$,
$2(s+1)/(3s+1) \leq 1$ which shows that $\tau_l \leq \bmin$.

For the second part, if $s'$ denotes $b1/b0$, we have $b1\leq \bmax$ and
$s'<s$, therefore \eqref{eq:tau-s} and the fact that the function
$2(s+1)/(s+3)$ is increasing show that $\tau_u(b0,b1) \leq
\tau_u(\bmin,\bmax)$. 
Similarly, since $\bmin\leq b_0$ and the function $2(s+1)/(3s+1)$ is
decreasing, $\tau_l(b0,b1) \geq
\tau_l(\bmin,\bmax)$
\end{proof}

The following Lemma gives a simple expression for $\dawmaxhead$ in this special
case.

\begin{lem}
\label{lem:ocz-dawmax}
 With conditions of Theorem~\ref{thm:ocz}, we have the following formulation
for $\dawmaxhead$,
\begin{equation}
 \dawmax{a}{\bmin,\bmax} = \begin{cases}
                            \frac{a \log \left(\frac{2
\bmax}{\bmax+\bmin}\right)}{\bmax-\bmin} & a\geq \tau_u, \\
			    \frac{a \log \left(\frac{2
\bmin}{\bmax+\bmin}\right)}{\bmax-\bmin}+1 & a\leq \tau_l.
                           \end{cases}
\end{equation}
\end{lem}

\begin{proof}
 As we have shown before, $\daw{\ta}{a}{\bmin,\bmax}$ is differentiable
and concave in $\ta$, therefore its maximum value either happens at endpoints
or could be obtained by setting its derivative equal to zero. However, since
the function is concave, the maximum happens at $\bmax$ if and only if the
derivative is nonnegative entirely in the interval, which reduces to the
condition that the derivative is nonnegative at $\bmax$. Simplifying this
condition, we realize that this happens when $a\geq \tau_u$, therefore
substituting $\ta = \bmax$ we get the expression for the first case. Using a
similar method and by setting the derivative at $\bmin$ to be less than or
equal to zero, we get the second case.
\end{proof}

In the next Lemma, we derive the exact form of $\dgamaxhead_k$.

\begin{lem}
\label{lem:geom-sec}
 For $a$ and $[\bmin,\bmax]$ fixed, the geometric sequence
$\hb_0,\dots,\hb_{2^k}$ where $\hb_0 = \bmin$, $\hb_{2^k} = \bmax$ and
\begin{equation}
 \log \hb_i = \left(1-\frac{i}{2^k}\right) \log \hb_0 + \frac{i}{2^k} \log
\hb_{2^k}
\quad 1<i<2^k,
\end{equation}
maximizes $\dgamax{k}{\bmin,\bmax}$.
\end{lem}
\begin{proof}
We will take two cases, $a\geq \tau_u$ or $a\leq \tau_l$. First assume that $a
\geq \tau_u$. Using Lemma~\ref{lem:ocz-dawmax} we have,
\begin{equation}
 \begin{split}
  \dga{k}{b_0,\dots,b_{2^k}} &= \sum_{i=1}^{2^k}
\frac{%
b_i - b_{i-1}
}{b_{2^k} - b_0}\dawmax{a}{b_{i-1},b_i} - \dawmax{a}{b_0,b_{2^k}} \\
  &= \sum_{i=1}^{2^k} \frac{a \log \left ( \frac{2b_i}{b_i + b_{i-1}} \right )
}{b_{2^k} - b_0} - \frac{a\log \left ( \frac{2b_{2^k}}{b_{2^k} + b_0} \right
)}{b_{2^k} - b_0},
 \end{split}
\end{equation}
where the last equality holds since $[b_{i-1},b_i]$ is a subinterval of
$[\bmin,\bmax]$ and hence using Lemma~\ref{lem:taus}, $a\geq
\tau_u(b_{i-1},b_0)$.
Note that $a, b_0$ and $b_{2^k}$ are constant, therefore by defining $s_i =
b_i/b_{i-1}$ we should maximize the following,
\begin{equation}
 \begin{split}
  \sum_{i=1}^{2^k} \log \left ( \frac{2b_i}{b_i + b_{i-1}} \right ) &=
\sum_{i=1}^{2^k} \log \left ( \frac{2}{1 + \frac{1}{s_i}} \right )= \log \prod_{i=1}^{2^k} \frac{2}{1+\frac{1}{s_i}}.
 \end{split}
\end{equation}
Since $\log$ is increasing, in order to maximize this, we need to minimize
$\alpha$ where,
\begin{equation}
 \alpha = \prod_{i=1}^{2^k} \left ( 1 + \frac{1}{s_i} \right).
\end{equation}
If we define $\hs_i = \hb_i / \hb_{i-1}$, since $\hb_i$ is a geometric
sequence,
\begin{equation}
 \hs_i = \left ( \frac{\bmax}{\bmin} \right ) ^{1/2^k}.
\end{equation}
Now define $\rho_i = \log s_i$ and $\hr_i = \log \hs_i$. Note that $\hr_i$ is a
constant sequence. In fact, since $\prod s_i =
\prod \hs_i = \bmax / \bmin$, $\sum \rho_i = \sum \hr_i = \log \bmax - \log
\bmin$. Also $\hb_i$ is geometric, hence for all $1\leq j \leq 2^k$
\begin{equation}
\label{eq:hr_i-sum-rho_i}
 \hr_j = \frac{\log \bmax - \log \bmin}{2^k} = \frac{\sum_{i=1}^{2^k}
\rho_i}{2^k}.
\end{equation}
Now, by defining $f(x) = \log ( 1+e^{-x})$ which is convex,
\begin{equation}
 \begin{split}
  \log \alpha &= \sum_{i=1}^{2^k} \log \left ( 1 + \frac{1}{s_i} \right ) = \sum_{i=1}^{2^k} \log \left ( 1+ \exp (-\rho_i)\right ) = \sum_{i=1}^{2^k} f(\rho_i) \\
  &\stackrel{(a)}{\geq} 2^k f\left ( 
\frac{\sum_{i=1}^{2^k}
\rho_i}{2^k} \right ) \stackrel{(b)}{=}   2^k f\left ( 
\frac{\sum_{i=1}^{2^k} \hr_i}{2^k} \right )\stackrel{(c)}{=} \sum_{i=1}^{2^k} f(\hr_i),
 \end{split}
\end{equation}
where $(a)$ uses Jensen's inequality and the fact that $f(x)$ is convex,
$(b)$ uses \eqref{eq:hr_i-sum-rho_i} and $(c)$ uses the fact that $\hr_i$ is a
constant sequence. Thus $\hb_i$ minimizes $\alpha$ or equivalently
maximizes $\dgahead_k$.

Now consider he case where $a \leq \tau_l$. In this case we have,
\begin{equation}
 \begin{split}
  \dga{k}{b_0,\dots,b_{2^k}} &= \sum_{i=1}^{2^k}
\frac{%
b_i - b_{i-1}
}{b_{2^k} - b_0}\dawmax{a}{b_{i-1},b_i}, - \dawmax{a}{b_0,b_{2^k}} \\
  &= \sum_{i=1}^{2^k} \frac{b_i - b_{i-1}}{b_{2^k} - b_0} \left ( \frac{a \log
\left ( \frac{2 b_{i-1}}{b_i + b_{i-1}} \right )}{b_i - b_{i-1} } +1 \right ) 
- \left ( \frac{a \log \left ( \frac{2b_0}{b_{2^k} + b_0} \right )}{b_{2^k} -
b_0} +1\right )\\
&= \sum_{i=1}^{2^k} \frac{a \log
\left ( \frac{2 b_{i-1}}{b_i + b_{i-1}} \right )}{b_{2^k} - b_0} 
- \frac{a \log \left ( \frac{2b_0}{b_{2^k} + b_0} \right )}{b_{2^k} - b_0},
 \end{split}
\end{equation}
where again we have used Lemma~\ref{lem:taus} which guarantees that $a\leq
\tau_l(b_{i-1},b_i)$.
Omitting the constant terms, we should maximize
\begin{equation}
 \sum \log \frac{2b_{i-1}}{b_i + b_{i-1}} = \log \prod \frac{2}{1+s_i}.
\end{equation}
Since $\log$ is increasing, we should minimize $\beta = \prod_{i=1}^{2^k}
1+s_i$. By defining $g(x) = f(-x) = \log \left ( 1+ e^x\right ) $ which is
convex, we have
\begin{equation}
 \begin{split}
  \log \beta &= \sum_{i=1}^{2^k} \log \left ( 1 + s_i \right ) = \sum_{i=1}^{2^k} \log \left ( 1+ \exp (\rho_i)\right ) = \sum_{i=1}^{2^k} g(\rho_i) \\
  &\stackrel{(a)}{\geq} 2^k g\left ( 
\frac{\sum_{i=1}^{2^k}
\rho_i}{2^k} \right ) \stackrel{(b)}{=}   2^k g\left ( 
\frac{\sum_{i=1}^{2^k} \hr_i}{2^k} \right )\stackrel{(c)}{=} \sum_{i=1}^{2^k} g(\hr_i),
 \end{split}
\end{equation}
where $(a)$ uses Jensen's inequality and convexity of $g$, $(b)$ uses
\eqref{eq:hr_i-sum-rho_i} and $(c)$ uses the fact that $\hr_i$ is a
constant sequence. Thus $\hb_i$ minimizes $\beta$ or equivalently
maximizes $\dgahead_k$.
\end{proof}

\begin{remark}
 Note that this lemma shows that the optimal series of divisions for $k$
questions is exactly the same for that of $k-1$ questions together with the
optimal dividing question for each of the $2^{k-1}$ subintervals.
\end{remark}

\begin{lem}
\label{lem:interval-concave}
 If $a \notin [\tau_l(\bmin,\bmax),\tau_u(\bmin,\bmax)]$, then $\dgamaxhead_1$
is interval concave.
\end{lem}

\begin{proof}
Assume $b_0\geq \bmin$ and $b_2 \leq \bmax$, as a result of
Lemma~\ref{lem:taus} part (b), $a\notin [\tau_l(b_0,b_2),\tau_u(b_0,b_2)]$.
 We can derive the formulation for
$\dgahead_2$. Using Lemma~\ref{lem:ocz-dawmax}, for the case of $a\leq \tau_l$:
\begin{equation}
 \begin{split}
  \dga{1}{b_0,b_1,b_2} &= \sum_{i=1}^{2}
\frac{%
b_i - b_{i-1}
}{b_{2} - b_0}\dawmax{a}{b_{i-1},b_i} - \dawmax{a}{b_0,b_2} \\
&= \sum_{i=1}^{2} \frac{a \log
\left ( \frac{2 b_{i-1}}{b_i + b_{i-1}} \right )}{b_{2} - b_0} 
- \frac{a \log \left ( \frac{2b_0}{b_{2} + b_0} \right )}{b_{2} - b_0} \\
&= \frac{a \log \left(\frac{2 b_1
(b_2+b_0)}{(b_1+b_2)
(b_1+b_0)}\right)}{b_2-b_0}.
 \end{split}
\end{equation}
Similarly, for the case of $a\geq \tau_u$:
\begin{equation}
\label{eq:dga-012-second}
 \begin{split}
  \dga{1}{b_0,b_1,b_2} &= \sum_{i=1}^{2}
\frac{%
b_i - b_{i-1}
}{b_{2} - b_0}\dawmax{a}{b_{i-1},b_i} - \dawmax{a}{b_0,b_{2}} \\
  &= \sum_{i=1}^{2} \frac{a \log \left ( \frac{2b_i}{b_i + b_{i-1}} \right )
}{b_{2} - b_0} - \frac{a\log \left ( \frac{2b_{2}}{b_{2} + b_0} \right
)}{b_{2} - b_0} \\
&= \frac{a \log \left(\frac{2 b_1
(b_2+b_0)}{(b_1+b_2)
(b_1+b_0)}\right)}{b_2-b_0}.
 \end{split}
\end{equation}
We observe that $\dga{1}{b_0,b_1,b_2}$ is the same in the two cases. Using
Lemma~\ref{lem:geom-sec} and substituting $b_1 = \sqrt{b_0
b_2}$,
\begin{equation}
 \dgamax{1}{b_0,b_2} =   \frac{a \log \left(\frac{2
(b_2+b_0)}{\left(\sqrt{b_2}+\sqrt{b_0}\right)^2}
\right)}{b_2-b_0}.
\end{equation}
By defining $s = \sqrt{b_2 / b_0}$, the ratio of interval endpoints, we can
rewrite
$\dgamaxhead_2$ in the following form,
\begin{equation}
\label{eq:delta1-simplified-form}
 \dgamax{1}{b_0,b_2} = \frac{a}{b_0} \frac{\log
\left(\frac{\left(s-1\right)^2}{\left(s+1\right)^2}+1\right)}{s^2-1}
.
\end{equation}

 We show that if $b_0\leq b_1 \leq b_2$, then $\dgamax{1}{b_0,b_1} \leq
\dgamax{1}{b_0,b_2}$ and $\dgamax{1}{b_1,b_2} \leq
\dgamax{1}{b_0,b_2}$ which is sufficient for a function to be interval
concave. Note that the first term in \eqref{eq:delta1-simplified-form},
$a/b_0$, is decreasing in $b_0$, thus it suffices to show that $f(s)$
defined as,
\begin{equation}
\label{eq:AW-fs}
f(s)= \frac{\log
\left(\frac{\left(s-1\right)^2}{\left(s+1\right)^2}+1\right)}{s^2-1}
,
\end{equation}
is increasing in $s$ when $1\leq s \leq \sqrt{2}$ ($s$ is the square root
of the ratio of
$\bmax$ and $\bmin$ and hence is greater than $1$, also $\bmax/\bmax$ is at
most $2$, since $\bmin\geq 1/2$ and $\bmax\leq 1$). Monotonicity of $f$ could
be shown analytically. Its plot is provided in Figure~\ref{fig:ocz-f(s)}.
\end{proof}

\begin{figure}
 \centering
  \includegraphics[width=0.5\textwidth]{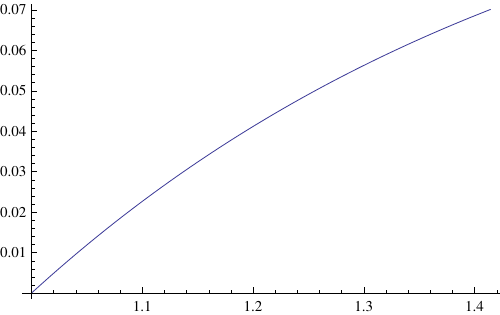}
  \caption{The plot of $f(s)$ as defined in \eqref{eq:AW-fs} for $1\leq s\leq
\sqrt{2}$\label{fig:ocz-f(s)}}
\end{figure}


Now we have sufficient tools to prove Theorem~\ref{thm:ocz}.

\begin{proof}[Proof of Theorem~\ref{thm:ocz}]
 For the sake of simplicity, we use $\dgamaxhead_k$ to denote
$\dgamax{k}{\bmin,\bmax}$. For proving the concavity of the sequence, it
suffices to
prove that for $k\geq 2$,
\begin{equation}
 \dgamaxhead_k - \dgamaxhead_{k-1} \geq \dgamaxhead_{k-1} - \dgamaxhead_{k-2},
\end{equation}
where $\dgamaxhead_0$ is defined to be $0$. Assume $b_0,\dots, b_{2^k}$ is the
sequence given by Lemma~\ref{lem:geom-sec} which maximize $\dgahead_k$.
Note that since the sequence is geometric, the sequence $b_i, 0\leq i \leq
2^k,i \equiv 0 \pmod{2}$ is the sequence maximizing $\dgahead_{k-1}$ and also
the sequence $b_i, 0\leq i \leq
2^k,i \equiv 0 \pmod{4}$ maximizes $\dgahead_{k-2}$. Hence,
\begin{equation}
 \begin{split}
  \dgamaxhead_k &= \sum_{i=1}^{2^k} \frac{%
b_i - b_{i-1}
}{b_{2^k} - b_0} \dawmax{a}{b_{i-1},b_i} -
\dawmax{a}{b_0, b_{2^k}}, \\
  \dgamaxhead_{k-1} &= \sum_{\stackrel{i=1}{2|i}}^{2^k}
\frac{%
b_i - b_{i-2}
}{b_{2^k} - b_0}\dawmax{a}{b_{i-2},b_i} -
\dawmax{a}{b_0, b_{2^k}}, \\
\dgamaxhead_{k-2} &= \sum_{\stackrel{i=1}{4|i}}^{2^k}
\frac{%
b_i - b_{i-4}
}{b_{2^k} - b_0} \dawmax{a}{b_{i-4},b_i} -
\dawmax{a}{b_0, b_{2^k}}.
 \end{split}
\end{equation}
Note that when $k=2$ the last equality reduces to
$\dgamaxhead_0 = 0$ which is consistent with our definition.
Subtracting $\dgamaxhead_{k-1}$ from $\dgamaxhead_k$ and simplifying,
\begin{equation}
 \begin{split}
  \dgamaxhead_k - \dgamaxhead_{k-1} &= \sum_{\stackrel{i=1}{2|i}}^{2^k}
\frac{b_i - b_{i-2}}{b_{2^k} - b_0} \dga{1}{b_{i-2},b_{i-1},b_i} \\
  &= \sum_{\stackrel{i=1}{2|i}}^{2^k} \frac{b_i - b_{i-2}}{b_{2^k} - b_0}
\dgamax{1}{b_{i-2},b_i},
 \end{split}
\end{equation}
where we have used the fact that since $b_i$ is geometric, $b_{i-1} = \sqrt{b_i
b_{i-2}}$. Similarly,
\begin{equation}
 \dgamaxhead_{k-1} - \dgamaxhead_{k-2} = \sum_{\stackrel{i=1}{4|i}}^{2^k}
\frac{b_i - b_{i-4}}{b_{2^k} - b_0} \dgamax{1}{b_{i-4},b_i}.
\end{equation}
Now,
\begin{equation}
 \begin{split}
  \dgamaxhead_k - \dgamaxhead_{k-1} &=  \sum_{\stackrel{i=1}{2|i}}^{2^k}
\frac{b_i - b_{i-2}}{b_{2^k} -
b_0} \dgamax{1}{b_{i-2},b_i} \\
  &= \sum_{\stackrel{i=1}{4|i}}^{2^k} \frac{b_{i-2} - b_{i-4}}{b_{2^k} - b_0}
\dgamax{1}{b_{i-4},b_{i-2}} + \frac{b_i - b_{i-2}}{b_{2^k} - b_0}
\dgamax{1}{b_{i-2},b_i} \\
&= \sum_{\stackrel{i=1}{4|i}}^{2^k} \frac{b_i - b_{i-4}}{b_{2^k} - b_0}
\left ( \frac{b_{i-2} - b_{i-4}}{b_i - b_{i-4}}\dgamax{1}{b_{i-4},b_{i-2}} +
\frac{b_i - b_{i-2}}{b_i - b_{i-4}} \dgamax{1}{b_{i-2},b_i} \right ) \\
  &\leq \sum_{\stackrel{i=1}{4|i}}^{2^k} \frac{b_i - b_{i-4}}{b_{2^k} - b_0}
\dgamax{1}{b_{i-4},b_i} = \dgamaxhead_{k-1} - \dgamaxhead_{k-2},
 \end{split}
\end{equation}
where we have used the fact from Lemma~\ref{lem:interval-concave} that
$\dgamaxhead_1$ is interval concave.
\end{proof}


\subsection{Proof of Theorem~\ref{thm:upper-bound-general}
} \label{proofsAW2}
\begin{proof}[Proof of Theorem~\ref{thm:upper-bound-general}]
 We prove this by induction. In fact we prove a stronger statement; we claim
that for all $k\geq 1$ and $b_0, b_{2^k}$ such that $\bmin\leq b_0 <
b_{2^k} \leq \bmax$,
\begin{equation}
 \dgamax{k}{b_0,b_{2^k}} \leq k \tD(b_0, b_{2^k}),
\end{equation}
which reduces to what we expect by substituting $b_0 = \bmin$ and $b_{2^k} =
\bmax$. Note that for $k=1$ this reduces to
\eqref{eq:thm:upper-bound-general-k=1} which is assumed to be true. Now assume
it is true for $k-1$. If $b_1,\dots,b_{2^k-1}$ are the divisions which
maximize $\dgamax{k}{b_0,b_{2^k}}$, we have:
\begin{equation}
 \begin{split}
  \dgamax{k}{b_0,b_{2^k}} &= \dga{k}{b_0,b_1,\dots,b_{2^k}} \\
  &=\sum_{i=1}^{2^k}
\frac{b_i-b_{i-1}}{b_{2^k}-b_0}
\dawmax{a}{b_{i-1},b_i} - \dawmax{a}{\bmin,\bmax}\\
&= \frac{b_{2^{k-1}} - b_0}{b_{2^k} - b_0} \dga{k-1}{b_0,\dots,b_{2^{k-1}}} +
\frac{b_{2^k} - b_{2^{k-1}}}{b_{2^k} - b_0}
\dga{k-1}{b_{2^{k-1}},\dots,b_{2^k}} \\
 &\quad + \frac{b_{2^{k-1}} - b_0}{b_{2^k} - b_0} \dawmax{a}{b_0,b_{2^{k-1}}}
+ \frac{b_{2^k} - b_{2^{k-1}}}{b_{2^k} - b_0} \dawmax{a}{b_{2^{k-1}}, b_{2^k}}\\
&\quad- \dawmax{a}{b_0,b_{2^k}} \\
&= \frac{b_{2^{k-1}} - b_0}{b_{2^k} - b_0} \dga{k-1}{b_0,\dots,b_{2^{k-1}}} +
\frac{b_{2^k} - b_{2^{k-1}}}{b_{2^k} - b_0} \dga{k-1}{b_{2^{k-1}},
\dots,b_{2^k}} \\
&\quad + \dga{1}{b_0,b_{2^{k-1}}, b_{2^k}}\\
&\leq \frac{b_{2^{k-1}} - b_0}{b_{2^k} - b_0} \dgamax{k-1}{b_0,b_{2^{k-1}}} +
\frac{b_{2^k} - b_{2^{k-1}}}{b_{2^k} - b_0} \dgamax{k-1}{b_{2^{k-1}}, b_{2^k}}
\\
&\quad + \dgamax{1}{b_0,b_{2^k}}.
 \end{split}
\end{equation}
Now by using the induction hypothesis
\begin{equation}
 \begin{split}
  \dgamax{k}{b_0,b_{2^k}} &\leq (k-1) \left ( \frac{b_{2^{k-1}} - b_0}{b_{2^k} -
b_0} \tD(b_0,b_{2^{k-1}}) + \frac{b_{2^k} - b_{2^{k-1}}}{b_{2^k} - b_0}
\tD(b_{2^{k-1}},b_{2^k}) \right ) \\
&\quad + \tD(b_0,b_{2^k}).
 \end{split}
\end{equation}
Since $\tD$ is interval concave, we have
$
 \dgamax{k}{b_0,b_{2^k}} \leq k \tD(b_0,b_{2^k}). \qedhere
$
\end{proof}

\subsection{Proof of Theorem  \ref{thm:general-bound-differentiableNew}}\label{secProofThem5AW2}
Using Theorem~\ref{thm:upper-bound-general}, it suffices to prove the following theorem:
\begin{thm}
\label{thm:general-bound-differentiable}
 Assume that for $\bmin\leq x < y \leq \bmax$, $\dgamax{1}{x,y}$ is
differentiable with respect to $y$. Then $\tD(x,y)$ defined by 
\begin{equation}
 \tD(x,y) = \max_{x\leq \gamma_1 \leq \gamma_2 \leq y} \Lambda(\gamma_1,\gamma_2),
\end{equation}
where $ \Lambda(x,y)=\frac{\partial}{\partial
y} \Gamma(x,y)$ defined as follows
\begin{equation}
 \Gamma(x,y) = \begin{cases}
                (y-x) \dgamax{1}{x,y} & y>x, \\
                0 & y=x,
               \end{cases}
\end{equation}
is an interval concave upper bound for $\dgamaxhead_1$.
\end{thm}

\begin{proof}[Proof of Theorem~\ref{thm:general-bound-differentiable}]
 We know from Proposition~\ref{prop:delta-positive} that $\dgamaxhead_1$ is
bounded, therefore $\Gamma$ is continuous at $x=y$. Furthermore, since
$\dgamaxhead_1$ is
differentiable with respect to $y$, for a fixed $x$, it is continuous with
respect to $y$. Therefore for $b>a$, $\Gamma(a,y)$ is continuous when $y$
changes in $[a,b]$ and differentiable in $(a,b)$ as $\dgamaxhead_1$ is. Using
mean value theorem, there exists a $y^\ast \in (a,b)$ where
\begin{equation}
 \Gamma(a,b) = (b-a) \frac{\partial}{\partial y} \Gamma(a,y^\ast).
\end{equation}
Now by the definition of $\tD$,
\begin{equation}
 (b-a) \dgamax{1}{a,b} = \Gamma(a,b) = (b-a) \frac{\partial}{\partial y}
\Gamma(a,y^\ast) \leq (b-a) \tD(a,b),
\end{equation}
hence for any $a<b$, $\dgamax{1}{a,b} \leq \tD(a,b)$ and therefore $\tD$ is an
upper bound on $\dgamaxhead_1$. 

It only remains to prove that it is interval
concave. Note that if $x\leq t\leq y$, then
\begin{equation}
 \tD(x,t) = \max_{x\leq \gamma_1 \leq \gamma_2 \leq t}\frac{\partial}{\partial
y} \Gamma(\gamma_1,\gamma_2) \leq \max_{x\leq \gamma_1 \leq \gamma_2 \leq
y}\frac{\partial}{\partial y}
\Gamma(\gamma_1,\gamma_2) = \tD(x,y),
\end{equation}
likewise,
\begin{equation}
 \tD(t,y) = \max_{t\leq \gamma_1 \leq \gamma_2 \leq y}\frac{\partial}{\partial
y} \Gamma(\gamma_1,\gamma_2) \leq \max_{x\leq \gamma_1 \leq \gamma_2 \leq
y}\frac{\partial}{\partial y}
\Gamma(\gamma_1,\gamma_2) = \tD(x,y),
\end{equation}
therefore,
\begin{equation}
\frac{t-x}{y-x} \tD(x,t) + \frac{y-t}{y-x} \tD(t,y) \leq \tD(x,y),
\end{equation}
which shows that $\tD$ is interval concave.
\end{proof}






\end{document}